\documentclass[12pt]{article}

\usepackage{graphicx}
\usepackage{amsmath}
\usepackage{amssymb}
\usepackage{amsthm}
\usepackage{url}
\usepackage{thm-restate}

\usepackage[margin=1in]{geometry}

\usepackage[colorlinks = true]{hyperref}

\hypersetup{
  pdftitle = {Gapped and gapless phases of frustration-free spin-1/2 chains},
  pdfauthor = {Sergey Bravyi, David Gosset}
}

\usepackage{xcolor}
\definecolor{darkred}  {rgb}{0.5,0,0}
\definecolor{darkblue} {rgb}{0,0,0.5}
\definecolor{darkgreen}{rgb}{0,0.5,0}
\hypersetup{
  urlcolor   = blue,         % color of external links
  linkcolor  = darkblue,     % color of internal links
  citecolor  = darkgreen,    % color of links to bibliography
  filecolor  = darkred       % color of file links
}

\newcommand{\be}{\begin{equation}}
\newcommand{\ee}{\end{equation}}
\newcommand{\ba}{\begin{array}}
\newcommand{\ea}{\end{array}}
\newcommand{\bea}{\begin{eqnarray}}
\newcommand{\eea}{\end{eqnarray}}

\newcommand{\calD}{{\cal D }}
\newcommand{\calH}{{\cal H }}
\newcommand{\calL}{{\cal L }}

\newcommand{\calG}{{\cal G }}
\newcommand{\calV}{{\cal V }}

\newcommand{\calC}{{\cal C }}

\newcommand{\calK}{{\cal K }}

\newcommand{\EE}{\mathbb{E}}

\newcommand{\CC}{\mathbb{C}}

\newcommand{\RR}{\mathbb{R}}

\newcommand{\expect}[1]{\EE{\left[ {#1}  \right]}}

\newcommand{\trace}[1]{{\mathrm{Tr}{#1}}}

\newcommand{\pff}[1]{\mathrm{Pf}\left({#1}\right)}

\newcommand{\herm}[1]{\mathrm{Herm}{{\left(#1\right)}}}

\newcommand{\sn}[2]{{\mathrm{sn}{\left( {#1} ; {#2} \right)}}}
\newcommand{\cn}[2]{{\mathrm{cn}{\left( {#1} ; {#2} \right)}}}

\newtheorem{dfn}{Definition}
\newtheorem{prop}{Proposition}
\newtheorem{lemma}{Lemma}

\newtheorem{corol}{Corollary}
\newtheorem{fact}{Fact}

\newtheorem{theorem}{Theorem}
\newtheorem*{theorem*}{Theorem}

\begin{document}

\title{Complexity of quantum impurity problems}
\author{Sergey Bravyi and David Gosset\\ \textit{IBM T.J. Watson Research Center}}

\date{}
\maketitle

\abstract{We give a quasi-polynomial time classical algorithm for estimating the ground state energy and for computing low energy states of quantum impurity models. Such models describe a bath of free fermions 
coupled to  a small  interacting subsystem called an impurity. The full system  consists of $n$ fermionic modes and has a Hamiltonian $H=H_0+H_{imp}$,
where $H_0$ is  quadratic in creation-annihilation operators and 
$H_{imp}$ is an arbitrary Hamiltonian acting on a subset of $O(1)$ modes. 
We show that the ground energy of $H$ can be approximated with an additive error $2^{-b}$
 in time $n^3 \exp{[O(b^3)]}$.
Our algorithm also finds a low energy state  that achieves this approximation.
The low energy state is
represented as a superposition of $\exp{[O(b^3)]}$ fermionic Gaussian states.
To arrive at this result we prove several theorems 
concerning exact ground states of impurity models.
In particular, we show that eigenvalues of the ground state covariance matrix  decay
exponentially 
with the exponent depending very mildly on the spectral gap of $H_0$.
A key ingredient of our proof is  Zolotarev's rational approximation to the $\sqrt{x}$ function. We anticipate that our algorithms may be used
in hybrid quantum-classical simulations of strongly correlated materials
based on dynamical mean field theory.
We implemented a simplified practical version of our algorithm and benchmarked it using the
single impurity Anderson model. 
 }

\newpage

\tableofcontents

\newpage

%%%%%%%%%%%%%%%%%%%%%%%%%%%%%%%%%%%%%%%%%%%%%%%%%%%%%%%%%%%%%%%%%%
%%%%%%%%%%%%%%%%%%%%%%%%%%%%%%%%%%%%%%%%%%%%%%%%%%%%%%%%%%%%%%%%%%

\section{Introduction}

In this paper we study ground states and low energy states of quantum impurity models. Such models describe a bath of free fermions 
coupled to a small  interacting subsystem called an \textit{impurity}. 

Hamiltonians of this form were famously studied in the 1960s and 70s by Anderson, Kondo, Wilson, and many others to investigate the physics of a magnetic impurity embedded in a metal \cite{Anderson61,Kondo64,Wilson75}.  This line of research elucidated the theoretical mechanism of the Kondo effect\footnote{While it is generally expected that a metal should become a better conductor as temperature is reduced,  for some metals with dilute impurities the resistivity achieves a minimum value at a nonzero temperature.  In short, the resistivity can increase as temperature is lowered due to scattering of conduction electrons in the metal with a localized electron in the impurity (see, e.g., Ref.~\cite{K01}).} which had been observed experimentally decades earlier \cite{deHaas36}. It also spurred the development of Wilson's numerical renormalization group \cite{Wilson75}, a non-perturbative numerical method which reproduces the low temperature physics of these systems.
 
The study of quantum impurity models extends beyond this direct application and
provides a powerful numerical method for calculating 
electronic structure of  strongly correlated materials such as  transition metal compounds and high-temperature
 superconductors~\cite{GeorgesDMFT}. These materials
are described by fermionic lattice models with interactions throughout the system, instead of localized within a small
subsystem.  Nevertheless, impurity models can be used to study such materials
within an approximation known as dynamical mean field theory (DMFT)~\cite{kotliar2006electronic}. 
 The impurity is typically chosen to model a group of atoms contained within a unit cell of the lattice, whereas the bath models the bulk of the material.  

Recently there has been growing interest in solving impurity problems in the quantum information community. It was suggested by Bauer et al. \cite{Bauer15} that a small quantum computer with a few hundred qubits
can potentially speed up certain steps in material simulations based on the DMFT method.
In particular, Ref.~\cite{Bauer15} proposed a quantum algorithm for computing 
the Green's function of a quantum impurity model.  Kreula et al. \cite{Kreula16} subsequently proposed a proof-of-principle demonstration of this  algorithm. 
 
Quantum impurity problems are also interesting from the standpoint of Hamiltonian complexity theory~\cite{QHCreview}.
In general, estimating the ground energy of a quantum many-body system composed of  
spins or fermi modes with local interactions is a hard problem. Formally, this problem is complete for the complexity class QMA -- a quantum analogue of 
NP~\cite{KitaevBook,schuch2009computational}. The QMA-completeness result implies that, in the worst case, the ground energy of interacting fermi systems cannot be computed efficiently (assuming QMA$\ne$BQP). In contrast, Hamiltonians describing free fermions are  exactly solvable and their
ground energy can be computed in polynomial time. Quantum impurity models provide a natural arena for studying the complexity of fermionic systems
in an intermediate regime interpolating between the free and the fully interacting 
 cases.

In this paper we describe a classical algorithm for approximating the ground energy and for computing low energy states of quantum impurity models.  We focus  on the worst-case computational complexity
of this problem as a function of the system size  and the desired approximation error.
We also prove several  theorems concerning  exact ground states
of impurity models that appear to be new. 

\subsection{Main results}

To state our results let us first define a general quantum impurity model. We consider the $2^n$-dimensional Hilbert space $\mathcal{H}_n$ of $n$ fermi modes, spanned by Fock basis vectors
\[
|x_1,x_2,\ldots, x_n\rangle=(a_1^{\dagger})^{x_1} (a_2^{\dagger})^{x_2}\ldots (a_n^{\dagger})^{x_n}|0^n\rangle.
\]
Here $a^{\dagger}_j,a_j$ are fermionic creation and annihilation operators, $x_j\in \{0,1\}$ is the occupation number
 of the $j$th mode and $|0^n\rangle$ is the vacuum state which satisfies $a_j|0^n\rangle=0$ for all $j$. 
A quantum impurity model is a Hamiltonian $H$ which acts on $\mathcal{H}_n$ as
\[
H=H_0+H_{imp},
\]
where $H_{0}$ describes a bath of free fermions and $H_{imp}$ describes the impurity.
To specify the form of $H_0$ and $H_{imp}$ it will be convenient to use Majorana operators $c_1,c_2,\ldots,c_{2n}$ defined by
\begin{equation}
\label{Majorana}
c_{2j-1}=a_j+a_j^{\dagger} \quad \text{and} \quad c_{2j}=-i(a_j-a_j^{\dagger}).
\end{equation}
The Majorana operators are hermitian and satisfy
\begin{equation}
\label{MCR}
c_p c_q+c_qc_p=2\delta_{p,q}I
\end{equation}
for all $1\leq p,q \leq 2n$.   We may write any free fermion Hamiltonian $H_0$ as
\begin{equation}
\label{bath1}
H_{0}=e_0 I + \frac{i}4\sum_{p,q=1}^{2n} h_{p,q} c_p c_q
\end{equation}
where $h$ is real anti-symmetric matrix  and $e_0=\| h\|_1/4$ is an energy shift chosen such that $H_0$ has zero ground energy.
Here and in the following we use $\|\cdot\|_1$ to denote the trace norm, while $\|\cdot\|$ denotes the operator norm. 
Let us choose the energy scale such that $\|h\| \le 1$. Then single-particle excitation energies of $H_{0}$
 belong to the interval $[0,1]$ while eigenvalues of $H_0$ belong to the interval $[0,n]$.

Let us agree that $H_{imp}$ acts non-trivially only on the Majorana modes $c_1,\ldots,c_m$.
Here $m$ is the impurity size.
We shall be interested in the case $m\ll n$. The impurity Hamiltonian must 
include only even-weight  Majorana operators (i.e., fermionic parity is conserved) but otherwise can be completely 
arbitrary. We may write the impurity Hamiltonian as 
\[
H_{imp}=\sum_{\substack{x\in \{0,1\}^m\\ |x|=0\; \mathrm{mod} \; 2}} g_x  c_1^{x_1}c_2^{x_2}\ldots c_m^{x_m}
\]
where $g_x$ are some coefficients\footnote{In order for $H_{imp}$ to be hermitian,
the coefficients $g_x$ must be real for $|x|=0{\pmod 4}$ and imaginary 
for $|x|=2{\pmod 4}$.}.
We do not impose any restrictions on the magnitude of these coefficients or on the norm of $H_{imp}$.
Let us write
\[
e_g=\min_{\phi\in \mathcal{H}_n} \frac{\langle \phi|H|\phi\rangle}{\langle \phi |\phi\rangle}
\]
for the ground energy of the full Hamiltonian.  Our main result is as follows.

\begin{restatable}[\textbf{Quasipolynomial algorithm}]{theorem}{quasi}
There exists a classical algorithm which takes as input a quantum impurity model $H$,  a target precision 
$\gamma\in (0,1/2]$, and outputs an estimate $E$ such that $|E-e_g|\le \gamma$. 
The algorithm has runtime
\[
O(n^3) \exp{\left[ O(m\log^{3}(m\gamma^{-1})\right]}.
\]
\label{thm:1}
\end{restatable}
For a fixed impurity size $m$ the runtime is polynomial  in $n$
and  quasi-polynomial   in $\gamma^{-1}$.
We are not aware of any obstacles to achieving  a polynomial scaling in $\gamma^{-1}$
and leave this as as open problem. On the other hand, the dependence 
of the runtime  on $n$ and $m$ is nearly optimal.
Indeed, choosing $H_0=0$ and $H_{imp}=0$ reduces the problem  to approximating the ground
energy of $H_{imp}$ and $H_{0}$ respectively. In the worst case this requires time
$2^{\Omega(m)}$ and $\Omega(n^3)$ respectively (using existing methods).
As far as we know, our algorithm is the first proposed method which gives
a rigorous bound on the approximation error for general impurity models.

The proof of Theorem~\ref{thm:1} including a complete description of the algorithm is given in 
Section~\ref{sec:generalcase}. At a high level the algorithm proceeds as follows.
We introduce a deformed impurity problem in which the single-particle energies of the bath Hamiltonian are approximated by a set of equally spaced grid points\footnote{We note that a 
discretization of the bath Hamiltonian is also used in the numerical renormalization group method~\cite{Wilson75}.}. 
 This deformed impurity problem has a special feature that a large number of fermionic modes can be decoupled from the impurity by a (Gaussian) unitary transformation. We show that the full Hamiltonian $H$
has a low-energy state within a subspace $\calV$ spanned by eigenstates of the deformed bath Hamiltonian such that
(a) bath modes coupled to the impurity have at most  $O(m\log^2{(m\gamma^{-1})})$ excitations 
and (b) all decoupled modes are unoccupied. We show that the dimension of $\calV$ is upper bounded by 
$\exp{\left[ O(m\log^{3}(m\gamma^{-1})\right]}$. 
Note that the dimension of $\calV$ has no dependence on $n$.
We approximate the ground energy $e_g$ by restricting the deformed impurity model onto the subspace $\calV$ and using exact diagonalization to compute the
smallest eigenvalue. The corresponding smallest eigenvector
$\psi \in \calV$ can be written as a superposition of at most $\dim{(\calV)}$ 
fermionic Gaussian states. 

The analysis of our algorithm relies on new results concerning ground states of quantum impurity models; in particular, Theorem~\ref{thm:2} (described below) and its corollaries. 
A simplified practical version of our algorithm is described in Section~\ref{sec:variational} and benchmarked using the single impurity Anderson model~\cite{Anderson61}.

{\em Remark:} The runtime quoted in theorem \ref{thm:1} counts the total number of elementary algebraic operations $+, \times,  /, \sqrt{\cdot}$. Moreover, for the sake of readability,  throughout the paper 
 we ignore errors
 incurred in
the standard linear algebra subroutines. In particular, we assume 
that  eigenvalues  of a hermitian $N\times N$ matrix $A$
can be computed exactly in time $O(N^3)$. 
Strictly speaking, the cost of this computation has a mild dependence
on the desired precision  and the norm of $A$. 
Applying Householder transformations to make
$A$  tri-diagonal and using rigorous eigenvalue
algorithms for tri-diagonal matrices~\cite{bini1998computing}
one can estimate all eigenvalues of $A$ with an additive error $\delta$
in time $O(N^3) poly(\log{(N)}, \log{(\delta^{-1})},\log{(\|A\|)})$,
see Theorem~7.1 of Ref.~\cite{bini1998computing}.
Taking into account this overhead
would alter the asymptotic runtime stated in Theorem~\ref{thm:1} 
by a factor $poly(\log{(n)},\log{(\| H_{imp}\|)})$.

Let us now discuss ground states of quantum impurity models and their features. Since impurity models are usually not exactly solvable\footnote{In certain cases one can compute the ground energy of  quantum impurity models
exactly in the thermodynamic
limit  $n\to \infty$ using the  Bethe Ansatz method~\cite{wiegmann1983exact,kawakami1981exact}.
This method is applicable only if the couplings between the bath and the impurity
have a certain special symmetry and the bath has a linear dispersion law.},
 their ground
states lack an analytic expression. 
Moreover, even if an analytic expression could be found, it would likely depend on subtle
details of the impurity Hamiltonian $H_{imp}$ which would limit its utility. A natural question is whether ground states posses some universal features
that depend only on the bath Hamiltonian $H_0$ and the size of the impurity $m$.
In this paper we provide one example of such a universal feature.
To state our  result consider an arbitrary impurity model
$H=H_0+H_{imp}$. Choose a new set of creation-annihilation operators
$b_j^\dag,b_j$  that diagonalize the bath Hamiltonian:
\[
H_0=\sum_{j=1}^n \epsilon_j b_j^\dag b_j, \quad 0\leq \epsilon_j\leq 1.
\]
Here $\epsilon_j$ are  single-particle excitation energies  of the bath. The fact that $\epsilon_j\in [0,1]$ follows from our assumption that $\|h\|\leq 1$, see the remarks after Eq.~\eqref{bath1}. 
Let  $\omega>0$ be the spectral gap of the bath.
That is, each $\epsilon_j$ is either zero or contained in the interval $[\omega,1]$.
\begin{theorem}[\bf Exponential Decay]
\label{thm:2}
There exists a ground state $\psi$ of $H$ such that the following holds. 
Let $C$ be a hermitian $n\times n$ matrix defined by 
\begin{equation}
C_{jk}=\langle \psi|b^{\dagger}_j b_k|\psi\rangle
\label{Cmat}
\end{equation}
and let $\sigma_1\geq \sigma_2 \geq \ldots \sigma_n\geq 0$ be its eigenvalues. Then for all $j$
\begin{equation}
\label{decay1}
\sigma_j\le c \exp{\left[ -\frac{j}{14m\log{(2\omega^{-1})}}\right]}.
\end{equation}
Here $c>0$ is some universal constant. 
\end{theorem}
Assuming that the impurity has a constant size $m=O(1)$,  the theorem
asserts that the eigenvalues of the ground state covariance matrix 
decay exponentially with an exponent that depends very mildly on the spectral gap 
of the bath. Moreover, if all excitation energies of the bath are strictly positive, we show that 
Eqs.~(\ref{Cmat},\ref{decay1}) hold for {\em any} ground state $\psi$ of $H$.
Let us emphasize that the exponential decay Eq.~(\ref{decay1}) 
is a universal feature of a ground state that  has no dependence on $H_{imp}$.
Also we note that Theorem~\ref{thm:2} assumes nothing about the spectrum of the
full Hamiltonian $H$.

An important corollary of Theorem~\ref{thm:2} is that an exact ground state of $H$
can be well approximated by a superposition of a small number of fermionic Gaussian
states. Indeed, we show that all excitations present in the bath
can be ``localized" on a small subset of modes by some Gaussian 
unitary operator. 
Informally, each zero eigenvalue of the covariance
matrix $C$ can be identified with an empty fermionic mode $|0\rangle$ (after a suitable
Gaussian unitary transformation). Thus, if $C$ has at most $k$ non-zero eigenvalues, 
at least $n-k$ modes must be in the vacuum state. Any state of the remaining $k$ modes
can be written as a superposition of at most 
$2^k$ Gaussian states. Choosing a suitable cutoff value to truncate small eigenvalues of $C$
yields a good approximation of $\psi$ by a superposition of a few Gaussian states,
see Section~\ref{sec:ground} for details. Theorem~\ref{thm:2} 
plays a central role in the analysis of our quasi-polynomial algorithm.

The proof of Theorem~\ref{thm:2} proceeds in two steps. First we use a  variational characterization
of ground states to show that the covariance matrix $C$ must be a feasible solution of a certain
semidefinite program that depends only on $H_0$ and the linear subspace spanned by the impurity modes.
Secondly we prove that any feasible solution $C$ of this program exhibits an exponential
decay of eigenvalues as stated in Eq.~(\ref{decay1}). This step exploits 
the machinery of rational approximations developed by Zolotarev~\cite{zolotarev1877application}
in 1877 -- an extension of Chebyshev's well-known theory of polynomial approximation.
A formal proof of the theorem is presented in Section~\ref{sec:Cthm}.
 
\subsection{Discussion and open problems}

We have shown that the structure of quantum impurity models can be exploited to enable fast computation of the ground energy and low energy states. Our work may find application in the hybrid quantum-classical DMFT algorithm proposed by Bauer et al. in Ref.~\cite{Bauer15}.  This algorithm has two steps which are performed on a quantum computer. The first step is to prepare a ground state $\psi$ of a quantum impurity model. Bauer et al. suggest using quantum adiabatic evolution followed by phase estimation for the state preparation~\cite{Bauer15}. The second step is to compute the impurity model Green's functions; in this step one simulates Schr{\"o}dinger time evolution with the quantum impurity model Hamiltonian starting from a state  simply related to  $\psi$.  
Our work suggests that the state preparation step can be simplified 
by  classically computing an approximate version of the ground state $\psi$.
This approximate ground state is specified as a superposition of a small number of Gaussian states
and can prepared efficiently on a quantum computer using techniques discussed in Section~\ref{sec:QCMA}.
This would obviate the need for the quantum adiabatic evolution -- a heuristic which usually cannot be rigorously justified 
due to a lack of lower bounds on the minimal spectral gap. 

One may ask: is there a good classical algorithm to simulate the time evolution of quantum impurity models? If so, this might obviate the need for the second step of the algorithm from Ref.~\cite{Bauer15}. As we now explain, a recent work by Brod and Childs~\cite{BC14} provides evidence that 
such efficient classical simulation may not be possible. Consider a system of $n$ qubits and Hamiltonian 
\begin{equation}
H(t)=g(t)(X_{n-2}X_{n}+Y_{n-2}Y_{n})+ \sum_{i=1}^{n-2} f_{i}(t)(X_i X_{i+1} +Y_i Y_{i+1}).
\label{eq:XY}
\end{equation}
Here $f_{i}(t),g(t)$ are time-dependent coefficients and $X_j,Y_j$ are the Pauli operators acting on the $j$th qubit. 
Section IV of Ref.~\cite{BC14} shows  that Schr{\"o}dinger time evolution with the above Hamiltonian can efficiently simulate a quantum computation on $\Omega(n)$ qubits.  The Hamiltonian Eq.~\eqref{eq:XY} can be rewritten in terms of Majorana operators $\{c_1,c_2,\ldots,c_{2n}\}$ using the standard Jordan-Wigner transformation (see Eqs.~(\ref{eq:JW1}-\ref{eq:JW4})); it takes the form of a quantum impurity model with impurity size $m=6$:
\[
H(t)=-g(t)(c_{2n-4}c_{2n-1}c_{2n-3}c_{2n-2}-c_{2n-5}c_{2n}c_{2n-3}c_{2n-2})-i\sum_{j=1}^{n} f_j(t)c_{2j}c_{2j+1}.
\]
Putting this together we see that Schr{\"o}dinger time evolution with a \textit{time-dependent} impurity model Hamiltonian can perform efficient universal quantum computation (and is therefore unlikely to be efficiently classically simulable). It is an open question whether or not a time-independent impurity Hamiltonian can also perform efficient universal quantum computation.

The most direct open question raised by our work is whether or not our bounds can be improved. For example, we do not know if the exponential decay stated in theorem \ref{thm:2} can be strengthened, for example, by eliminating the dependence on the spectral gap of the bath $\omega$.  A related question is whether or not there is an algorithm for estimating the ground energy of quantum impurity problems which scales polynomially as a function of $\gamma^{-1}$ (where $\gamma$ is the desired precision). 

Our results demonstrate that a ground state of a quantum impurity model can be approximated by a superposition of a small number of fermionic Gaussian states.  In Section~\ref{sec:variational} we describe a  simplified practical version of our algorithm which is based on using such states as a variational ansatz.
The algorithm allows one to minimize the energy of an arbitrary  fermionic Hamiltonian 
with quadratic and quartic interactions over  superpositions of $\chi$ Gaussian states,
where $\chi$ is a fixed parameter. 
We hope that this variational algorithm may be useful in other contexts
beyond the  study of quantum impurity models. For example, in quantum chemistry, the Hartree-Fock approximation
is based on minimizing the energy of a fermionic Hamiltonian over Slater determinant states.
A generalized Hartree-Fock method proposed by Kraus and Cirac in Ref.~\cite{KC10}
is a variational algorithm that minimizes the energy  within the larger class of Gaussian states. 
In contrast to Slater determinants, Gaussian states are capable of describing certain
correlations between electrons such as emergence of Cooper pairs in the BCS theory of superconductivity. 
Our work 
extends the algorithm of Ref.~\cite{KC10} to arbitrary superpositions  of $\chi$ Gaussian states.
In Section~\ref{sec:variational} we use the single impurity Anderson model~\cite{Anderson61} 
as a toy model to benchmark our variational algorithm. We found that the $\chi=2$  algorithm approximates the
ground energy within the first eight significant digits for $n\le 40$.

\subsection{Miscellaneous results}
\label{sec:misc}

In this section we collect results that are not directly related to 
our quasi-polynomial algorithm. These results provide 
additional insights on the structure of ground states of quantum impurity models and
the complexity of estimating their ground energy.

\paragraph{Energy distribution:}
The following theorem states that the ground state of a quantum impurity model has almost all of its support in a subspace consisting of low energy states for the bath Hamiltonian $H_0$. It is an analogue of a result due to Arad, Kuwahara, and Landau in the context of local spin systems \cite{AKL14}.

\begin{restatable}{theorem}{Arad}
\label{thm:Arad2}
Let $Q_\tau$ be the projector onto a subspace spanned by eigenvectors of $H_0$
with energy at most $\tau$. 
Let $\psi$ be any ground state of the full Hamiltonian $H$. 
Then 
\begin{equation}
\label{low_energy3}
\| (I-Q_\tau) \psi\| \le 2 \exp{\left[-\frac{\tau}{4}\log\left(\frac{\tau}{8em}\right)\right]}
\end{equation}
for all $\tau\ge 8em$. Here $e\equiv \exp{(1)}$.
\end{restatable}
While it has no dependence on the norm of $H_{imp}$, the utility of theorem \ref{thm:Arad2} depends on the spectrum of $H_0$. The result can be powerful in certain cases, e.g., if $H_0$ has a constant spectral gap. On the other hand if all excitation energies of $H_0$ are sufficiently close to zero ($\leq 8em/n$, say) then the result is trivial. The proof of theorem \ref{thm:Arad2} is given in Section \ref{subs:Arad2proof}.

\paragraph{Efficient algorithm for gapped impurity models:}
The algorithm from theorem \ref{thm:1} does not require any condition on the spectral gap of the quantum impurity model Hamiltonian.  We show that if the full Hamiltonian $H$ has a constant spectral gap then its ground energy can be approximated efficiently using a different technique.
 \begin{restatable}{theorem}{gap}
Suppose the impurity has size $m=O(1)$. Suppose the full Hamiltonian $H$ has a non-degenerate ground state and a constant energy gap above the ground state. Then
there exists a classical algorithm that approximates 
 the ground energy $e_g$ within a given precision $\delta$  in time $\mathrm{poly}(n,\delta^{-1})$.
\label{thm:gap}
\end{restatable}

The proof of the theorem is given in Section \ref{sec:gappedcase}. It proceeds by establishing an efficiently computable mapping (unitary transformation) between the impurity model and a Hamiltonian which describes a chain of $O(n)$ qudits with nearest neighbor interactions and maximum qudit dimension $2^{m}$. The mapping does not require any condition on the gap, and since it is unitary the spectrum of the two models coincide. In the gapped case the ground energy can be computed efficiently using known algorithms for 1D gapped systems \cite{1Dgapped}.

\paragraph{Approximation with inverse polynomial precision:}
Let us now consider the complexity of estimating the ground energy of a quantum impurity model to inverse polynomial precision. 
Formally, consider a decision version of the problem:

\begin{restatable}{Qproblem}{qprob}
We are given a quantum impurity model $H$ with $n$ fermi modes and impurity size $m=O(1)$, and two energy thresholds $a<b$ such that $b-a=1/\mathrm{poly(n)}$. We are promised that either $e_g\leq a$ (yes instance) or $e_g \geq b$ (no instance) and asked to decide which is the case.
\end{restatable}
Note that the algorithm from Theorem~\ref{thm:1} has quasipolynomial run time if $\gamma$ scales inverse polynomially with $n$. Although our algorithm is not efficient in this precision regime, we are able to prove the following complexity upper bound.
\begin{restatable}{theorem}{qcma}
\label{thm:QCMA}
The quantum impurity problem is contained in QCMA.
\end{restatable}
The proof of theorem \ref{thm:QCMA} is given in Section \ref{sec:QCMA}. Here QCMA is a quantum analog of NP \cite{QNP}. Roughly speaking, it consists of those decision problems where every yes instance has a polynomial-sized classical proof which can be efficiently verified using a quantum computer. 

Finally, in Appendix \ref{app:normest} we include an additional algorithmic tool for manipulating superpositions of Gaussian states. Although we do not use this tool in the present paper, we hope that it finds some application elsewhere. In particular, we describe a classical algorithm which estimates the norm of a state $\phi$ given as a superposition of $\chi$ fermionic Gaussian states. The runtime scales only linearly with $\chi$, improving upon a naive $O(\chi^2)$ algorithm. A similar fast norm estimation algorithm for superpositions of \textit{stabilizer states} was given in Ref~\cite{BG16}.

%%%%%%%%%%%%%%%%%%%%%%%%%%%%%%%%%%%%%%%%%%%%%%%%%%%%%%%%%%%%%%%%%%
%%%%%%%%%%%%%%%%%%%%%%%%%%%%%%%%%%%%%%%%%%%%%%%%%%%%%%%%%%%%%%%%%%

\section{Background}
\label{sec:background}

To make the paper self-contained, in this section we briefly summarize some basic 
facts concerning free fermion Hamiltonians and fermionic Gaussian states.
The material of this section is mostly based on Refs.~\cite{terhal2002classical,bravyi2004lagrangian}.

\subsection{Canonical modes}
\label{sec:cmodes}

Consider a quadratic  Hamiltonian
\[
H_{0}=e_0 I + \frac{i}4\sum_{p,q=1}^{2n} h_{p,q} c_p c_q
\]
where $h$ is a real anti-symmetric matrix and $e_0$ is an energy shift chosen
such that $H_0$ has zero ground energy.
Given a complex vector $x\in \CC^{2n}$, define an operator
\[
b(x)=\sum_{j=1}^{2n} x_j c_j.
\]
Let us say that $b(x)$ is a {\em canonical mode} of $H_0$ if
\begin{equation}
\label{cmodes1}
[H_0,b(x)]=-\epsilon b(x), \quad \quad \epsilon\ge 0.
\end{equation}
Majorana commutation rules Eq.~(\ref{MCR}) give
\[
[H_0,b(x)]=i b(hx),
\]
that is,  $x$ must be an eigenvector of $h$ with an eigenvalue $i\epsilon$.
Furthermore, since $b(x)$ reduces  the energy of any eigenvector
of $H_0$ by $\epsilon$ and $b(x)^2$ is proportional to the identity
due to Eq.~(\ref{MCR}), we conclude that $b(x)^2=0$ whenever $\epsilon>0$.
Choosing an orthonormal set of eigenvectors of $h$ and noting that 
\[
\{b^\dag(x),b(y)\}=\{b(x^*),b(y)\}=2\sum_{j=1}^{2n} x_j^* y_j
\]
one can construct a complete set of canonical modes $b_1,\ldots,b_n$
such that 
\begin{equation}
\label{cmodes2}
H_0=\sum_{j=1}^n \epsilon_j b_j^\dag b_j, \qquad \epsilon_j\ge 0,
\end{equation} 
\begin{equation}
\label{cmodes3}
 b_j^2=0, \qquad  \{b_i,b_j^\dag\}=\delta_{i,j}I.
\end{equation}
Specifically, if $u_1,\ldots,u_n\in \CC^{2n}$ are orthonormal eigenvectors of $h$  such that
$hu_j=i\epsilon_j u_j$ with $\epsilon_j\ge 0$ then $b_j=b(u^j)$. 
We shall refer to the operators $b_1,\ldots,b_n$ constructed above
as {\em canonical modes} of $H_0$. 
Canonical modes can be computed  in time $O(n^3)$ by
diagonalizing $h$. 

\subsection{Pfaffians}
\label{sec:pf}

Suppose $n=2k$  and $M\in \CC^{n\times n}$ is a complex anti-symmetric matrix.
The Pfaffian of $M$ denoted $\pff{M}$ is a complex number defined as 
\begin{equation}
\label{pf}
\pff{M}=\frac1{2^k k!} \sum_{\sigma\in S_n} (-1)^\sigma M_{\sigma(1),\sigma(2)} \cdots M_{\sigma(n-1),\sigma(n)}.
\end{equation}
Here the sum runs over the symmetric group $S_n$ and $(-1)^\sigma$ is the parity of a
permutation $\sigma$. Let us agree that $\pff{M}=0$ whenever $M$ has odd size. 
For small $n$ one can compute Pfaffians directly from the definition:
\begin{equation}
\label{PF24}
\ba{rcl}
n=2 & : &  \pff{M}=M_{1,2} \\
n=4 & : & \pff{M}=M_{1,2}M_{3,4}-M_{1,3}M_{2,4} + M_{1,4}M_{2,3}\\
\ea
\end{equation}
The well-known properties of the Pfaffian are
\begin{equation}
\label{pfp1}
\pff{M}^2=\det{(M)}
\end{equation}
and
\begin{equation}
\label{pfp2}
\pff{RMR^T}=\det{(R)} \pff{M}
\end{equation}
for any complex matrix $R$.
Using Eq.~(\ref{pfp1}) one can compute $\pff{M}$ up to an overall sign in time $O(n^3)$.
Most of the algorithms for computing the Pfaffian proceed by 
transforming $M$ into a
tri-diagonal form~\cite{Wimmer2011,rubow2011factorization}.
This transformation gives a complex matrix $R$ whose determinant is easy to compute
such that 
\begin{equation}
\label{pfp3}
RMR^T=\sum_{p=1}^{n-1} x_p (|p\rangle\langle p+1| -|p+1\rangle\langle p|)
\end{equation}
for some complex coefficients $x_p$.
Definition Eq.~(\ref{pf}) and Eq.~(\ref{pfp2}) then imply that 
\begin{equation}
\label{pfp4}
\pff{M}=\det{(R)}\cdot (x_1 x_3\cdots x_{n-1}).
\end{equation}
From Eqs.~(\ref{pfp3},\ref{pfp4}) one can compute $\pff{M}$ including the overall sign using $O(n^3)$
arithmetic operations.
A detailed discussion of algorithms and optimized implementations
can be found in Ref.~\cite{Wimmer2011}.

\subsection{Gaussian unitary operators and  gaussian states}
\label{sec:Cliff}

A unitary operator $U$ acting on the Fock space $\calH_n$
is called Gaussian\footnote{Gaussian unitary
operators are sometimes called canonical  or Bogolyubov transformations.}
 if its conjugated action  maps any Majorana operator $c_p$ to a linear
combination of Majorana operators $c_1,\ldots,c_{2n}$ that is, 
\begin{equation}
\label{UvsR}
Uc_pU^\dag =\sum_{q=1}^{2n} R_{p,q} c_q
\end{equation}
for some real orthogonal matrix $R$. Gaussian unitary operators
form a group $\calC_n$ which coincides with $O(2n)$ if one ignores the overall phase
of  operators.  The group $\calC_n$ is generated by operators
$U=\exp{(\frac{\theta}2 c_p c_q)}$ that implement  rotations
\[
Uc_p U^\dag = \cos{(\theta)} c_p - \sin{(\theta)} c_q\qquad
\mbox{and} \quad Uc_q U^\dag =\sin{(\theta)} c_p +  \cos{(\theta)} c_q.
\]
and by reflection-like operators $U=c_p$ that flip the sign of all $c_q$
with $q\ne p$. 

A state  $\phi\in \calH_n$ is called  Gaussian
 iff it is obtained from the vacuum
state $|0^n\rangle$ by a Gaussian unitary, i.e., $|\phi\rangle=U|0^n\rangle$
for some $U\in \calC_n$.
Let $\calG_n$ be the set of all Gaussian states.
Any  state $\phi\in \calG_n$ can be specified up to an overall phase by its 
covariance matrix $M$ of size $2n\times 2n$ which is defined as
\begin{equation}
\label{M}
M_{p,q}=(-i/2)\langle \phi| (c_p c_q-c_q c_p)|\phi\rangle.
\end{equation}
 By definition, $M$ is a real anti-symmetric matrix.
For example, a Fock basis state $|y\rangle$ has a block-diagonal covariance matrix 
\begin{equation}
\label{V}
M_{y}\equiv \bigoplus_{j=1}^n \left[ \ba{cc} 0 & (-1)^{y_j} \\ (-1)^{y_j+1}  & 0 \\ \ea \right].
\end{equation}
A Gaussian state $|\phi\rangle=U|y\rangle$ has covariance matrix $M=RM_{y} R^T$,
where $R\in O(2n)$ is defined by Eq.~(\ref{UvsR}). 
This implies that a valid covariance matrix must satisfy $M^2=-I$. 
Conversely, any real anti-symmetric matrix $M$ such that $M^2=-I$
is a covariance matrix of some Gaussian state. 

Let $N=\sum_{j=1}^n a_j^\dag a_j$ be the
particle number operator and 
\begin{equation}
\label{Parity}
P=(-1)^N = (-i)^n c_1 c_2 \cdots c_{2n-1} c_{2n}
\end{equation}
be the fermionic parity operator. Note that any Gaussian unitary $U\in \calC_n$
either commutes or anti-commutes with $P$ since the generators
$\exp{(\frac{\theta}2 c_p c_q)}$ and $c_p$ commute and anti-commute with $P$
respectively. Since $P|0^n\rangle=|0^n\rangle$, it follows that any Gaussian
state $\phi\in \calG_n$ has a fixed parity: $P\phi=\sigma \phi$ for some
$\sigma=\pm 1$.  From Eq.~(\ref{pfp2}) one infers
that $\sigma=\pff{M}$, where $M$ is the covariance matrix of $\phi$.
 We shall say that $\phi$ is even (odd)
if $\sigma=1$ ($\sigma=-1$).  
Gaussian states that are eigenvectors of the number operator $N$
are sometimes called Slater determinants.

We shall say that a state $\psi\in \calH_n$ has a Gaussian rank $\chi$
if it can be written as a superposition of at most $\chi$ Gaussian states.
Gaussian rank is analogous to the Slater number studied in Ref.~\cite{schliemann2001quantum}.
We shall be mostly interested in low-rank Gaussian states, that is,
states with Gaussian rank $\chi=O(1)$ independent of $n$.

Expectation value of any observable on a Gaussian state $\phi\in \calG_n$
can be efficiently computed  using Wick's theorem.
By linearity, it suffices to consider observables proportional to 
Majorana monomials
\begin{equation}
\label{cx}
c(x)=c_1^{x_1} c_2^{x_2} \cdots c_{2n}^{x_{2n}}, \quad \quad x_p\in \{0,1\}.
\end{equation}
Such monomials  form an orthogonal basis in the algebra of 
operators acting on $\calH_n$.
Wick's theorem can be stated  in terms of Pfaffians defined in Section~\ref{sec:pf} as 
\begin{equation}
\label{Wick}
\langle \phi| c(x) |\phi\rangle =  \pff{iM[x]} \quad \mbox{for all $x$},
\end{equation}
where $\phi$ is a Gaussian state, $M$ is the covariance matrix of $\phi$, and
$M[x]$ is a submatrix of $M$ that includes only rows and columns
$j$ such that $x_j=1$. For example,
\[
\langle \phi |c_p c_q c_r c_s|\phi\rangle= -(M_{p,q} M_{r,s} -M_{p,r} M_{q,s} + M_{p,s} M_{q,r})
\]
for any $p<q<r<s$. Note also that $\langle \phi |P|\phi\rangle = \pff{M}$.

Let $H_0$ be a quadratic Hamiltonian considered in Section~\ref{sec:cmodes}
and $b_1,\ldots,b_n$ be its canonical modes. 
Let us order the canonical modes  such that
\[
0\le \epsilon_1\le \epsilon_2 \le \ldots \le \epsilon_n
\]
and let $k$ be the number of zero-energy modes, that is,
$\epsilon_1=\ldots=\epsilon_k=0$ and $\epsilon_{k+1}>0$. Then the ground subspace
of $H_0$ has a form 
\[
\ker{(H_0)}=\mathrm{span}( U|x_1,\ldots,x_k,0,0,\ldots,0\rangle \, : \,  x_i\in \{0,1\} ),
\]
where $U\in \calC_n$ is a Gaussian unitary such that $b_j=Ua_jU^\dag$
for all $j$. In particular, if $H_0$ has no zero-energy modes then
its unique ground state is $U|0^n\rangle$.
The parity of $U|0^n\rangle$ is determined by the sign of $\pff{h}$,
see~\cite{kitaev2001unpaired} for more details.

\subsection{Inner product formulas}
\label{sec:inner}

In this section  we state several useful formulas for 
various inner products  that involve Gaussian states.
They can be viewed as a slightly generalized version of the standard
inner product formulas for Slater determinants~\cite{lowdin1955quantum}.
For the sake of completeness we provide a proof of all
formulas in Appendix~\ref{app:inner}.

Consider Gaussian states $\phi_1,\phi_2\in \calG_n$
with the same parity $\sigma$ 
and let $\rho_a=|\phi_a\rangle\langle \phi_a|$.
Let $M_a$ be the covariance matrix of $\rho_a$.
The magnitude of the inner product $\langle \phi_1|\phi_2\rangle$   is given by
\begin{equation}
\label{ip2}
|\langle \phi_1|\phi_2\rangle|^2=\trace{(\rho_1 \rho_2)} = \sigma 2^{-n}\cdot  \pff{M_1+M_2}.
\end{equation}
Furthermore, $\trace{(\rho_1 \rho_2)}=0$ if $\phi_0,\phi_1$ have different parity.

Suppose $\trace{(\rho_1 \rho_2)}\ne 0$. Can we compute the inner product
$\langle \phi_1|\phi_2\rangle$ including the overall phase~?
To make this question meaningful, first we need to specify a Gaussian state including
the overall phase. To this end, let us  fix some {\em reference state} $\phi_0\in \calG_n$,
for example the vacuum state or a randomly chosen Gaussian state.
We shall specify a Gaussian  state $\phi\in \calG_n$ by
its covariance matrix $M$ and by its inner product with the reference state 
$\langle \phi_0|\phi\rangle$. Here we assume that $\phi$ and $\phi_0$ are not orthogonal.
We show that
\begin{equation}
\label{ip3a}
\langle \phi_0|\phi_1\rangle \cdot \langle \phi_1|\phi_2\rangle \cdot \langle \phi_2|\phi_0\rangle
=\sigma 4^{-n} i^n \, \pff{\left[  \ba{ccc}
iM_0 &-I  & I \\
I & iM_1 &  -I \\
-I & I & iM_2 \\
\ea
 \right]},
\end{equation}
where $\phi_a\in \calG_n$ are Gaussian states with covariance matrices $M_a$
and parity $\sigma$. 
If states $\phi_0,\phi_1,\phi_2$ do not have the same parity then at least
one inner product $\langle \phi_a|\phi_b\rangle=0$ and the righthand
side of Eq.~(\ref{ip3a}) is zero. 
One can rewrite Eq.~(\ref{ip3a})  in a more compact form
that only involves Pfaffians of matrices of size $2n$,
\begin{equation}
\label{ip3}
\langle \phi_0|\phi_1\rangle \cdot \langle \phi_1|\phi_2\rangle \cdot \langle \phi_2|\phi_0\rangle
= 4^{-n} \cdot \pff{M_1+M_2}\cdot \pff{\Delta+M_0}
\end{equation}
where 
\begin{equation}
\label{Delta}
\Delta= (-2I + iM_1 - iM_2)(M_1+M_2)^{-1}.
\end{equation}
Note that $M_1+M_2$ is invertible iff  $\phi_1,\phi_2$
are not orthogonal, see Eq.~(\ref{ip2}).
One can easily check that $\Delta$ is an anti-symmetric matrix (use the identity $M_a^2=-I$).
Using Eqs.~(\ref{ip3a},\ref{ip3}) one can compute the inner product
$\langle \phi_1|\phi_2\rangle$ including the overall phase in time $O(n^3)$.

In order to compute matrix elements $\langle \phi_1|c(x)|\phi_2\rangle$ we shall
need  a generalized Wick's theorem that involves a pair of Gaussian states.
Suppose $x\in \{0,1\}^{2n}$ is an even-weight string
and let $w=|x|$ be its Hamming weight. 
Define a matrix $J_x$ 
of size $w\times 2n$ such that 
$(J_x)_{i,j}=1$ if $j$ is the position of the $i$-th nonzero element of $x$
and $(J_x)_{i,j}=0$  otherwise. Define a diagonal matrix  $D_x$ 
if size $2n\times 2n$  such that $(D_x)_{j,j}=1-x_j$. 
For example, if $x=0^{2n}$ then $D_x=I$ and $J_x$ is an empty matrix. 
We will show that 
\begin{equation}
\label{Wick2}
\langle \phi_0|\phi_1\rangle \cdot \langle \phi_1|c(x)|\phi_2\rangle \cdot \langle \phi_2|\phi_0\rangle =
 \sigma 4^{-n}  i^n \pff{R_x},
\end{equation}
where
\[
R_x=
 \left[ \ba{c|c|c|c}
iM_0  & -I & I &  \\
\hline
I & iM_1 & -I  &  \\
\hline
-I & I & i D_x M_2  D_x & J_x^T +iD_x M_2 J_x^T  \\
\hline
 &  &  -J_x+ i J_x M_2 D_x  & i J_x M_2 J_x^T  \\
\ea\right]
\]
is an anti-symmetric matrix of size $6n+w$.
As before, $\sigma=\pm 1$ is the common parity of $\phi_0,\phi_1,\phi_2$.

In the special case when $\phi_1,\phi_2$  are not orthogonal
one can use  a simplified formula
\begin{equation}
\label{Wick3}
\frac{\trace{\left( \rho_2 \rho_1 c(x) \right)}}{\trace{(\rho_2 \rho_1)}}
=\pff{i\Delta[x]^*}.
\end{equation}
Here $\Delta$ is defined by Eq.~(\ref{Delta})
and $\Delta[x]$ is a submatrix of $\Delta$ that includes only rows and columns
from the support of $x$. We use a notation $\Delta^*$ for the complex conjugate matrix.
Furthermore, if $M_1=M_2=M$ then $\Delta=M$ so that Eq.~(\ref{Wick3})
reduces to the standard Wick's theorem, see Eq.~(\ref{Wick}).
Note that Eq.~(\ref{Wick3}) can be rewritten as
\[
\langle \phi_1|c(x)|\phi_2\rangle = \langle \phi_1|\phi_2\rangle \cdot \pff{i\Delta[x]^*}.
\]
Computing the inner product $\langle \phi_1|\phi_2\rangle$ using Eq.~(\ref{ip3a})
or Eq.~(\ref{ip3}) this gives $\langle \phi_1|c(x)|\phi_2\rangle$.

Note that once the matrix $\Delta$ has been computed,
Eq.~(\ref{ip3})  determines the inner product $\langle\phi_1|\phi_2\rangle$,
and thus Eq.~(\ref{Wick3}) enables computation of
$\langle \phi_1|c(x)|\phi_2\rangle$ in time $O(|x|^3)$ independent of $n$ for any $x$.
This is particularly useful when $c(x)$ describes a single term in some fermionic
Hamiltonian $H$. In this case $x$ usually has Hamming weight $O(1)$
and one can compute  $\langle \phi_1|H|\phi_2\rangle$ in time $O(k)$,
where $k$ is the number of terms in $H$. 
To the best of our knowledge, Eqs.~(\ref{ip3a},\ref{ip3},\ref{Wick2},\ref{Wick3}) are new.

%%%%%%%%%%%%%%%%%%%%%%%%%%%%%%%%%%%%%%%%%%%%%%%%%%%%%%%%%%%%%%%%%%
%%%%%%%%%%%%%%%%%%%%%%%%%%%%%%%%%%%%%%%%%%%%%%%%%%%%%%%%%%%%%%%%%%

\section{Ground states of quantum impurity models}
\label{sec:ground}

In this section we study exact ground states of  quantum impurity models
and establish some of their features. We first prove Theorem~\ref{thm:2} and discuss its implications
for the approximation of ground states by superpositions of Gaussians. We will also see how bath excitations can be 
``localized" by a Gaussian unitary operator. Second, we prove Theorem~\ref{thm:Arad2}, i.e., we show that any ground state of the full Hamiltonian  has most of its weight
on a certain low-energy subspace of the bath Hamiltonian $H_0$. This feature of ground states 
has been originally proved   for quantum spin systems
by Arad,  Kuwahara, and  Landau~\cite{AKL14}. Our proof 
borrows many ideas from  Ref.~\cite{AKL14}, although it is technically
different and gives a slightly stronger bound.

\subsection{Proof of Theorem~\ref{thm:2} }
\label{sec:Cthm}

Let $b_j$ and $\epsilon_j$ be the canonical modes and single-particle excitation energies of the bath Hamiltonian $H_0$,
see Section~\ref{sec:cmodes}. Then 
\[
H_0=\sum_{j=1}^n \epsilon_j b_j^\dag b_j, \qquad 0\leq \epsilon_j\leq 1.
\]
We shall first assume that $\epsilon_j\geq \omega>0$ for all $j=1,2,\ldots,n$. In this case we will show that Eq.~\eqref{decay1} holds for any ground state $\psi$ of $H$. At the end of the proof we handle the case where one or more single particle excitation energy is zero and we show that in this case Eq.~\eqref{decay1} holds for at least one ground state of $H$ .

For any complex vector $x\in \CC^n$ define a fermionic operator
\[
b(x)\equiv \sum_{j=1}^n x_j b_j.
\]
Recall that the impurity is formed
by the first $m$ Majorana modes $c_1,\ldots,c_m$.  
Define a linear subspace $\calL\subseteq \CC^n$ such that $x\in \calL$
iff the expansion of $b(x)$ in terms of the Majorana operators
$c_1,\ldots,c_{2n}$  does not include $c_1,\ldots,c_m$. 
We note that 
\begin{equation}
\label{dimL}
\dim{(\calL)}\ge n-m
\end{equation}
since $\calL$ is described by  $m$ linear constraints on $n$ variables
$x_1,\ldots,x_n$.
Majorana commutation rules $c_p c_q=-c_p c_q$ for $p\ne q$ imply that
$b(x)$ anti-commutes with $c_1,\ldots,c_m$ whenever $x\in \calL$.
Since  $H_{imp}$ includes
only even-weight monomials in $c_1,\ldots,c_m$, we get
\begin{equation}
\label{commute1}
[H_{imp},b(x)]=0 \quad \mbox{for all} \quad x\in \calL.
\end{equation}

Suppose $\psi$ is any normalized ground state of $H$. Then 
\begin{equation}
\label{ground1}
\langle \psi| b(x)^\dag [b(x),H] |\psi\rangle \le 0.
\end{equation}
since $b(x)\psi$ cannot have energy smaller than $\psi$.
Combining this and Eq.~(\ref{commute1}) we conclude that 
\begin{equation}
\label{ground2}
\langle \psi| b(x)^\dag [b(x),H_0] |\psi\rangle \le 0 \quad \mbox{for all} \quad x\in \calL.
\end{equation}
Using the commutation rules $[b_j, b_k^\dag b_k]=\delta_{j,k} b_j$ one gets
\begin{equation}
\label{ground3}
\sum_{j,k=1}^n \bar{x}_j x_k \epsilon_k \langle \psi| b_j^\dag b_k|\psi\rangle \le 0
\quad \mbox{for all $x\in \calL$}.
\end{equation}
Define a ground state covariance matrix $C$ and a single-particle energy matrix $E$ such that 
\begin{equation}
\label{CE1}
C_{j,k}=\langle \psi|b_j^\dag b_k|\psi\rangle \quad \mbox{and} \quad
E_{j,k}=\epsilon_j \delta_{j,k}.
\end{equation}
Since for now we assume all single particle energies are at least $\omega$ we have
\begin{equation}
\label{eq:Econstraint}
\omega I \leq E\leq I.
\end{equation}
We note that 
\begin{equation}
\label{CE2}
0\le C\le I 
\end{equation}
Indeed, $C\ge 0$ since $C$ is a covariance matrix. 
Commutation rules $b_j^\dag b_k+b_k b_j^\dag = \delta_{j,k}I$
and the assumption $\langle \psi|\psi\rangle=1$
imply that  $(I-C)_{j,k}= \langle \psi |b_k b_j^\dag|\psi\rangle$, that is,
$I-C\ge 0$ which proves Eq.~(\ref{CE2}). Next we observe that $CE$ must have a real non-positive expectation
value on any vector $x\in \calL$ due to 
 Eq.~(\ref{ground3}).  Equivalently, a restriction of $CE$ onto $\calL$
defines a hermitian negative semi-definite matrix. 
Denoting $\Lambda$ a projector onto $\calL$ one can rewrite Eq.~(\ref{ground3}) in a matrix form as
\begin{equation}
\label{CE3}
\Lambda (CE-EC) \Lambda = 0 \quad \mbox{and} \quad \Lambda CE \Lambda\le 0.
\end{equation}
Finally, 
\begin{equation}
\label{CE4}
\mathrm{rank}{(\Lambda)}\ge n-m \qquad \mbox{and} \qquad \Lambda^2=\Lambda.
\end{equation}
due to Eq.~(\ref{dimL}). We shall use Eqs.~(\ref{eq:Econstraint}-\ref{CE4}) to prove the following  lemma.
It establishes a bound slightly stronger than Eq.~\eqref{decay1} in the case where $\epsilon_j\ge \omega$ for all $j$. 
\begin{lemma}
\label{lemma:decay}
Let $\Lambda, C,E$ be $n\times n$ hermitian matrices and $\omega$ be a positive number such that Eqs.~(\ref{eq:Econstraint},\ref{CE2},\ref{CE3},\ref{CE4}) are satisfied. Let $\sigma_1\ge \sigma_2 \ge \ldots \ge \sigma_n$ be the eigenvalues of $C$.
Then for all $j$
\begin{equation}
\label{decay2}
\sigma_j\le c \exp{\left[ -\frac{j}{7m\log{(2\omega^{-1})}}\right]}.
\end{equation}
Here $c>0$ is some universal constant.
\end{lemma}
\begin{proof}
Let us first  sketch the main steps of the proof.
Suppose one can find a subspace $\calD\subseteq \calL$ such that 
\begin{equation}
\label{sketch1a}
[C,E^{1/2}]\cdot \calD=0.
\end{equation}
Here and below we define the square root $E^{1/2}$ such that its eigenvalues  are non-negative. 
Let $\Delta$ be the projector onto $\calD$. 
Clearly, $\Delta E^{1/2}CE^{1/2}\Delta\ge 0$ since $C\ge 0$. On the other hand,
\[
\Delta E^{1/2}CE^{1/2}\Delta =\Delta CE \Delta - \Delta [C,E^{1/2}]E^{1/2}\Delta=\Delta CE\Delta
\]
due to Eq.~(\ref{sketch1a}). The inclusion $\calD\subseteq \calL$ and Eq.~(\ref{CE3}) imply
$\Delta CE\Delta=\Delta \Lambda CE \Lambda \Delta \le 0$ and thus
\begin{equation}
\label{sketch1b}
\Delta E^{1/2}CE^{1/2} \Delta\le 0.
\end{equation}
To avoid a contradiction one has to assume that $\Delta E^{1/2}CE^{1/2} \Delta=0$.
Since $C\ge 0$,   this is possible only if 
$E^{1/2}\cdot \calD \subseteq \mathrm{ker}(C)$.
If the subspace $\calD$ is sufficiently large, this would show that $C$ has sufficiently many zero
eigenvalues. Unfortunately, condition Eq.~(\ref{sketch1a}) appears to be too strong.
Instead we shall realize an approximate version of Eq.~(\ref{sketch1a}). In other words, we shall construct a subspace $\calD\subseteq \calL$ and
a hermitian operator $Z$ acting on $\CC^n$ such that 
\begin{equation}
\label{sketch2}
[C,Z]\cdot \calD=0 \quad \mbox{and} \quad Z\approx E^{1/2}.
\end{equation}
The desired operator $Z$  will be constructed using a low-degree
rational approximation to the square root function due to
Zolotarev~\cite{zolotarev1877application}.
This approximation has a form 
\[
x^{1/2} \approx \frac{x P_d(x)}{Q_d(x)} \qquad \mbox{for $\omega \le x\le 1$}
\]
where $P_d(x),Q_d(x)$ are degree-$d$ polynomials and the approximation error
scales as $\exp{\left[-cd/\log{(2\omega^{-1})}\right]}$ for some constant $c>0$.
(For comparison, approximating
 $x^{1/2}$  by the Taylor series at $x=1$ truncated at some order $d$
achieves an approximation error $\exp{\left[-cd \omega\right]}$
which is significantly worse if $\omega$ is small.)
We shall approximate $E^{1/2}$ by an operator $Z=E P_d(E) Q_d^{-1}(E)$.
Let us write 
\begin{equation}
\label{sketch3}
Z=M_1M_2\ldots M_{2d+1},
\end{equation}
where each $M_i$ is a "monomial" $M_i=E+\lambda_i I$ or $M_i=(E+\lambda_iI)^{-1}$
for some real numbers $\lambda_i$.
A subspace $\calD\subseteq \calL$ satisfying Eq.~(\ref{sketch2})
 is constructed in two steps. First, construct a subspace
$\calD_0\subseteq \calL$ such that $[C,E]\cdot \calD_0=0$. 
We shall choose $\calD_0$ as the intersection of $\calL$ and $\mathrm{ker}([C,E])$.
Secondly define
\begin{equation}
\label{Dspace}
\calD=\calD_0 \cap \left( M_1^{-1} \calD_0 \right) \cap \left( M_2^{-1} M_1^{-1} \calD_0 \right) \cap \cdots
\cap \left(M_{2d+1}^{-1} \cdots M_2^{-1} M_1^{-1} \calD_0\right).
\end{equation}
We claim that $[C,Z]\phi=0$ for any  state $\phi\in \calD$. 
Using the chain rule for the commutator $[C,Z]=[C,M_1\cdots M_{2d+1}]$ it suffices to check that 
\[
[C,M_{j+1}]M_j \cdots M_2 M_1 \phi=0 \quad \mbox{for all $\phi\in \calD$}
\]
for all $j$. Suppose first that $M_{j+1}=E+\lambda_{j+1}I$. 
By construction of $\calD$, there exists a state $\phi'\in \calD_0$ such that
$\phi = M_j^{-1} \cdots M_2^{-1} M_1^{-1} \phi'$. Since all $M$'s commute, one has
\[
[C,M_{j+1}]M_j \cdots M_2 M_1 \phi= [C,M_{j+1}] \phi'=[C,E] \phi'=0
\]
since $\phi'\in \calD_0$ and $[C,E]\cdot \calD_0=0$.
Suppose next that $M_{j+1}=(E+\lambda_{j+1} I)^{-1}$. Then
\[
[C,M_{j+1}]M_j \cdots M_2 M_1 \phi
=-M_{j+1}[C,M_{j+1}^{-1}] M_{j+1} M_j \cdots M_2 M_1 \phi.
\]
Again, by construction of $\calD$, there exists $\phi'\in \calD_0$
such that $\phi=M_{j+1}^{-1} \cdots M_2^{-1} M_1^{-1} \phi'$.
Taking into account that $M_{j+1}^{-1}=E+\lambda_{j+1}I$ one gets
\[
[C,M_{j+1}]M_j \cdots M_2 M_1 \phi=-M_{j+1} [C,M_{j+1}^{-1}] \phi'=-M_{j+1} [C,E] \phi'=0
\]
since $\phi'\in \calD_0$ and $[C,E]\cdot \calD_0=0$.
This proves that $[C,Z]\cdot \calD=0$ as promised.
How large is the subspace $\calD$~?
We shall use the fact that $\calL$ has dimension at least $n-m$ to show
that $\calD_0$ has dimension at least $n-3m$ and  $\calD$ has dimension  $r\ge n-6md+O(1)$.
Repeating the steps that has lead to Eq.~(\ref{sketch1b}) with $E^{1/2}$
replaced by $Z$ we will show that $\Delta E^{1/2} CE^{1/2}\Delta$ 
has an exponentially small norm. We then use Cauchy's interlacing theorem  to show that
$C$ has at least $r=n-O(md)$ eigenvalues with magnitude at most
$\exp{\left[-cd/\log{(2\omega^{-1})}\right]}$. 
Since we are free to choose $d$ arbitrarily, this is possible only if the eigenvalues of $C$
decay exponentially.

Let us now proceed to a formal proof of Lemma~\ref{lemma:decay}. Define a subspace
\begin{equation}
\label{D0}
\calD_0=\calL \cap   \ker{(CE-EC)}
\end{equation}
and let $\Delta_0$ be the projector on $\calD_0$.
By definition, $[C,E]\Delta_0=0$. From Eqs.~(\ref{CE3},\ref{CE4}) one infers that  the commutator $[C,E]$ has rank at most $2m$,
that is, $\ker{(CE-EC)}$ has dimension at least $n-2m$.
Therefore 
\[
\dim{(\calD_0)} \ge \dim{(\calL)} + \dim{( \ker{[C,E]})} -n \ge (n-m)  + (n-2m) -n = n-3m.
\]
We conclude that any matrix $C$ that satisfies Eqs.~(\ref{CE2},\ref{CE3})
must also satisfy 
\begin{eqnarray}
0\le C\le I \label{SDP1}\\
(CE-EC)\Delta_0 =0 \label{SDP2}\\
\Delta_0 CE \Delta_0  \le  0 \label{SDP3} 
\end{eqnarray}
and
\begin{equation}
\label{SDP4}
\mathrm{rank}(\Delta_0) \ge  n-3m.
\end{equation}
Here in the third line we noted that $\Delta_0 CE \Delta_0$
is a restriction of a negative semi-definite operator $\Lambda CE \Lambda$ onto  $\calD_0$.
Our goal is to show that any matrix $C$ satisfying Eqs.~(\ref{SDP1}-\ref{SDP3})
has exponentially decaying eigenvalues as claimed in Eq.~(\ref{decay2}).

We shall approximate the matrix square root $E^{1/2}$ by rational functions of a sufficiently small degree. 
Optimal rational approximations of a given degree to the sign and square root functions were found
in 1877 by Zolotarev~\cite{zolotarev1877application}. 
They have been used in more recent times for high-energy physics simulations, 
see~\cite{Kennedy2004approximation,chiu2002note}.
The following lemma quantifies the quality of  Zolotarev's approximation. 

\begin{restatable}[\bf Zolotarev's approximation]{lemma}{zolo}
\label{lemma:Zo}
For any integer $d\ge 1$ and real number $0<\omega<1$
there exist degree-$d$ polynomials 
$P_d(x)$, $Q_d(x)$ such that
\begin{equation}
\label{Zo1}
\left| \sqrt{x} - x\frac{P_d(x)}{Q_d(x)} \right|
 \le 2\sqrt{x} \cdot \exp{\left[ -\frac{d}{\log{(2\omega^{-1})}}\right]}.
\end{equation}
for all $x\in [\omega,1]$.
All roots of the polynomials $P_d(x)$ and $Q_d(x)$ are negative real numbers.
\end{restatable}
It should be pointed out that the  construction of the polynomials $P_d(x),Q_d(x)$ depends on $\omega$.
For the sake of readability  we omit the dependence  of $P_d(x),Q_d(x)$ on $\omega$ in our notations.
Since the proof of  Lemma~\ref{lemma:Zo} is a simple combination of known 
facts~\cite{gonvcar1969zolotarev,nakatsukasa2015computing}, we postpone it until
Section~\ref{subs:rational}. This section also provides  explicit formulas for the
roots of $P_d(x),Q_d(x)$. We note that it is possible to improve the exponent on the right-hand side of Eq.~\eqref{Zo1} to $-d\pi^2/\log(256\omega^{-1})$ (see the proof), but for ease of notation we use the above bound.

We shall use Zolotarev's approximations to prove the following.
\begin{prop}
\label{prop:decay}
Let  $d\ge 1$ be any integer and let $r\equiv n-6m(d+1)$.
Then  there exists a projector $\Delta$ of rank at least $r$ such that 
any solution $C$ of Eqs.~(\ref{SDP1}-\ref{SDP3}) obeys
\begin{equation}
\label{eq:deltaC}
\| \Delta C\Delta \| \le 10 \exp{\left[ -\frac{d}{\log{(2\omega^{-1})}}\right]}.
\end{equation}
\end{prop}
\begin{proof}
Consider some fixed integer $d$ and 
a gap $\omega>0$ such that all eigenvalues of $E$ lie in the interval
$[\omega,1]$. 
Let $P_d,Q_d$ be the degree-$d$ polynomials
from Lemma~\ref{lemma:Zo} 
such that  $xP_d(x)Q_d^{-1}(x)$ approximates $x^{1/2}$ for all $\omega \le x\le 1$. Define
\[
Z\equiv E P_d(E) Q_d^{-1}(E).
\]
Then Lemma~\ref{lemma:Zo} implies
\begin{equation}
\label{eq:sqrtE1}
Z=E^{1/2} (I+A)  \qquad  \|A\|\leq 2\exp{\left[ -\frac{d}{\log{(2\omega^{-1})}}\right]}
\end{equation}
Without loss of generality we shall assume that the right hand side of Eq.~\eqref{eq:sqrtE1} is at most $\frac{1}{2}$, forcing $\|A\|\leq \frac{1}{2}$. Indeed, if the right hand side of Eq.~\eqref{eq:sqrtE1} is larger than $\frac{1}{2}$, the right-hand side of Eq.~\eqref{eq:deltaC} is $>1$ and the proposition holds trivially. Defining $B=(I+A)^{-1}-1$ and rearranging Eq.~\eqref{eq:sqrtE1} we get
\begin{equation}
\label{eq:sqrtE2}
E^{1/2}=Z(I+B) \qquad   \qquad \|B\|=\|A(I+A)^{-1}\|\leq 2\|A\|.
\end{equation}
Moreover, the operators $Z,A,B,E$ are diagonal over the same basis and mutually commute.

We may decompose
\begin{equation}
\label{Zo6}
Z=\prod_{i=1}^{2d+1} M_i
\end{equation}
where 
$M_i=E+\lambda_i I$ or $M_i=(E+\lambda_iI)^{-1}$.
The operators $M_i$ are hermitian and non-singular since all
roots of $P_d,Q_d$ lie on the negative real axis while $E\ge \omega I>0$. 
Furthermore, all monomials $M_i$ commute with each other.
Let $\calD\subseteq \calD_0$ be the subspace defined by Eq.~(\ref{Dspace}).
We claim that 
\begin{equation}
\label{Dspace1}
\dim{(\calD)} \ge n-6m(d+1).
\end{equation}
Indeed,  since $\calD_0$ has dimension at least $n-3m$,
the orthogonal complement $\calD_0^\perp$ has dimension at most $3m$. 
Likewise, the orthogonal complement
to $(M_j^{-1} \cdots M_1^{-1} \calD_0)^\perp$ has dimension at most $3m$.
Since $\calD^\perp$ is contained in the sum of  $2d+2$ 
subspaces $\calD_0^\perp$ and $(M_j^{-1} \cdots M_1^{-1} \calD_0)^\perp$
with $j=1,\ldots,2d+1$, see Eq.~(\ref{Dspace}), one infers that $\calD^\perp$
has dimension at most $3m(2d+2)$. This proves Eq.~(\ref{Dspace1}).
Below we assume that $d$ is small enough such that $\calD$ is non-empty.
Let $\Delta$ be the projector onto $\calD$.

The arguments below Eq.~(\ref{Dspace}) prove that 
$C$ commutes with $Z$ if restricted onto
the subspace $\calD$, that is,
\begin{equation}
[C,Z]\Delta=0.
\label{eq:commuteCZ}
\end{equation}

Using Eqs.~(\ref{eq:sqrtE1},\ref{eq:commuteCZ}) and the fact that $E,A,Z$ mutually commute we obtain
\begin{align}
\Delta E^{1/2} C E^{1/2} \Delta& =\Delta E^{1/2} C Z\Delta-\Delta E^{1/2} C E^{1/2}A \Delta\nonumber\\
&=\Delta E^{1/2}Z C\Delta-\Delta E^{1/2} C A E^{1/2}\Delta \nonumber\\
&=\Delta E^{1/2} E^{1/2}(I+A) C\Delta-\Delta E^{1/2} C A E^{1/2}\Delta \nonumber\\
&=\Delta EC\Delta+\Delta E^{1/2}A E^{1/2}C\Delta-\Delta E^{1/2} C A E^{1/2}\Delta
\label{eq:threeterms}
\end{align}
Here the first line uses $E^{1/2}=Z-E^{1/2}A$.
Again using the fact that $E,A,Z,B$ mutually commute along with Eqs. ~(\ref{eq:sqrtE1},\ref{eq:sqrtE2},\ref{eq:commuteCZ}) to rearrange the second term:
\begin{align}
\Delta E^{1/2}A E^{1/2}C\Delta & =\Delta E^{1/2}A (I+B)ZC\Delta\nonumber\\
&=\Delta E^{1/2}A (I+B)C Z\Delta\nonumber\\
&=\Delta E^{1/2}A (I+B)C (I+A) E^{1/2}\Delta
\label{eq:firstterm}
\end{align}
Combining Eqs.~(\ref{eq:threeterms},\ref{eq:firstterm}) gives
\begin{equation}
\Delta E^{1/2} C E^{1/2} \Delta-\Delta E C \Delta =\Delta E^{1/2}\bigg(A (I+B)C (I+A)-CA \bigg) E^{1/2} \Delta
\label{eq:Edelta}
\end{equation}
From Eqs.~(\ref{SDP2},\ref{SDP3}) and $\Delta\leq \Delta_0$ we infer that $\Delta E C \Delta$ is Hermitian and negative semidefinite. Using this fact in Eq.~\eqref{eq:Edelta} we obtain the operator inequality
\begin{equation}
\Delta E^{1/2} C E^{1/2} \Delta \leq \Delta E^{1/2}\bigg(A (I+B)C (I+A)-CA \bigg) E^{1/2} \Delta.
\label{eq:finalEdelta}
\end{equation}

Let $\Delta'$ be the projector onto the support of 
$E^{1/2} \Delta E^{1/2}$. 
From Eq.~(\ref{eq:finalEdelta}) one gets
\begin{align}
\label{Gamma2}
\| \Delta'C \Delta' \| &=\max_{\phi \in \Delta} \frac{ \langle \phi|E^{1/2} C E^{1/2}|\phi\rangle}
{\langle \phi|E|\phi\rangle} \nonumber\\
&\le \max_{\phi \in \Delta} \frac{ \langle \phi| E^{1/2}\big(A (I+B)C (I+A)-CA \big) E^{1/2} |\phi\rangle}
{\langle \phi|E|\phi\rangle} \nonumber\\
&\leq \|A (I+B)C (I+A)-CA\|\nonumber\\
&\leq \|A (I+B)C (I+A)\|+\|CA\|.
\end{align}
Noting that $\|I+B\|,\|I+A\|\leq 2$ and $\|C\|\leq 1$, and using 
Eq.~\eqref{eq:sqrtE1} to bound $\|A\|$ we get
\begin{equation}
\| \Delta'C \Delta' \| \leq 10 \exp{\left[ -\frac{d}{\log{(2\omega^{-1})}}\right]}.
\end{equation}
Since $E$ is invertible, $\Delta'$ has rank at least $n-6m(d+1)$, see Eq.~(\ref{Dspace1}).
This proves Eq.~(\ref{eq:deltaC}) with $\Delta=\Delta'$.
\end{proof}

We can now complete the proof of Lemma~\ref{lemma:decay}. Note that it suffices to consider the case where
\[
j\geq c_1 m
\]
for some universal constant $c_1$. Indeed, we may then choose the constant $c$ in Eq.~\eqref{decay2} to satisfy $c>e^{c_1/7\log(2)}$; with this choice for all $j\leq c_1 m$ the right hand side of Eq.~\eqref{decay2} is $\geq 1$ and the claim holds trivially. 

Let $\Delta$ be the projector of rank $r\ge n-6m(d+1)$ from Proposition~\ref{prop:decay}.
Let $\sigma_j(C)$ and $\sigma_j( \Delta  C\Delta )$
be the $j$-th largest eigenvalues of the respective operators.
Cauchy's interlacing theorem implies 
\[
\sigma_j(C)\le \sigma_{j-(n-r)}( \Delta  C\Delta ).
\]
Choose 
\[
d+1=\left\lfloor \frac{j-1}{6m}\right\rfloor
\]
so that $r\ge  n-j+1$. Note that we may choose a universal constant $c_1$ such that $d\geq j/7m$ for all $j\geq c_1m$. Then, for all $j\geq c_1m$, Proposition~\ref{prop:decay} gives
\[
\sigma_j(C)\le \sigma_{j-(n-r)}( \Delta  C\Delta )\le \sigma_1(\Delta C\Delta)=\| \Delta C\Delta\|
\le 
10 \exp{\left[ -\frac{j}{7m\log{(2\omega^{-1})}}\right]}.
\]
As noted above, it is sufficient to establish the bound only for $j\geq c_1 m$. This proves the lemma.
\end{proof}

Let us now handle the case when
$H_0$ has zero-energy modes. Suppose
\[
\epsilon_1=\epsilon_2=\ldots =\epsilon_T=0
\]
and $\epsilon_j\geq \omega$ for $j>T$. We shall first reduce to the case where the total number $T$ of zero energy modes satisfies $T\leq m$.

Recall the subspace $\calL\subseteq \mathbb{C}^n$ defined around Eq.~(\ref{commute1}). Define another linear subspace $\mathcal{Q}\subseteq \mathbb{C}^n$ such that $x\in \mathcal{Q}$ iff $x_{j}=0$ for all $j>T$. We have
\[
\dim(\mathcal{L}\cap \mathcal{Q} )\geq T-m
\]
since $\dim(\mathcal{L})\geq n-m$ and $\dim(Q)=T$. Note that $[b(x),H_0]=0$ for all $x\in \mathcal{Q}$ and $[b(x),H_{imp}]=0$ for all $x\in \mathcal{L}$ and therefore
\[
[b(x),H]=0 \quad \text{ for all } x\in \mathcal{L}\cap \mathcal{Q}. 
\]
We may form a $T\times T$ unitary matrix $U$ where the rows (when padded with $n-T$ zeros) span $\mathcal{Q}$ and the first $\dim(\mathcal{L}\cap \mathcal{Q} )$ rows span $\mathcal{L}\cap \mathcal{Q}$. The new modes
\[
\tilde{b}_j =\begin{cases} \sum_{k=1}^{T}  U_{jk} b_k & j=1,\ldots, T \\
b_j & T<j\le n
\end{cases}
\]
then satisfy 
\[
[\tilde{b}_j,H]=0 \qquad \text{for all } j=1,\ldots, T-m.
\]
Thus we may write $H=I\otimes H'$ where $H'=H_0'+H_{imp}'$ describes the nontrivial action of $H$ on modes $\tilde{b}_j$ for $j>T-m$. Here $H_0'$ has $m$ zero energy modes and spectral gap $\omega$.
We may therefore choose a ground state $\psi$ of the full Hamiltonian $H$ which satisfies $|\psi\rangle=|\tilde{0}^{(T-m)}\otimes \phi\rangle$ where $\phi$ is a ground state of $H'$ and $\tilde{0}^{(T-m)}$ is the vacuum state for all modes $\tilde{b}_j$ with $j\leq (T-m)$.

The covariance matrix $\tilde{C}$ of $\psi$ defined with respect to the new modes, i.e.,
\[
\tilde{C}_{jk}=\langle \psi|\tilde{b}^{\dagger}_j \tilde{b}_k|\psi\rangle
\]
is unitarily equivalent to the matrix $C$ defined in Eq.~\eqref{Cmat} and therefore has the same spectrum. Moreover, $\tilde{b}_j|\psi\rangle=0$ for all $j\leq (T-m)$ and therefore
\[
\tilde{C}=\left(\begin{array}{cc}0_{(T-m)\times (T-m)} & 0_{(T-m)\times (n+m-T)}\\ 0_{ (n+m-T)\times (n+m-T)} & C' \end{array}\right)
\]
where $C'$ is the covariance matrix for $\phi$. Thus, to prove theorem \ref{thm:2} for the original impurity model $H$ it suffices to prove that Eq.~\eqref{decay1} holds for $C'$, the covariance matrix of an (arbitrary) ground state $\phi$ of $H'$. This shows that without loss of generality we may consider the case where $T\leq m$.

To complete the proof it remains to handle the case where $E$ has at most $m$ zero eigenvalues
and all other eigenvalues are at least $\omega$. In this case we show that Eq.~\eqref{decay1} holds for all ground states $\psi$ of $H$.
Define a subspace 
\[
\calL'=\calL\cap \ker{(E)}^\perp.
\]
Recall that $\dim{(\calL)}\ge n-m$. Then 
$\calL'$ has dimension at least $n-2m$. Let $\Lambda'$ be the projector onto $\calL'$, so
\[
\mathrm{rank}(\Lambda')\geq n-2m.
\]
Define a deformed energy matrix $E'$ by setting
$\epsilon_j=\omega$ for $j=1,\ldots,m$. Now 
\[
\omega I \leq E'\leq I.
\]
Furthermore, $E\Lambda'=E'\Lambda'$. Combining this and  Eq.~(\ref{CE2}) gives
\[
\Lambda'(CE'-E'C)\Lambda'=0 \quad \mbox{and} \quad \Lambda' CE'\Lambda'\le 0.
\]

Now apply Lemma~\ref{lemma:decay} with $\Lambda$ and $E$ replaced by
$\Lambda'$ and $E'$, and with $m$ replaced by $2m$. This gives Eq.~\eqref{decay1} and completes the proof.

{\em Remark:} We observe that Eqs.~(\ref{CE2},\ref{CE3}) define a semi-definite
program (SDP) with a variable $C$. The SDP depends on the energy matrix $E$ and the
projector $\Lambda$. Thus  the ground state covariance matrix $C$ must be
a feasible solution of the SDP.  Choosing an objective function
to be maximized one can extract  some useful information
about ground states of a particular impurity model. For example,
maximizing $\trace{(C)}$ subject to Eqs.~(\ref{CE2},\ref{CE3})   provides
an upper bound on the average  number of excitations in the bath
since $\trace{(C)}=\langle \psi |\sum_{j=1}^n b_j^\dag b_j|\psi\rangle$.
We observed numerically that this upper bound is typically much 
better than what one would expect from our  general results, e.g. Lemma~\ref{lemma:decay}. In fact, for the vast majority of impurity models that we examined (but not for all) we observed that $\trace{(C)}\le 2m$
regardless of the gap of $H_0$.

\subsection{Rational approximations to the square root function}
\label{subs:rational}

In this section we prove Lemma~\ref{lemma:Zo}
and provide explicit formula for the roots of the polynomials $P_d(x),Q_d(x)$.
For convenience, we repeat the statement of the lemma.
\zolo*
\begin{proof}
We shall need  bounds on the convergence of Zolotarev's approximations
to the sign function. The following fact 
was proved in Ref.~\cite{gonvcar1969zolotarev} and stated more explicitly 
 in Ref.~\cite{nakatsukasa2015computing} (see page~9).
 \begin{fact}
\label{fact:1}
For each $0<\delta<1$ and integer $d\ge 1$  there exist degree-$d$ polynomials $P_d,Q_d$  such that  
\begin{equation}
\label{Zo2}
\max_{x\in J(\delta)} \left| \mathrm{sgn}{(x)} - x\frac{P_d(x^2)}{Q_d(x^2)} \right|
 \le 2\exp{\left[ -\frac{d \pi K(\sqrt{1-\mu^2})}{4 K(\mu)}\right]}.
\end{equation}
Here $J(\delta)=[-1,-\delta] \cup [\delta,1]$,
\begin{equation}
\label{Zo3}
\mu \equiv \frac{1-\sqrt{\delta}}{1+\sqrt{\delta}}
\end{equation}
and $K(\mu)$ is the complete elliptic integral of the first kind: 
\[
K(\mu)=\int_{0}^{\pi/2} \frac{d\theta}{\sqrt{1-\mu^2 \sin^2{(\theta)}}}
\]
Furthermore, all roots of $P_d$ and $Q_d$ are real and non-positive.
\end{fact}
The lemma follows easily from the above fact. 
Indeed, let $r(d,\delta)$ be right-hand side of Eq.~(\ref{Zo2}).
Then
\[
\left| |x| - \frac{x^2 P_d(x^2)}{Q_d(x^2)} \right| = |x|\cdot \left| \mathrm{sgn}(x)- \frac{x P_d(x^2)}{Q_d(x^2)} \right|
\le   |x|\cdot r(d,\delta)
\]
for all $x\in J(\delta)$.  Changing variables $y=x^2$ one gets 
\begin{equation}
\label{Zo4}
\left| \sqrt{y} - y\frac{P_d(y)}{Q_d(y)} \right| \le \sqrt{y} \cdot  r(d,\sqrt{\omega})
\end{equation}
for all $\omega\le y\le 1$.
It remains to explicitly  compute the error bound $r(d,\delta)$. We shall use some bounds on the elliptic integrals from Ref. \cite{elliptic}. Define
\[
F(\mu)=\frac{\pi}{2} \frac{K(\sqrt{1-\mu^2})}{K(\mu)}.
\]
Using equations (1.3, 1.5) from \cite{elliptic} we have
\begin{equation}
F(\mu)=\frac{\pi^2}{2F(\frac{1-\mu}{1+\mu})}=\frac{\pi^2}{2F(\sqrt{\delta})}.
\label{eq:fmu}
\end{equation}
Now equation (1.6) from \cite{elliptic} states that $F(r)<\log(4/r)$ for all $r\in (0,1)$. Using this bound in the denominator of Eq.~\eqref{eq:fmu} gives
\[
F(\mu)>\frac{\pi^2}{\log(16\delta^{-1})}.
\]
We arrived at
\[
r(d,\delta)=2\exp{\left[-\frac{d}{2}F(\mu)\right]} \le 2\cdot \exp{\left[ -\frac{d\pi^2}{2\log{(16\delta^{-1})}}\right]},
\]
and therefore 
\[
r(d,\sqrt{\omega})\leq 2\cdot \exp{\left[ -\frac{d\pi^2}{\log{(256\omega^{-1})}}\right]}\leq 2\cdot \exp{\left[ -\frac{d\pi^2}{8\log{(2\omega^{-1})}}\right]}.
\]
Substituting this into Eq.~(\ref{Zo4}) and using the fact that $\pi^2/8>1$ completes the proof.
\end{proof}

\begin{figure}[htbp]
\centerline{\includegraphics[width=10cm]{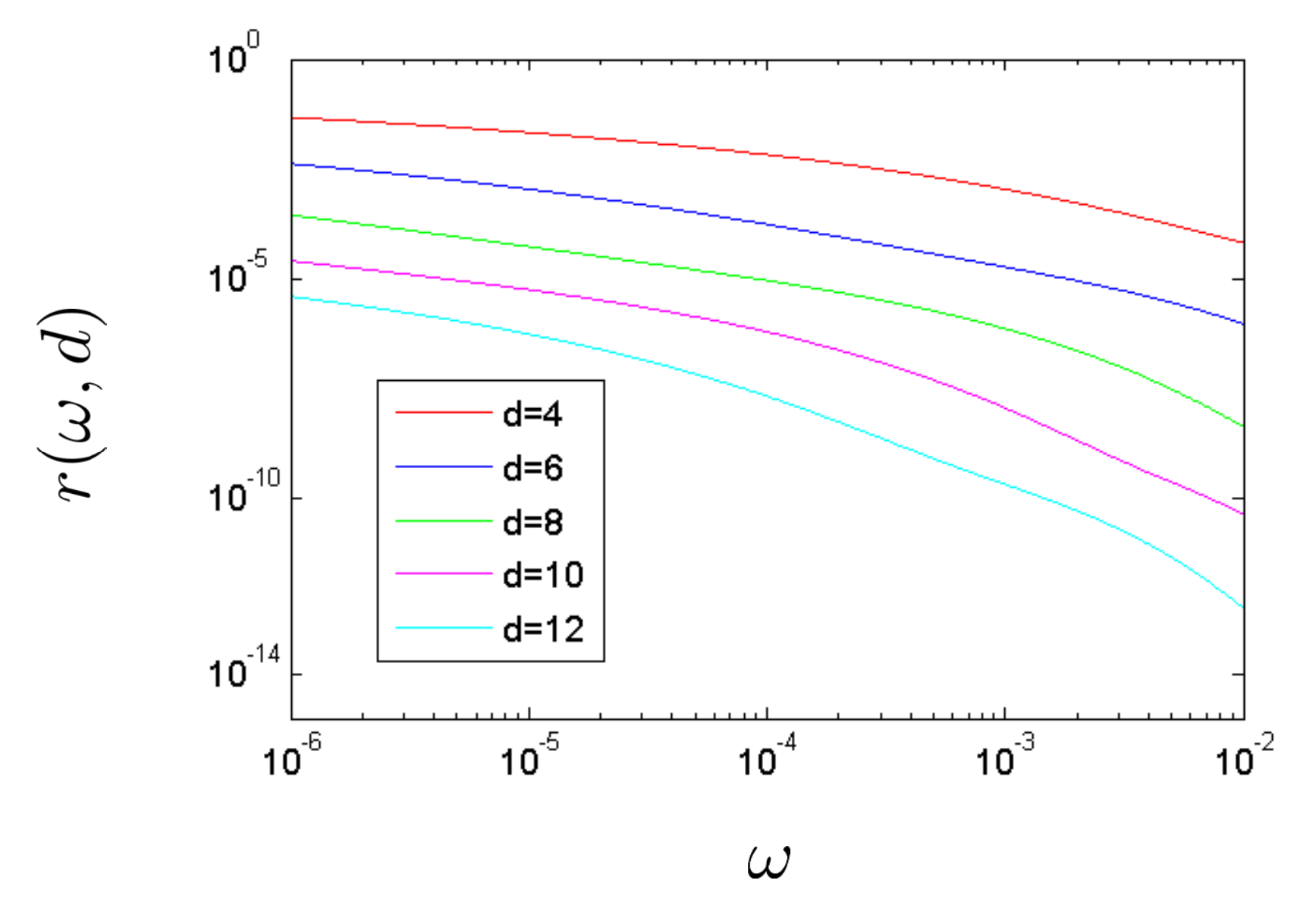}}
\caption{Approximating $\sqrt{x}$ on the interval
$\omega\le x\le 1$ by Zolotarev's rational function $xP_d(x)Q_d^{-1}(x)$
with degree-$d$ polynomials $P_d(x),Q_d(x)$.
The plot shows the worst-case relative approximation error defined in Eq.~(\ref{relative_error}).
}
\label{fig:plot1}
\end{figure}

Next let us describe an explicit construction of the polynomials $P_d(x),Q_d(x)$
from Lemma~\ref{lemma:Zo}.
This material is based on 
Refs.~\cite{zolotarev1877application,gonvcar1969zolotarev,nakatsukasa2015computing}.
Below we assume that $0<\omega\le 1$. Define
\[
\mu\equiv \sqrt{1-\omega}.
\]
We shall  need  Jacobi elliptic functions $\sn{u}{\mu}$ and $\cn{u}{\mu}$. They are defined
for $u\ge 0$  by
\[
\sn{u}{\mu}=\sin{(\phi(u))} \quad  \mbox{and} \quad  \cn{u}{\mu}=\cos{(\phi(u))},
\]
where $\phi(u)\ge 0$ is the unique solution  of
\[
u=\int_0^{\phi(u)}  \frac{d\theta}{\sqrt{1-\mu^2 \sin^2{(\theta)}}}.
\]
For each $j=1,\ldots,2d$ define
\begin{equation}
\label{roots}
\lambda_j=\omega \left[  \frac{\sn{\frac{j K(\mu)}{2d+1}}{\mu}}{\cn{\frac{j K(\mu)}{2d+1}}{\mu}} \right]^2.
\end{equation}
Define
\[
M=2\left[ \prod_{j=1}^{2d} \frac{(1+\lambda_{2j})}{(1+\lambda_{2j-1})} + 
\sqrt{\omega} \prod_{j=1}^{2d} \frac{(\omega+\lambda_{2j})}{(\omega+\lambda_{2j-1})}  \right]^{-1}
\]
Then
\begin{equation}
\label{PQd}
P_d(x)=M \prod_{j=1}^{d} (x+\lambda_{2j})
\quad \mbox{and} \quad Q_d(x)=\prod_{j=1}^{d} (x+\lambda_{2j-1}).
\end{equation}
The coefficients $\lambda_j$ can be easily  computed using any computer
algebra system such as  MATLAB. 
We plot the worst-case relative error 
\begin{equation}
\label{relative_error}
r(\omega,d)\equiv \max_{\omega\le x\le 1} x^{-1/2}\left| x^{1/2} - xP_d(x) Q_d^{-1}(x) \right|
\end{equation}
as a function of $\omega$ for a few small values of $d$ on Fig.~\ref{fig:plot1}.

Finally, let us point out that  $\sqrt{x}$ can  be approximated by a rational function
of a given degree $d$ on the {\em full} interval $[0,1]$. However,   in this 
case  the 
approximation error decays exponentially with $\sqrt{d}$.
A simple example of such approximation has been proposed by Newman~\cite{newman1964rational}
 (although Ref.~\cite{newman1964rational} concerns with approximating the $|x|$ function, a simple
change of variable shows that the same results hold for $\sqrt{x}$).
In contrast, Zolotarev's approximation holds only on the interval $\omega\le x\le 1$ but
the approximation error decays exponentially with $d$. 
Ref.~\cite{newman1964rational}   also demonstrates that 
rational approximations can be exponentially more accurate compared with 
the polynomial approximations of the same degree.

\subsection{Approximation by low-rank Gaussian states}
\label{subs:approx}

In this section we derive two corollaries of  Theorem~\ref{thm:2} that characterize 
ground states of quantum impurity models.  
Let $\psi$ be some ground state of the full Hamiltonian $H$
and $C$ be the covariance matrix of $\psi$ defined in  
 Theorem~\ref{thm:2}.
It is not hard to see that, roughly speaking, eigenvalues of $C$ are related to the number of bath excitations present in $\psi$.  For example, $\mathrm{Tr}(C)$ is the expected value of the number operator $\sum_{j}b^{\dagger}_j b_j$ in the state $\psi$.  In the extreme case when $\mathrm{Tr}(C)=0$ the ground state $\psi$ itself is
a Gaussian state (the ground state of $H_0$). More generally, the exponential decay stated in Theorem~\ref{thm:2} is sufficient to show that $\psi$ is approximated by a relatively small number $\chi\ll 2^n$ of Gaussian states.
Furthermore, all excitations present in the bath can be ``localized" on a small subset of modes
by some Gaussian unitary operator. 
Recall that $\omega$ denotes the spectral gap of $H_0$.
\begin{corol}
\label{cor:gapped}
There exists a ground state $\psi$ of $H$ such that the following holds. For any $\delta>0$ there exists an integer 
\begin{equation}
\label{chi}
k=28m \log{(2\omega^{-1})}
\left[ \log{(\delta^{-1})} + \log{(m)} + \log{\log{(2\omega^{-1})}} +O(1)\right],
\end{equation}
a Gaussian  unitary operator $U$, and 
some (non-Gaussian) state $\phi\in \calH_k$  such that 
\begin{equation}
\label{main}
\| \psi - U|\phi\otimes 0^{n-k}\rangle \| \le \delta.
\end{equation}
\end{corol}
\begin{proof}Let $\psi$ and $C$ be the state and corresponding matrix from Theorem~\ref{thm:2}.
 Let $V$ be a unitary operator such that 
\[
V^\dag CV  =\sum_{j=1}^n \sigma_j |j\rangle\langle j|, \quad  \quad \sigma_1\ge \sigma_2 \ge \ldots \ge \sigma_n.
\]
Define a new  set of fermi modes $\tilde{b}_1,\ldots,\tilde{b}_n$ such that 
$\tilde{b}_j=\sum_{k=1}^n V_{k,j} b_k$.
Then 
\begin{equation}
\label{diagonal1}
\langle \psi| \tilde{b}_i^\dag \tilde{b}_j |\psi\rangle = \sigma_i \delta_{i,j}.
\end{equation}
Define a projector $\Pi_j=\tilde{b}_j \tilde{b}_j^\dag$ so that
$\langle \psi |\Pi_j|\psi\rangle =1-\sigma_j$.
Then $\|(I-\Pi_j)\psi\|=\sqrt{\sigma_j}$ and thus
\[
\| \psi - \Pi_n \Pi_{n-1} \cdots \Pi_{k+1}\psi \| \le
\sum_{j=k+1}^n \sqrt{\sigma_j} \leq O(1) \cdot m\log{(2\omega^{-1})}
\exp{\left[ -\frac{k}{28m\log{(2\omega^{-1})}}\right]}
\]
where in the last inequality we used Theorem~\ref{thm:2}. Thus
\[
\| \psi -\Pi_n \Pi_{n-1} \cdots \Pi_{k+1}\psi \| \le \delta/2
\]
for $k$ given by Eq.~\eqref{chi}. Choose 
\[
|\psi'\rangle=\gamma \Pi_n \Pi_{n-1} \cdots \Pi_{k+1} |\psi\rangle,
\]
where $\gamma$ is the normalizing coefficient. The above shows that $\gamma^{-1}\ge 1-\delta/2$.
Thus $\psi'$ approximates $\psi$ within error $\delta$. Finally,
since $\Pi_j$ projects onto a subspace in which the mode $\tilde{b}_j$ is empty,
the state  $\psi'$ is the vacuum state for the subset of  $n-k$  fermi modes
$\tilde{b}_n,\ldots,\tilde{b}_{k+1}$.
Therefore $|\psi'\rangle= U |\phi\otimes 0^{n-k}\rangle$ for some Gaussian
unitary operator $U$ and some (non-Gaussian) state $\phi\in \calH_k$. 
\end{proof}

From Eq.~(\ref{main}) one infers that $\psi$ is $\delta$-close to a superposition
of $\chi=2^k$ Gaussian states where $k$ is defined by Eq.~(\ref{chi}).
The next corollary states our best asymptotic upper bound on  the number
of Gaussian states one needs to approximate an exact ground state within 
a specified precision. This corollary provides a partial justification for the
variational algorithm of Section~\ref{sec:variational} that minimizes
the energy of $H$ over low-rank superpositions of Gaussian states.
It also confirms our intuition (stated above) that eigenvalues of $C$ are related to the number of bath excitations in the ground state.
\begin{corol}
There exists a ground state $\psi$ of $H$ such that the following holds. For any $\delta\in (0,1/2]$ there exists a 
normalized state $\phi\in \calH_n$ such that $\|\phi-\psi\|\leq \delta$ and
\begin{equation}
|\phi\rangle=\sum_{a=1}^{\chi} z_a |\theta_a\rangle, \qquad \chi=e^{O(m\log(\omega^{-1})\log(\delta^{-1}))}
\label{eq:thetastates}
\end{equation}
where $\theta_1,\ldots,\theta_\chi\in \calG_n$ are orthonormal Gaussian states.
Each state $\theta_a$ is an eigenvector of the particle number operator
$N =\sum_{j=1}^n b_j^\dag b_j$ such that 
\begin{equation}
\label{particle_number_bound}
N |\theta_a\rangle = k_a|\theta_a\rangle,
\qquad k_a\le cm\log{(2\omega^{-1})} \log{(\delta^{-1})}.
\end{equation} 
Here $c>0$ is a universal constant.
Furthermore, the projector 
\begin{equation}
\label{Ptheta}
P\equiv \sum_{a=1}^\chi |\theta_a\rangle\langle \theta_a|
\end{equation}
commutes with $H_{imp}$.
\label{cor:gaussianrank}
\end{corol}

\begin{proof}
We shall reuse some definitions and notation from the proof of Theorem~\ref{thm:2} 
(repeated here for convenience). For any complex vector $x\in \CC^n$ define a fermionic operator
\[
b(x)\equiv \sum_{j=1}^n x_j b_j.
\]
Each canonical bath mode $b_j$ can be written as a linear combination 
of Majorana modes $c_1,\ldots,c_{2n}$. 
Recall that the impurity is formed
by the first $m$ modes $c_1,\ldots,c_m$.  
Define a linear subspace $\calL\subseteq \CC^n$ such that $x\in \calL$
iff the expansion of $b(x)$ in terms of the Majorana operators does not include $c_1,\ldots,c_m$. 

As noted in Eqs.~(\ref{commute1},\ref{dimL}) we have $[H_{imp},b(x)]=0$ for all $ x\in \calL$ and $\dim(\mathcal{L})\geq n-m$.
Define $\Lambda$ to be the projector onto $\mathcal{L}$. 
Let $C$ and $\psi$ be as in Theorem~\ref{thm:2}. Write
\[
\lambda_1 \geq \lambda_2\geq \ldots \lambda_n\geq 0
\]
for the eigenvalues of $\Lambda C \Lambda$. By Cauchy's interlacing theorem and Theorem~\ref{thm:2} we have 
\[
\lambda_i \leq \sigma_i \leq c\exp{\left[  -\frac{j}{14m\log{(2\omega^{-1})}}\right]}.
\]
where $\sigma_i$ is the $i$-th largest  eigenvalue of $C$. Let $\{x^{j}:j\in [n]\}$ be an orthonormal basis of eigenvectors of $\Lambda C \Lambda$, with
\[
\mathrm{span}\{x^1,\ldots,x^L\}=\mathcal{L} \qquad L=\dim{\mathcal{L}}.
\]
and $\Lambda C \Lambda x^{j}=\lambda_j x^{j}$. Define a new set of  creation-annihilation operators $\hat{b}^\dagger_j,\hat{b}_j$ such that
\[
\hat{b}_j=b(x^{j})=\sum_{k=1}^{n} x^{j}_k b_k \qquad j=1,2,\ldots, n.
\]
We have
\begin{equation}
[\hat{b}_j,H_{imp}]=0 \qquad j=1,2,\ldots ,L.
\label{eq:hatcommute}
\end{equation}
Now divide up the set of positive integers into intervals of size $\sim 14m\log(2\omega^{-1})$. That is, let $Q=\lceil 14m\log(2\omega^{-1})\rceil$ and define
\begin{align}
I_1&=\{1,2,\ldots,Q\}\nonumber\\
I_2&=\{Q+1,\ldots,2Q\}\nonumber\\
I_3&=\{2Q+1,\ldots,3Q\}\nonumber\\
\vdots\nonumber
\end{align}
Let $[L]\equiv \{1,2,\ldots,L\}$.
For each positive integer $s$ define
\[
N_s=\sum_{j\in I_s\cap [L]} \hat{b}^{\dagger}_j \hat{b}_j.
\]
There are $\sim L/Q$ of these operators which are nonzero. Also note that Eq.~\eqref{eq:hatcommute} implies
\begin{equation}
[N_s,H_{imp}]=0 
\label{eq:commuteN}
\end{equation}
for all $s$. We have 
\begin{equation}
\langle \psi|N_k |\psi\rangle=\sum_{j\in I_k\cap [L]} \lambda_j\leq c_0 m\log(2\omega^{-1}) e^{-(k-1)}
\label{eq:meank}
\end{equation}
where $c_0$ is another universal constant.
Since the operators $\{N_k\}$ mutually commute, we can view $\psi$ as inducing a probability distribution $\Psi$ given by
\[
\Psi(n_1,n_2\ldots)=\langle \psi |M_{n_1,n_2,\ldots}|\psi\rangle
\]
where $M_{n_1,n_2,\ldots}$ is the projector onto Fock basis states with simultaneous eigenvalues $n_1,n_2,\ldots$ for operators $N_1,N_2,\ldots$.
Define
\[
R_k=c_0m\log(2\omega^{-1})e^{-(k-1)/2} \quad \text{and} \quad p_k=\frac{R_k}{Q}\leq c_1 e^{-(k-1)/2}
\]
where $c_1$ is another universal constant. From Markov's inequality and Eq.~\eqref{eq:meank} we obtain
\[
\mathrm{Pr}_{\Psi}[n_k \geq R_k]\leq e^{-(k-1)/2}.
\]
Applying a union bound we arrive at
\[
\mathrm{Pr}_{\Psi}[n_k \leq R_k \text{ for all } k\geq J]\geq 1- \sum_{j\geq J} e^{-(j-1)/2}
\]
Let us choose $J$ so that
\begin{equation}
\sum_{j\geq J} e^{-(j-1)/2}\leq \delta^{2}/2 \quad \text{and} \quad p_k<\frac{1}{2} \text{ for all } k\geq J,
\label{eq:sumj}
\end{equation}
i.e., $J=O(\log(\delta^{-1}))$. Then
\[
\mathrm{Pr}_{\Psi}[n_k \leq R_k \text{ for all } k\geq J]\geq 1- \delta^{2}/2.
\]
Given a binary string $\theta \in \{0,1\}^n$ let $|\theta\rangle \in \calH_n$ be 
the Fock basis vector with respect to the modes  $\hat{b}_1,\ldots,\hat{b}_n$.
In other words, $\hat{b}_j^\dagger \hat{b}_j |\theta\rangle =\theta_j |\theta\rangle$
for all $j$. 
Define a projector $P$ onto a linear subspace 
\begin{equation}
\mathrm{span}\left( |\theta\rangle \, : \, 
N_k |\theta\rangle = n_k|\theta\rangle, \qquad n_k\le R_k  \qquad \mbox{for $J\le k \le n$} \right).
\label{eq:Pspace}
\end{equation}
From the above we have $\langle \psi |P |\psi\rangle \geq 1-\delta^2/2$. Defining 
$|\phi\rangle=\|P|\psi\rangle\|^{-1}P|\psi\rangle$ we therefore have
\[
\|\psi-\phi\|^2\leq 2(1-\sqrt{1-\delta^2/2})\leq \delta^2.
\]
Thus $\psi$ is $\delta$-close to $\phi$. We now show that $P$ projects onto a space of dimension upper bounded as
\[
e^{O(m\log(\omega^{-1})\log(\delta^{-1}))}.
\]

Fock basis states $|\theta\rangle$ in the image of $P$ are indexed by bit strings of length $n$ where the last $n-L$ bits are unrestricted while the first $L$ bits are divided up into the intervals $I_s\cap [L]$. 
The hamming weight of $\theta$  in interval $I_s$ is required to be at most $R_s$ when $s\geq J$. Letting $F(m,l)$ denote the number of bit strings of length $m$ with hamming weight at most $l$ we obtain
\[
\mathrm{Tr}(P)\leq 2^{(n-L)} 2^{Q(J-1)}\prod_{k\geq J} F(Q,R_k)\leq 2^{(m+Q(J-1))}\prod_{k\geq J} F(Q,R_k),
\]
where in the second inequality we used the fact that $L\geq n-m$. Recall that $p_k=R_k/Q$.
Since $p_k<1/2$ for all $k\geq J$ we may bound
\[
F(Q,R_k)=\sum_{s=0}^{R_k}\binom{Q}{s}\leq 2^{H(p_k)Q} \quad \text{for all } k\geq J
\]
where $H(p)=-p\log_2(p)-(1-p)\log_2(1-p)$ is the binary entropy function.
Therefore 
\begin{equation}
\mathrm{Tr}(P) \leq 2^{m+Q(J-1)+Q\sum_{k\geq J}H(p_k)}
\label{eq:basis}
\end{equation}
Now using the inequality $H(x)\leq 2\sqrt{x}$ we get 
\[
\sum_{k\geq J} H(p_k)\leq 2\sqrt{c_1}\sum_{k\geq J}e^{-(k-1)/4}=O(\delta).
\]
Plugging this into Eq.~\eqref{eq:basis} and using the fact that $Q=O(m\log(\omega^{-1}))$ and $J=O(\log(\delta^{-1}))$ gives 
\[
\mathrm{Tr}(P)=e^{O(m\log(\omega^{-1})\log(\delta^{-1}))}
\]
which completes the proof of Eq.~\eqref{eq:thetastates}. Here we take $\chi=\mathrm{Tr}(P)$ and $\{|\theta_a\rangle\}$ to be the set of Fock basis states which span the image of $P$, so that $P=\sum_{a=1}^{\chi} |\theta_a\rangle\langle \theta_a|$.

Using the fact that $\{x^i: i=1,2\ldots,n\}$ is an orthonormal basis of $\mathbb{C}^n$ we get
\[
\hat{N}\equiv \sum_{j=1}^{n} \hat{b}_j^{\dagger}\hat{b}_j=\sum_{j,f,g=1}^{n} \bar{x}^{j}_f x^{j}_g b_f^{\dagger} b_g=\sum_{j=1}^{n} b_j^{\dagger} b_j\equiv N.
\]
By construction, $\hat{N} |\theta_a\rangle=k_a|\theta_a\rangle$, where
\[
\kappa_a \leq m+Q(J-1)+\sum_{k\geq J} R_k \leq cm\log(2\omega^{-1})\log(\delta^{-1})
\]
for some universal constant $c>0$. This implies $N|\theta_a\rangle=k_a|\theta_a\rangle$,
as claimed in Eq.~(\ref{particle_number_bound}). 
Furthermore, $[P,N]=[P,\hat{N}]=0$ since $P$ and $\hat{N}$ are both diagonal in the Fock basis
$\{ |\theta\rangle\}$. We also have $[P,H_{imp}]=0$ which follows from Eqs.~(\ref{eq:commuteN},\ref{eq:Pspace}).
\end{proof}

For future reference, let us summarize properties of $P$ that follow trivially
from Corollary~\ref{cor:gaussianrank}:
\begin{equation}
[P,H_{imp}]=[P,\sum_{j=1}^{n} b^{\dagger}_j b_j]=0 \quad \text{ and } \quad \|\sum_{j=1}^{n} b^{\dagger}_j b_j P\|\leq cm\log(2\omega^{-1})\log(\delta^{-1}).
\label{eq:Pcond}
\end{equation}

\subsection{Energy distribution}
\label{subs:Arad2proof}

Suppose $\psi$ is a ground state of the full Hamiltonian $H=H_0+H_{imp}$.
In general, $\psi$ is not a ground state  of the bath Hamiltonian $H_0$.  However,
Corollary~\ref{cor:gaussianrank} of 
Theorem~\ref{thm:2} implies that the average number of excitations in the bath
is small as long as $H_0$ has a non-negligible spectral gap. 
Can we make a similar statement without any assumptions on the gap of $H_0$ ?
In this section we prove a partial result along these lines that  concerns
the {\em energy} of excitations present in the bath.
Namely,  we show that any ground state of $H$ has most of its weight on a certain low-energy
subspace of $H_0$. Surprisingly, the energy cutoff that defines this low-energy
subspace is independent of the strength of $H_{imp}$. The following theorem was inspired by an analogous theorem proved 
for quantum spin systems by Arad et al~\cite{AKL14}.

\Arad*

\begin{proof}
Let us first establish some additional notation and conventions. Recall 
\[
H=H_0+H_{imp} \quad 
\]
where $H_0$ is given by Eq.~\eqref{bath1} for some antisymmetric real matrix $h$. 
We shall assume without loss of generality that $H_{imp}$ has smallest eigenvalue equal to zero. 
Indeed, we may always add a constant times the identity to $H_{imp}$ to ensure that this holds, and the statement of the theorem is identical before and after such a transformation. Without loss of generality we can also assume that
\begin{equation}
\label{hspecial}
h_{p,q}=0 \quad \mbox{for $1\le p\le m$ and $2m<q\le 2n$},
\end{equation}
that is, the impurity modes $c_1,\ldots,c_m$ are only coupled to the modes $c_{m+1},\ldots,c_{2m}$. 
Indeed, suppose $h$ is arbitrary.  Let $g$ be the submatrix of $h$ formed by the first $m$ columns and the last $2n-m$ rows.
Choose an orthogonal matrix $R'$ of size $2n-m$ such that $R'g$ is an upper triangular matrix. Let $R=I_m\oplus R'$. 
Then $RhR^T$ has zero matrix elements for $1\le p\le m$ and $2m<q\le 2n$.
Let $U\in \calC_n$ be the Gaussian unitary operator such that $Uc_p U^\dag=\sum_q R_{p,q} c_q$,
see Eq.~(\ref{UvsR}). By definition, $U$ acts trivially on the impurity modes $c_1,\ldots,c_m$
and the transformed bath Hamiltonian $UH_0U^\dag$ has a form Eq.~(\ref{bath1})
where $h$ satisfies Eq.~(\ref{hspecial}).

We shall often work  with a truncated version of the bath Hamiltonian 
defined as
\begin{equation}
\label{Hbath}
H_{bath}=\frac{\|\tilde{h}\|_1}{4}I + \frac{i}4 \sum_{p,q=1}^{2n} \tilde{h}_{p,q} c_p c_q,
\end{equation}
where $\tilde{h}$ is a matrix obtained from $h$ by setting to zero the first $m$ rows and columns.
By definition, $H_{bath}$ acts trivially on the impurity modes and has ground energy zero.

Finally, define a Hamiltonian
\[
H_{mix}=H_{0}-H_{bath}.
\]
From Eq.~(\ref{hspecial}) one infers that $H_{mix}$ acts non-trivially only 
on modes $c_1,\ldots,c_{2m}$. Its norm is bounded as
\begin{equation}
\label{mix_norm}
\| H_{mix}\|= \frac14 \| h-\tilde{h}\|_1+\frac14 \big|\|\tilde{h}\|_1-\|h\|_1\big| \le \frac12 \| h-\tilde{h}\|_1\leq  2m
\end{equation}
where we used the triangle inequality, the fact that $h-\tilde{h}$ has rank at most $2m$, and $\| h-\tilde{h}\|\le \|h\| +\| \tilde{h}\|\le 2$. 

The following lemma asserts that any 
ground state of the full Hamiltonian $H$ 
has most of its weight on a properly chosen low-energy subspace of $H_{bath}$
that depends only on $m$.
\begin{lemma}
\label{lemma:Arad1}
Let $P_\tau$ be the projector onto a subspace spanned by eigenvectors of $H_{bath}$
with energy at most $\tau$. 
Let $\psi$ be any ground state of $H$. 
Then
\begin{equation}
\label{low_energy1}
\| (I- P_{\tau})\psi\| \le \exp{\left[ -\frac{\tau}{2} \log{\left(\tau/4em\right)}  \right]}.
\end{equation}
Furthermore,  for all $\tau\ge 8em$ one has
\begin{equation}
0\le \frac{\langle \psi |P_{\tau} H P_{\tau}|\psi\rangle}{\langle \psi|P_{\tau}|\psi\rangle}
-\langle \psi|H|\psi\rangle
\leq
4m \exp{\left[-\frac{\tau}{2}\log{\left(\frac{\tau}{4em}\right)}\right]} 
\label{low_energy2}
\end{equation}
Here $e\equiv \exp{(1)}$.
\end{lemma}
\begin{proof}
Define a function
\[
M(s)=\langle \psi |e^{2s H_{bath}}|\psi\rangle = \| e^{sH_{bath}} \psi\|^2.
\]
Then 
\begin{equation}
\label{Arad5}
\| (I- P_{\tau})\psi\|^2 =\| e^{-sH_{bath}} (I- P_{\tau})e^{sH_{bath}} \psi\|^2 \le e^{-2s \tau } M(s).
\end{equation}
Let $e_g=\langle\psi |H|\psi\rangle$ be the ground energy of the full Hamiltonian.
Note that $e_g\ge 0$ since by assumption $H_{imp}\ge 0$ and $H_0\ge 0$. 
Computing the derivative of $M(s)$ over $s$ one gets 
\[
\dot{M}(s) = 2\langle \psi| e^{2sH_{bath}} H_{bath}| \psi\rangle
=2\langle \psi| e^{2sH_{bath} } (e_g I - H_{imp}-H_{mix})|\psi \rangle.
\]
By construction, $H_{imp}$ commutes with $H_{bath}$,
so that $e^{2sH_{bath}}H_{imp}$ is hermitian.
Furthermore, 
since we assumed that $H_{imp}\ge 0$,  one has 
$e^{2sH_{bath}} H_{imp}\ge 0$. We arrive at
\begin{equation}
\label{Arad6}
\dot{M}(s) \le 2e_g M(s) - \langle \psi|G(s)|\psi\rangle, 
\end{equation}
where
\[
G(s)=e^{2sH_{bath}} H_{mix}+ H_{mix} e^{2sH_{bath}}.
\]
Let $H_{mix}(u)\equiv e^{uH_{bath} } H_{mix} e^{-uH_{bath} }$. Then
\begin{equation}
\label{G(s)}
G(s)=e^{sH_{bath}} \left( H_{mix}(s)+H_{mix}(-s) \right) e^{sH_{bath} }.
\end{equation}
We claim that 
\begin{equation}
\label{Hmix(s)}
\| H_{mix}(s)\| \le me^{2|s|}
\end{equation}
for all $s\in \RR$. Indeed, since $H_{mix}(-s)=H_{mix}(s)^\dag$, it suffices
to consider $s\ge 0$.
By definition of $H_{mix}$ one has 
\[
H_{mix}(s)=\frac14 (\|h\|_1-\|\tilde{h}\|_1)+\frac{i}4 \sum_{p,q=1}^{2n} (h-\tilde{h})_{p,q} \, c_p(s) c_q(s), \quad \mbox{where} \quad c_p(s)\equiv e^{sH_{bath}} c_p e^{-sH_{bath}}.
\]
The fact that $H_{bath}$ is quadratic in Majorana operators implies
$c_p(s)=\sum_{q=1}^{2n} R_{p,q}(s) c_q$, where
$R=e^{i\tilde{h}s}$ is a hermitian matrix of size $2n$ (recall that $h$ and $\tilde{h}$ are anti-hermitian).
Therefore
\[
H_{mix}(s)=\frac14 (\|h\|_1-\|\tilde{h}\|_1)+\frac{i}4 \sum_{p,q=1}^{2n} g_{p,q} c_p c_q, \quad \mbox{where} \quad g=R^T (h-\tilde{h})R.
\]
This implies 
\begin{equation}
\label{Hmix(s)}
\|H_{mix}(s)\| =\frac14 (\big|\|h\|_1-\|\tilde{h}\|_1\big|+\|g\|_1) \le\frac14 (\|h-\tilde{h}\|_1+\|g\|_1) \le \frac{m}{2}(1+e^{2\| h\| s})\|h-\tilde{h}\| \le m(1+e^{2s})
\end{equation}
where in the second inequality we recalled that $h-\tilde{h}$ has rank at most $2m$.

From  Eq.~(\ref{Hmix(s)}) we  infer  that  $H_{mix}(s)+H_{mix}(-s) \ge -2m(1+e^{2s})I$
and thus
\begin{equation}
\label{Arad7}
G(s)\ge - 2m(1+e^{2s})e^{2sH_{bath}}.
\end{equation}
Substituting this into Eq.~(\ref{Arad6}) yields 
\begin{equation}
\label{Arad8}
\dot{M}(s) \le 2e_g M(s) + 2m(1+e^{2s}) M(s).
\end{equation}
Since $e_g$ is a ground energy of the full Hamiltonian,
it is upper bounded by the energy of 
a tensor product of ground states of $H_{imp}$ and $H_{bath}$.
By assumption, $H_{imp}$ and $H_{bath}$ have zero ground energy, so
\[
e_g\le \|H_{mix}\|\le 2m,
\]
see Eq.~(\ref{mix_norm}), that is,
\begin{equation}
\label{Arad9}
\dot{M}(s) \le 6mM(s)+2me^{2s}M(s) \le 8me^{2s} M(s).
\end{equation}
Since $M(0)=1$, one gets 
\begin{equation}
\label{Arad10}
\log{M(s)} \le 4me^{2s} -4m\le 4me^{2s}.
\end{equation}
Substituting this into Eq.~(\ref{Arad5}) and choosing $e^{2s}=\tau/4m$ gives
\begin{equation}
\label{Arad11}
\| (I- P_{\tau})\psi\| \le e^{-s\tau + 2me^{2s} } \le \exp{\left[ -\frac{\tau}{2} \log{(\tau/4em)}  \right]}
\end{equation}
which proves Eq.~(\ref{low_energy1}).

We proceed to proving Eq.~(\ref{low_energy2}).
The first inequality in Eq.~(\ref{low_energy2})  simply states that the energy
of a normalized state $P_\tau \psi$ cannot be smaller than the ground energy of $H$.
Thus it suffices to prove the second inequality. 
Since $P_\tau$ commutes with $H_{imp}$ and $H_{bath}$, one has 
\[
\langle \psi |(I-P_\tau)HP_\tau|\psi\rangle = \langle \psi |(I-P_\tau)H_{mix} P_\tau|\psi\rangle.
\]
Using the above identity and the eigenvalue equation
$H\psi =e_g \psi$ one gets
\[
\left|e_g -\frac{\langle \psi |P_{\tau} H P_{\tau}|\psi\rangle}{\langle \psi|P_{\tau}|\psi\rangle}\right|
=\frac1{\langle \psi|P_{\tau}|\psi\rangle} \left|  \langle \psi |(I-P_\tau)H_{mix} P_\tau|\psi\rangle \right|.
\]
Recalling that $\|H_{mix}\|\le 2m$, see Eq.~(\ref{mix_norm})
and using Eq.~(\ref{low_energy1}) leads to
\[
\left|e_g -\frac{\langle \psi |P_{\tau} H P_{\tau}|\psi\rangle}{\langle \psi|P_{\tau}|\psi\rangle}\right|
\le \frac{2m}{\langle \psi|P_{\tau}|\psi\rangle} \| (I-P_\tau)\psi\| \le 4m 
\exp{\left[ -\frac{\tau}{2} \log{(\tau/4em)}  \right]}.
\]
Here we noted that $\langle \psi|P_{\tau}|\psi\rangle\ge 1/2$ for $\tau\ge 8em$ due to Eq.~(\ref{low_energy1}).
\end{proof}

We are now ready to complete the proof of Theorem~\ref{thm:Arad2}.
Let us first show that   for all $\tau\ge m$ one has
\begin{equation}
\label{PP'}
\| (I-Q_{2\tau})P_{\tau}\|\le \exp{\left[ -\frac{\tau}2 \log{\left(\frac{\tau}{2em}\right)} \right]}.
\end{equation}
Indeed, 
from $I-Q_{2\tau}=e^{-sH_0} (I-Q_{2\tau})e^{sH_0}$ 
and $e^{-sH_0} (I-Q_{2\tau})\le e^{-2s\tau} I$ one gets
\begin{equation}
\label{prop1bound}
\| (I-Q_{2\tau})P_{\tau}\|\le  e^{-2s\tau} \cdot K(s), \quad \mbox{where} \quad K(s)\equiv \| e^{sH_0} P_{\tau}\|.
\end{equation}
Let $\dot{K}(s)=\partial K(s)/\partial s$. Then
\begin{equation}
\label{Kdot1}
\dot{K}(s)\le \| e^{sH_0}  H_0 P_{\tau}\| =\| e^{sH_0}(H_{bath} + H_{mix})P_{\tau}\|
\le \| e^{sH_0} H_{bath} P_{\tau}\| + \| e^{sH_0} H_{mix}P_{\tau}\|.
\end{equation}
We may bound the first term in the righthand side of Eq.~(\ref{Kdot1}) as
\begin{equation}
\label{Kdot2}
\| e^{sH_0} H_{bath} P_{\tau}\| =\| e^{sH_0} P_{\tau} H_{bath} P_{\tau}\|
\le \| H_{bath} P_{\tau}\|\cdot K(s)  \le \tau K(s).
\end{equation}
Denote  $H_{mix}(s)=e^{sH_0} H_{mix} e^{-sH_0}$. Then Eq.~(\ref{Kdot1}) becomes
\begin{equation}
\label{Kdot3}
\dot{K}(s)\le \tau K(s)  + \| H_{mix}(s)\| \cdot K(s).
\end{equation}
The same arguments as in the proof of Lemma~\ref{lemma:Arad1} 
show that 
$\|H_{mix}(s)\| \le 2me^{2s}$.
Therefore 
\[
\dot{K}(s)\le (\tau + 2me^{2s})K(s).
\]
Since $K(0)=1$, this yields
$\log{K(s)} \le \tau s +me^{2s}$.
Substituting this into Eq.~(\ref{prop1bound}) and choosing
$e^{2s}=\tau/(2m)$ one gets Eq.~(\ref{PP'}).

Finally, we will show  that for all $\tau\ge 4em$
\begin{equation}
\label{Arad2}
\| (I-Q_{2\tau})\psi \| \le 2  \exp{\left[-\frac{\tau}2 \log\left(\frac{\tau}{4em}\right)\right]},
\end{equation}
which is equivalent to Eq.~(\ref{low_energy3}).
Inserting the identity decomposition $I=P_{\tau} + (I-P_{\tau})$ between
$I-Q_{2\tau}$  and $\psi$ in Eq.~(\ref{Arad2}) one gets
\begin{equation}
\label{Arad3}
\| (I-Q_{2\tau})\psi \| \le \| (I-Q_{2\tau}) P_{\tau}\psi\| +  \| (I-Q_{2\tau})(I- P_{\tau})\psi\|.
\end{equation}
Bounding the first term in Eq.~(\ref{Arad3}) using Eq.~(\ref{PP'})
and the second term using Eq.~(\ref{low_energy1}) gives 
\begin{equation}
\label{Arad4}
\| (I-Q_{2\tau})\psi \| \le  \exp{\left[ -\frac{\tau}2 \log{(\tau/2em)} \right]} + \exp{\left[ -\frac{\tau}{2} \log{\left(\tau/4em\right)}  \right]}.
\end{equation}
This completes the proof of Theorem~\ref{thm:Arad2}.
\end{proof}
%%%%%%%%%%%%%%%%%%%%%%%%%%%%%%%%%%%%%%%%%%%%%%%%%%%%%%%%%%%%%%%%%%%%
%%%%%%%%%%%%%%%%%%%%%%%%%%%%%%%%%%%%%%%%%%%%%%%%%%%%%%%%%%%%%%%%%%%%

\section{Algorithms and complexity}
\label{sec:energy}

In this section we consider the problem of approximating the ground energy of a quantum impurity model.  We begin in Section \ref{sec:decouplingtools} with some technical tools used in subsequent sections: a \textit{decoupling lemma} which describes how to identify and decouple canonical bath modes which commute with the impurity, and a \textit{truncation lemma} which shows that the ground energy can change by at most $\epsilon$ if we take all single-particle energies of $H_0$ below $\epsilon/m$ and round them up to $\epsilon/m$. In Section \ref{sec:quasipolynomial} we present and analyze our quasipolynomial algorithm for approximating the ground energy of a quantum impurity model (that is, we prove 
Theorem~\ref{thm:1}). In Section \ref{sec:quasipolynomial} we consider the case where the full Hamiltonian $H$ has a constant spectral gap and in this case we give a polynomial time algorithm (proving Theorem~\ref{thm:gap}). Finally, in Section \ref{sec:QCMA} we present the proof of Theorem~\ref{thm:QCMA}, that is, we show that the problem of estimating the ground energy of a quantum impurity model to inverse polynomial precision is contained in the complexity class QCMA.

\subsection{Bath decoupling and truncation}
\label{sec:decouplingtools}
In this Section we present two tools for simplifying quantum impurity problems. 

The first tool is a decoupling lemma which is useful when the single-particle spectrum $\{\epsilon_j: j\in [n]\}$ contains degeneracy. We say that a fermi mode $\tilde{b}_j$ is coupled to the impurity if $[\tilde{b}_j,H_{imp}]\neq 0$. The decoupling lemma states that we may choose a set of fermi modes which diagonalize the bath Hamiltonian $H_0$ and such that at most $m$ of the modes with a given single-particle energy $\epsilon_j$ are coupled to the impurity.  

\begin{lemma}[\textbf{Decoupling Lemma}]
Let $H=H_0+H_{imp}$ be a quantum impurity model, and write
\begin{equation}
H_0=\sum_{j=1}^{n}\epsilon_j b^{\dagger}_j b_j=\sum_{k} e_k \bigg(\sum_{j\in Q_k} b^{\dagger}_j b_j\bigg)
\label{eq:distinctE}
\end{equation}
where $\{e_k\}$ are the distinct single particle energies and $Q_k\subseteq [n]$ contains all modes with energy $e_k$. We may choose fermion operators $\{\tilde{b}_j: j=1,\ldots,n \}$ such that $H_0$ is diagonalized as
\begin{equation}
H_0=\sum_{k} e_k\bigg(\sum_{i\in A_k} \tilde{b}^\dagger_i\tilde{b}_i +\sum_{i\in B_k} \tilde{b}^\dagger_i\tilde{b}_i \bigg)
\label{eq:Hprime3}
\end{equation}
where $|B_k|\leq m$ for all $k$, and
\begin{equation}
\left[\tilde{b}_i,H_{imp}\right]=0 \quad  \text{for all } i\in \cup_k A_k.
\label{eq:commutingmodes}
\end{equation}
A particle-number conserving Gaussian unitary $V$ such that $\tilde{b}_j=V^{\dagger}b_jV$ for all $j\in [n]$ can be computed in $O(n^3)$ time.
\label{lem:fermionchoice}
\end{lemma}
\begin{proof}
Suppose that the operators $\{b_j\}$ in Eq.~\eqref{eq:distinctE} do not already satisfy the additional constraints in the Lemma. We will show how a Gaussian unitary transformation gives a new set of operators which satisfy these constraints. 

Each operator $b_i$ can itself be expressed as a linear combination of the Majorana operators $\{c_1,c_2,\ldots c_{2n}\}$. For all $i\in Q_k$ we have
\[
b_i=\sum_{j=1}^{2n} T^{(k)}_{ij} c_j
\]
where $T^{(k)}$  is a $|Q_k|\times 2n$ complex matrix. A QR decomposition gives a $|Q_k|\times |Q_k|$ unitary $U^{(k)}$ such that $U^{(k)}T^{(k)}$ is in reduced row echelon form. We may compute $U^{(k)}$ using $O(|Q_k|n^2)$ arithmetic operations. Setting
\begin{equation}
\tilde{b}_j=\sum_{s\in Q_k} U^{(k)}_{js} b_s=\sum_{p=1}^{2n}(U^{(k)}T^{(k)})_{jp} c_p \qquad \quad j=1\ldots,|Q_k| \quad k=0,1,2,\ldots
\label{eq:reduc}
\end{equation}
we see that the new creation and annihilation operators satisfy the fermion anticommutation relations and that 
\[
\sum_{j\in Q_k}\tilde{b}^{\dagger}_j\tilde{b}_{j}=\sum_{i\in Q_k}b^\dagger_i b_i
\]
and therefore $H_0$ satisfies Eq.~\eqref{eq:distinctE} with $b_i$ replaced by $\tilde{b}_i$. From Eq.~\eqref{eq:reduc} we see that the Gaussian unitary $V$ which maps $b_j\rightarrow \tilde{b}_j$ is associated with a linear transformation of fermion operators given by the block diagonal unitary $\bigoplus_k U^{(k)}$. The total runtime for computing all blocks of this unitary is $O(\sum_k |Q_k|n^2)=O(n^3)$.

Now focus on a fixed $k$, look at Eq.~\eqref{eq:reduc} and recall that  $U^{(k)}T^{(k)}$ is in reduced row echelon form. We see that there are at most $m$ modes in the set $\{\tilde{b}_j: \; j\in Q_k\}$ for which the right-hand side of Eq.~\eqref{eq:reduc} has a nonzero coefficient for the impurity Majorana modes $\{c_1,\ldots, c_{m}\}$. We partition $Q_k=A_k\cup B_k$ where $|B_k|\leq m$ indexes these modes. Eq.~\eqref{eq:commutingmodes} then follows, using the fact that $H_{imp}$ acts non-trivially only on the modes $c_1,\ldots,c_{m}$ and includes only even weight Majorana monomials.
\end{proof}

The second tool is a truncation lemma which bounds the change in the ground energy of $H$ when we truncate the single particle energies of the bath $H_0$ below a given threshold. Let a target precision $\gamma>0$ be given and let $\Omega$ index the single particle energies of $H_0$ which are at most $\gamma/m$, i.e., 
\[
\Omega=\left\{j\in [n]: \epsilon_j\leq \gamma/m \right\}   \quad \quad \Omega^c=[n]\setminus\Omega.
\]
Define a truncated impurity model
\begin{equation}
H(\gamma)=H_0(\gamma)+H_{imp} \qquad H_0(\gamma)=\frac{\gamma}{m}\sum_{j\in \Omega}b_j^{\dagger} b_j + \sum_{j\in \Omega^c}\epsilon_j b_j^{\dagger} b_j.
\label{eq:Hkappa}
\end{equation}
Here we have set all energies $\epsilon_j\leq \gamma/m$ to be equal to $\gamma/m$. Write $e_g(\gamma)$ for the ground energy of $H(\gamma)$.

\begin{lemma}[\textbf{Truncation lemma}]
\[
|e_g-e_g(\gamma)| \leq \gamma.
\]
\label{lem:deform}
\end{lemma}
\begin{proof}
First define another Hamiltonian 
\[
\hat{H}(\gamma)=\sum_{j\in \Omega^c}\epsilon_j b_j^{\dagger} b_j +H_{imp}
\]
and let its ground energy be $\hat{e}_g(\gamma)$. We have the operator inequality $\hat{H}(\gamma)\leq H\leq H(\gamma)$ and thus $\hat{e}_g(\gamma)\leq e_g\leq e_g(\gamma)$. To prove the proposition we now show that  $e_g(\gamma)-\hat{e}_g(\gamma)\leq \gamma$. Let $|\theta\rangle$ be a ground state of $\hat{H}(\gamma)$ such that
\[
\langle\theta|\sum_{j\in \Omega} b_j^{\dagger}b_j|\theta\rangle \leq m.
\]
It is always possible to choose such a ground state since all but $m$ of the modes in $\Omega$ can be decoupled from the impurity. That is, lemma \ref{lem:fermionchoice} implies that we may define new fermi modes $\tilde{b}_j$ such that
\[
\sum_{j\in \Omega} b_j^{\dagger}b_j=\sum_{j\in A} \tilde{b}_j^{\dagger}\tilde{b}_j+\sum_{j\in B} \tilde{b}_j^{\dagger}\tilde{b}_j
\]
and $[\tilde{b}_j,H_{imp}]=[\tilde{b}_j,\hat{H}(\gamma)]= 0$ for all  $j\in A$ and $|B|\leq m$. This implies we may choose a ground state $|\theta\rangle$ of $\hat{H}(\gamma)$ such that $\tilde{b}_j^{\dagger}\tilde{b}_j|\theta\rangle= 0$ for all $j\in A$ and thus
\[
\langle \theta |\sum_{j\in \Omega} b_j^{\dagger}b_j|\theta\rangle =\langle \theta |\sum_{j\in B} \tilde{b}_j^{\dagger}\tilde{b}_j|\theta\rangle\leq m.
\]

Now
\[
\hat{e}_g(\gamma)=\langle\theta|H(\gamma)|\theta\rangle+\langle\theta|\hat{H}(\gamma)-H(\gamma)|\theta\rangle\geq e_g(\gamma)-|\langle\theta|\hat{H}(\gamma)-H(\gamma)|\theta\rangle|.
\]
Therefore
\[
e_g(\gamma)-\hat{e}_g(\gamma)\leq |\langle\theta|\hat{H}(\gamma)-H(\gamma)|\theta\rangle|=\langle\theta|\sum_{j\in \Omega} \frac{\gamma}{m} b_j^{\dagger}b_j|\theta\rangle\leq \gamma.
\]
\end{proof}

\subsection{Quasipolynomial algorithm for general impurity models}
\label{sec:quasipolynomial}
In this Section we describe the quasipolynomial algorithm for approximating the ground energy and prove
Theorem~\ref{thm:1}, restated here for convenience.

\quasi*

We begin by introducing some additional notation used in this section.  Define an operator
\[
\mathcal{N}=\sum_{j=1}^{n} b_j^{\dagger}b_j
\]
which counts the number of excitations of the bath. For any $s\geq 0$ define $\mathcal{W}(s)$ to be the subspace spanned by all eigenvectors of $\mathcal{N}$ with eigenvalue at most $s$.  Corollary \ref{cor:gaussianrank} states that (at least one) ground state $\psi$ of an impurity model $H$ is approximated to precision $\delta$ by a state in $\mathcal{W}(s)$ whenever
\[
s\geq cm \log(2\omega^{-1}) \log(\delta^{-1})
\]
for some universal constant $c>0$. In this section we are interested in approximating the ground energy rather than the ground state itself. We use corollary \ref{cor:gaussianrank} and lemma \ref{lem:deform} to prove the following lemma.
\begin{lemma}
Let $\gamma\in (0,1/2]$ be a precision parameter. We have
\begin{equation}
e_g \leq \min_{\alpha \in \mathcal{W}(s)} \langle \alpha |H|\alpha\rangle \leq e_g +\gamma
\label{eq:eginterval}
\end{equation}
whenever 
\begin{equation}
s\geq cm \log^2 (m\gamma^{-1})
\label{eq:sstar}
\end{equation}
where $c>0$ is a universal constant.
\label{lem:es}
\end{lemma}
\begin{proof}
The lower bound in Eq.~\eqref{eq:eginterval} is trivial; below we prove the upper bound. Recall that we write $\omega$ for the spectral gap of the bath Hamiltonian $H_0$, that is, all nonzero single-particle excitation energies $\epsilon_j$ are in the interval $[\omega,1]$. Let $\gamma>0$ be the desired precision. 

As a first step we give a reduction to the special case where $\omega=\gamma/m$.  The reduction is based on the truncation lemma (lemma \ref{lem:deform}). Let $H(\gamma)$ be given as in Eq.~\eqref{eq:Hkappa}.  By definition the bath Hamiltonian $H_0(\gamma)$ has spectral gap $\omega=\gamma/m$. We have the operator inequality $H\leq H(\gamma)$ and therefore
\begin{equation}
 \min_{\alpha \in \mathcal{W}(s)} \langle \alpha |H|\alpha\rangle\leq \min_{\alpha\in \mathcal{W}(s)} \langle \alpha |H(\gamma)|\alpha\rangle.
\label{eq:hhtilde}
\end{equation}
Suppose that lemma \ref{lem:es} holds for precision $\gamma$ and Hamiltonian $H(\gamma)$. Then, using lemma \ref{lem:deform} and Eq.~\eqref{eq:hhtilde} we get 
\begin{equation}
 \min_{\alpha \in \mathcal{W}(s)} \langle \alpha |H|\alpha\rangle \leq e_g(\gamma)+\gamma \leq e_g+2\gamma.
\label{eq:reduc1}
\end{equation}
whenever
\begin{equation}
s \geq c'm\log^2(m/(2\gamma))\geq c m \log^2(m\gamma^{-1})
\label{eq:rescale}
\end{equation}
where $c'>0$ is another universal constant. Eqs.~(\ref{eq:reduc1},\ref{eq:rescale}) complete the reduction; they are Eqs.~(\ref{eq:eginterval},\ref{eq:sstar}) for the original Hamiltonian $H$, precision $2\gamma$, and constant $c'$. Thus we have shown Lemma \ref{lem:es} follows from its special case where $\omega=\gamma/m$.

To complete the proof, we now establish the lemma assuming $\omega=\gamma/m$. Fix some  $\delta\in (0,1/2]$ and let $\psi\in \calH_n$
be a normalized ground state of $H$ from Corollary~\ref{cor:gaussianrank}.
Let $P$ be the projector defined in Eq.~(\ref{Ptheta}).
Corollary \ref{cor:gaussianrank} states that $[H_{imp},P]=0$, and that $\|\psi-\phi\|\leq \delta$ with $P|\phi\rangle=|\phi\rangle$, which implies $\langle \psi|I-P|\psi\rangle \leq \delta^2$. Using these facts we get
\begin{align}
\bigg|e_g-\frac{\langle \psi|PHP|\psi\rangle}{\langle \psi|P|\psi\rangle}\bigg|&=\frac{1}{\langle \psi|P|\psi\rangle}\big|\langle \psi|(I-P)HP|\psi\rangle\big|\nonumber\\
&=\frac{1}{\langle \psi|P|\psi\rangle}\big|\langle \psi|(I-P)H_0P|\psi\rangle\big|\nonumber\\
&\leq \frac{\sqrt{\langle \psi|(I-P)|\psi\rangle}}{\langle \psi|P|\psi\rangle}\|H_0P\|\nonumber\\
&\leq 2\delta\|H_0 P\|
\label{eq:deltanorm}
\end{align}
where in the last line we used $\delta\leq \frac{1}{\sqrt{2}}$.  Now $H_0\leq \mathcal{N}$ and, since both operators are diagonal over the same basis we also have $H_0^2\leq \mathcal{N}^2$. Using this fact and Eq.~(\ref{eq:Pcond}) we get
\begin{equation}
\|H_0P\|\leq \|\mathcal{N} P\|\leq c m \log(2\omega^{-1})\log(\delta^{-1})= cm \log(2m\gamma^{-1})\log(\delta^{-1})
\label{eq:mplusk}
\end{equation}
for some constant $c>0$. In the last equality we substituted $\omega=\gamma/m$. Since $\gamma\leq 1/2$ and $m\geq 1$ we have $\log(2m\gamma^{-1})\leq O(1)\cdot \log(m\gamma^{-1})$ and therefore 
\begin{equation}
\|H_0P\|\leq \|\mathcal{N} P\|\leq c_1m \log(m\gamma^{-1})\log(\delta^{-1})
\label{eq:mplusk2}
\end{equation}
where $c_1$ is another universal constant. Combining Eqs.~(\ref{eq:deltanorm},\ref{eq:mplusk2})  we get
\begin{equation}
\bigg|e_g-\frac{\langle \psi|PHP|\psi\rangle}{\langle \psi|P|\psi\rangle}\bigg|\leq   2\delta c_1m \log(m\gamma^{-1})\log(\delta^{-1})
\label{eq:estare0}
\end{equation}
Now we choose $\delta$ such that the right hand side is at most $\gamma$. It suffices to take
\begin{equation}
\delta= \frac{\gamma}{Cm\log^2(m\gamma^{-1})}
\label{eq:dchoice}
\end{equation}
where $C$ is any universal constant satisfying $2c_1 C^{-1}(\log(C)+3)\leq 1$. Indeed, with this choice we have
\begin{equation}
2 \delta c_1m \log(m\gamma^{-1})\log(\delta^{-1})= \gamma \frac{2c_1}{C} \left[ \frac{\log(C)}{\log(m\gamma^{-1})}+1+2\frac{\log(\log(m\gamma^{-1}))}{\log(m\gamma^{-1})}\right]\leq \gamma
\label{eq:dgamma}
\end{equation}
where we used the fact that the quantity in square parentheses is at most $\log(C)+3$.

Thus
\begin{equation}
\frac{\langle \psi|PHP|\psi\rangle}{\langle \psi|P|\psi\rangle}\leq e_g+\gamma
\label{eq:e0P}
\end{equation}
From Eqs.~(\ref{eq:Pcond},\ref{eq:mplusk2})  we have $[P,\mathcal{N}]=0$ and 
\[
\|\mathcal{N}P\|\leq c_1m\log(m\gamma^{-1}) \log(\delta^{-1}) \leq \gamma/2\delta= (C/2)m \log^2(m\gamma^{-1}).
\]
where in the second inequality we used Eq.~\eqref{eq:dgamma} and in the last inequality we used Eq.~\eqref{eq:dchoice}. This implies that the image of $P$ is contained in the subspace $\mathcal{W}(s)$ whenever
\begin{equation}
s\geq (C/2)m \log^2(m\gamma^{-1}).
\label{eq:scondition}
\end{equation}
Thus, for all $s$ satisfying Eq.~\eqref{eq:scondition} we have 
\[
 \min_{\alpha \in \mathcal{W}(s)} \langle \alpha |H|\alpha\rangle\leq \frac{\langle \psi|PHP|\psi\rangle}{\langle \psi|P|\psi\rangle}\leq e_g+\gamma
\]
where in the second inequality we used Eq.~\eqref{eq:e0P}. 
\end{proof}

We now define a deformed impurity model. Let $\gamma\in(0,1/2]$ be a precision parameter, let $s^{\star}=\lceil cm \log^2 (m\gamma^{-1})\rceil$ be the smallest integer greater than or equal to the right hand side of Eq.~\eqref{eq:sstar} and consider a set of grid points 
\begin{equation}
G=\left\{x\gamma/s^{\star}: x\in \{1,2,\ldots\}\right\}.
\label{eq:J}
\end{equation}
For each $j$ let $\epsilon^\prime_j$ be the smallest element of $G$ which is at least $\epsilon_j$, so that 
\begin{equation}
\epsilon_j\leq \epsilon_j^{\prime}\leq \epsilon_j+\gamma/s^{\star} \qquad \quad j=1,2,\ldots,n.
\label{eq:eprime}
\end{equation}
Define deformed  Hamiltonians
\[
H^\prime _0=\sum_{j=1}^{n}\epsilon'_j b^{\dagger}_j b_j, \qquad \mbox{and} 
\qquad
H'=H_0'+H_{imp}.
\]
Finally, define
\begin{equation}
e_g^{\star}=\min_{\phi\in W(s^{\star})} \langle \phi |H'|\phi\rangle.
\label{eq:egstar}
\end{equation}
The following lemma shows that $e_g^{\star}$ is a good approximation to $e_g$. Our algorithm for approximating $e_g$ is based on computing $e_g^{\star}$.
\begin{lemma}
\[
|e_g-e_g^{\star}|\leq 2\gamma.
\]
\end{lemma}
\begin{proof}
Applying lemma \ref{lem:es} gives
\[
|e_g-\min_{\phi\in W(s^{\star})} \langle \phi |H|\phi\rangle|\leq \gamma.
\]
Therefore
\begin{align}
|e_g-e_g^{\star}|& \leq \gamma+\left|\min_{\phi\in W(s^{\star})} \langle \phi |H'|\phi\rangle-\min_{\phi\in W(s^{\star})} \langle \phi |H|\phi\rangle\right|\\
&\leq \gamma+\|(H'-H)|_{W(s^{\star})}\|
\label{eq:3gam}
\end{align}
where in the last line we used Weyl's inequality. Here we use the notation $M|_{\mathcal{S}}$ to denote the restriction of an operator $M$ to a subspace $\mathcal{S}$.  Now
\[
H'-H=\sum_{j=1}^{n} (\epsilon'_j-\epsilon_j)b^{\dagger}_j b_j
\]
and using Eq.~\eqref{eq:eprime} we arrive at
\[
0\leq H'-H\leq \frac{\gamma}{s^{\star}}\sum_{j=1}^{n} b^{\dagger}_j b_j.
\]
Thus
\begin{equation}
\|(H'-H)|_{W(s^{\star})}\|\leq \frac{\gamma}{s^{\star}}\bigg\| \bigg(\sum_{j=1}^{n} b^{\dagger}_j b_j\bigg)\bigg|_{W(s^{\star})}\bigg\|\leq \gamma,
\label{eq:norminsubspace}
\end{equation}
where in the last line we used the definition of $W(s^{\star})$. Plugging this into Eq.~\eqref{eq:3gam} completes the proof.
\end{proof}
Thus we have shown that to approximate $e_g$ it suffices to consider the deformed impurity model. Why is this useful to us? The total number of \textit{distinct} single-particle energies $\epsilon'_j$ is at most $1+|G\cap [0,1]|\leq 1+s^{\star}/\gamma$, which does not depend on $n$. Since there are $n$ modes in total, we see that on average a single-particle energy of the deformed bath $H'_0$ has degeneracy linear in $n$. Because of this massive degeneracy, we may use the decoupling lemma to show that many of the degrees of freedom (modes) can be decoupled from the impurity.

Applying lemma \ref{lem:fermionchoice} to the deformed impurity model $H'$ we get fermion operators $\{\tilde{b}_j: j=1,\ldots,n \}$ and subsets $A_k,B_k\subseteq [n]$ such that
\begin{equation}
H_0'=\sum_{k=1,2,\ldots } \frac{k\gamma}{s^{\star}}\left(\sum_{i\in A_k} \tilde{b}^\dagger_i\tilde{b}_i +\sum_{i\in B_k} \tilde{b}^\dagger_i\tilde{b}_i \right)
\label{eq:deform1}
\end{equation}
such that $|B_k|\leq m$ and the modes in $\cup_k A_k$ are decoupled from the impurity, that is, 
\begin{equation}
[\tilde{b}_j, H_{imp}]=0  \quad \text{ whenever }\quad \tilde{b}_j\in \cup_k A_k.
\label{eq:deform2}
\end{equation}
For ease of notation, let us order the modes so that the decoupled ones appear first
\[
\cup_k A_k =\{1,2,\ldots,N\} \qquad \quad \cup_k B_k=\{N+1,N+2,\ldots,n\}
\]
Note that the total number of \text{coupled} modes is upper bounded by $m$ times the number of distinct single particle energy levels $\epsilon'_j$, that is,
\begin{equation}
n-N\leq m (1+\left|G\cap [0,1]\right|)\leq m(1+s^{\star}/\gamma)\leq \frac{2ms^{\star}}{\gamma}.
\label{eq:numcoupled}
\end{equation}

For each $z\in \{0,1\}^n$ define a Fock basis state $|\tilde{z}\rangle$ with respect to the modes $\{\tilde{b}_j\}$, i.e., 
\begin{equation}
 \tilde{b}_j^{\dagger}\tilde{b}_j |\tilde{z}\rangle=z_j|\tilde{z}\rangle \qquad j=1,2,\ldots,n.
\label{eq:defz}
\end{equation}

Note that $|\tilde{z}\rangle$ is an eigenstate of the operator $\sum_{j} \tilde{b}^{\dagger}_j \tilde{b}_j$ with eigenvalue $\sum_{j=1}^{n} z_j $. Lemma \ref{lem:fermionchoice} states that the decoupling transformation preserves particle number, that is, 
\begin{equation}
\mathcal{N}=\sum_{j=1}^{n} b_j^{\dagger} b_j=\sum_{j=1}^{n} \tilde{b}_j^{\dagger} \tilde{b}_j.
\label{eq:numberop}
\end{equation}
By definition, the subspace $\mathcal{W}(s^{\star})$ is spanned by all eigenstates of the number operator Eq.~\eqref{eq:numberop} with eigenvalues at most $s^{\star}$. Therefore 
\[
\mathcal{W}(s^{\star})=\mathrm{span}\bigg\{|\tilde{z}\rangle: \sum_{i=1}^{n} z_i \leq s^{\star}\bigg\}.
\]
Define a subspace
\begin{equation}
\mathcal{V}=\mathrm{span}\bigg\{|\tilde{z}\rangle: \sum_{i=1}^{n} z_i \leq s^{\star}\;\; \text{ and } \;\; z_1=z_2=\ldots=z_N=0\bigg\} 
\label{eq:Vbasis}
\end{equation}
spanned by basis vectors where the decoupled modes are unoccupied. We now show that the minimization in Eq.~\eqref{eq:egstar} can be restricted to the subspace $\mathcal{V}$.
\begin{lemma}
\begin{equation}
e_g^{\star}= \min_{\phi \in \mathcal{V}} \langle \phi |H'|\phi\rangle.
\label{eq:egprime}
\end{equation}
\label{lem:restrict}
\end{lemma}
\begin{proof}
Eqs.~(\ref{eq:deform1},\ref{eq:deform2}) imply $[\tilde{b}_j^\dagger \tilde{b}_j,H^\prime]=0$ for all $j\in \{1,2,\ldots,N\}$ and therefore
\[
\langle \tilde{z}|H'|\tilde{y}\rangle=0 \quad \text{ whenever } y_i\neq z_i \text{ for some } i\in \{1,2,\ldots,N\}
\]
Thus the restriction $H'|_{\mathcal{W}(s^{\star})}$ is block diagonal in the basis $\{|\tilde{z}\rangle: z\in\{0,1\}^n\}$ with a block for each configuration $z_1z_2\ldots z_N$ of the decoupled modes. The smallest eigenvalue $e^{\star}_g$ of $H'|_{\mathcal{W}(s^{\star})}$ is the smallest eigenvalue of one of the blocks. In particular, for some $|\phi\rangle$ and $x\in \{0,1\}^N$ we have
\[
e_g^{\star}=\langle \phi|H'|\phi \rangle \quad \text{and} \quad \tilde{b}_j^\dagger \tilde{b}_j|\phi\rangle=x_j |\phi\rangle  \quad j=1,2,\ldots,N
\]
Now let $|\alpha \rangle=\tilde{b}_1^{x_1}\tilde{b}_2^{x_2}\ldots \tilde{b}_N^{x_N}|\phi\rangle$ and note that $\alpha \in \mathcal{V}$. Applying Lemma \ref{lem:fermionchoice} we get
\[
\langle \alpha |H^{\prime}_0|\alpha \rangle\leq \langle \phi|H^{\prime}_0|\phi\rangle \quad \text{and} \quad \langle \alpha |H_{imp}|\alpha \rangle= \langle \phi|H_{imp}|\phi\rangle
\]
and therefore $\langle \alpha |H'|\alpha \rangle\leq \langle\phi|H'|\phi\rangle$. Note that equality must hold since $\phi$ minimizes the energy of $H^{\prime}$ in $\mathcal{W}(s^{\star})$. We have shown there exists $\alpha \in \mathcal{V}$ with $e_g^{\star}=\langle \alpha|H'|\alpha\rangle$, which completes the proof.
\end{proof}
We now use Eqs.~(\ref{eq:Vbasis},\ref{eq:numcoupled}) to upper bound
\[
\dim(\mathcal{V})\leq \binom{n-N}{ s^{\star}} \leq \left(e\frac{n-N}{s^{\star}}\right)^{s^{\star}} \leq  \left(\frac{2em}{\gamma}\right)^{s^{\star}}=e^{O(m\log^3(m\gamma^{-1}))}.
\]
where in the second inequality we used the bound $\binom{l}{k}\leq (l e/k)^k$ where $e=\mathrm{exp}(1)$.

Let $D\equiv \dim{(\calV)}$. 
We claim that the righthand side of Eq.~(\ref{eq:egprime})  can be computed in time
$O(2^m n^3 D^2 + D^3)$. Indeed Eq.~(\ref{eq:Vbasis}) gives
an orthonormal set of Gaussian states  $\Phi=(\phi_1,\ldots,\phi_D)$ 
that spans $\calV$. By construction, the deformed bath Hamiltonian is diagonal in this basis
and one can compute a matrix element $\langle \phi_j|H_0'|\phi_j\rangle$ in time $O(n)$
by summing up energies of all excitations present in $\phi_j$.
Thus one can compute the matrix of $H_0'$ in the basis $\Phi$ in time
$O(nD)$. Consider now the impurity Hamiltonian $H_{imp}$. By construction,
$H_{imp}$ is a linear combination of $O(2^m)$ Majorana monomials $c(x)$.
Using the generalized Wick's theorem Eq.~(\ref{Wick2}) one can compute 
a single matrix element $\langle \phi_i|c(x)|\phi_j\rangle$ in time $O(n^3)$.
Thus one can compute the full matrix of $H_{imp}$ in the basis $\Phi$ in time
$O(2^m n^3 D^2)$. Once the matrices of $H_0'$ and $H_{imp}$ 
in the basis $\Phi$ are computed, one can calculate $e_g^{\star}$ using exact
diagonalization in time $O(D^3)$.  Recall that $|e_g-e_g^{\star}|\leq 2\gamma$ (to get rid of the factor of $2$ we may rescale the precision parameter $\gamma\rightarrow 2\gamma$ without altering the asympotic runtime of the algorithm). This completes the proof of Theorem~\ref{thm:1}. 

\textit{Remark}: The above algorithm can be used to produce a low energy state of $H$, the original (not deformed) impurity model. Indeed, in the last step of the algorithm, one may use an exact diagonalization routine which, along with the eigenvalue $e_g^{\star}$, computes a state $\alpha \in \mathcal{V}\subseteq W(s^{\star})$ satisfying $e_g^{\star}=\langle \alpha |H'|\alpha\rangle$. In this case we have
\[
e_g^{\star}-\langle \alpha|H|\alpha\rangle =\langle \alpha |H'|\alpha\rangle-\langle \alpha|H|\alpha\rangle \leq \|(H'-H)|_{W(s^{\star})}\|\leq \gamma
\]
where we used Eq.~\eqref{eq:norminsubspace}. Combining this with the fact that $|e_g-e_g^{\star}|\leq 2\gamma$, we see that the computed state $|\alpha\rangle$ satisfies 
\[
|e_g- \langle \alpha|H|\alpha\rangle |\leq 3\gamma.
\]

\label{sec:generalcase}
\subsection{Efficient algorithm for gapped impurity models}
\label{sec:gappedcase}
In this Section we prove theorem \ref{thm:gap}, restated here for convenience.
\gap*
\begin{proof}
Let $e^1,e^2,\ldots,e^{2n}$ be the standard basis of $\RR^{2n}$. Let $h$ be the  $2n\times 2n$ matrix defined in Eq.~\eqref{bath1}. Define a nested sequence of linear subspaces $\calL_1\subseteq \calL_2 \subseteq 
\ldots \subseteq \calL_v \subseteq \RR^{2n}$ such that 
\[
\calL_1=\mathrm{span}(e^1,e^2,\ldots,e^m) \quad
\mbox{and} \quad \calL_{j}=\mathrm{span}(\calL_1, h \calL_1,\ldots,h^{j-1} \calL_1)
\]
for $j\ge 2$.
We choose $v$ as the smallest integer such that $\calL_{v+1}=\calL_v$.
Obviously, $v=O(n)$. 
Let $L=\dim{(\calL_v)}$. By construction, $\calL_v$ is $h$-invariant.
Define a subspace
\[
\calK_j=\calL_j\cap \calL_{j-1}^\perp.
\]
Let us agree that $\calL_0=0$, so that $\calK_1=\calL_1$. We get a direct sum decomposition 
\begin{equation}
\label{Ksubspaces}
\RR^{2n}=\calK_1\oplus \calK_2 \oplus \cdots \oplus \calK_v \oplus \calL_v^\perp.
\end{equation} 
Note that $h$ is block-tridiagonal with respect to this decomposition, that is,
$\langle \alpha |h|\beta\rangle =0$ whenever $\alpha\in \calK_i$, $\beta\in \calK_j$,
and $|i-j|\ge 2$. Indeed, assume wlog that $i\ge j+2$.  Then 
\[
h|\beta\rangle \in h\calK_j \subseteq h\calL_j \subseteq \calL_{j+1} \subseteq \calL_{i-1}
\]
whereas  $|\alpha\rangle\in \calK_i \subseteq \calL_{i-1}^\perp$.
Furthermore, 
\[
\dim{(\calK_j)}=\dim{(\calL_j)}- \dim{(\calL_{j-1})} \le m
\]
since  $\calL_j$ is spanned by $\calL_{j-1}$ and $h^{j-1} \calL_1$.

Choose an orthonormal basis $f^1,f^2,\ldots,f^{2n}\in \RR^{2n}$ such that the
first $\dim{(\calK_1)}$ basis vectors span $\calK_1$, the next $\dim{(\calK_2)}$ 
basis vectors span $\calK_2$ and so on. The last $2n-\dim{(\calL_v)}$ basis vectors
span $\calL_v^\perp$. Define a new set of Majorana operators 
\begin{equation}
\label{newMajorana}
\tilde{c}_p=\sum_{q=1}^{2n} (f^p)_q c_q, \quad \quad p=1,\ldots,2n.
\end{equation}
Here $(f^p)_q$ is the $q$-th component of $f^p$. The operators $\tilde{c}_p$
obey the same commutation rules as $c_p$. 
By construction, $\tilde{c}_p=c_p$ for $1\le p\le m$ and thus $H_{imp}$
belongs to the algebra generated by $\tilde{c}_1,\ldots,\tilde{c}_m$.
Transforming  $h$ to the new basis we find that 
\begin{equation}
\label{Hnew1}
h=\left[ \ba{cc} h' & \\ & h'' \\ \ea \right],
\end{equation}
where the two blocks have dimension $L$ and $2n-L$ respectively.
Moreover, $h'$ is block-tridiagonal with non-zero matrix elements only 
between blocks $\calK_i$, $\calK_j$ with $|i-j|\le 1$.
 We conclude that 
\begin{equation}
\label{Hnew2}
H=H_{imp} + H_{A}+ H_{B}, \quad 
H_{A}=\frac{i}4 \sum_{p,q=1}^{L} h'_{p,q} \tilde{c}_p \tilde{c}_q,
\quad
H_{B}=\frac{i}4 \sum_{p,q=L+1}^{2n} h''_{p,q} \tilde{c}_p \tilde{c}_q.
\end{equation}
For simplicity, here we ignore the constant energy shift in Eq.~(\ref{bath1}).
The terms $H_{imp}+H_{A}$ and $H_{B}$ act on disjoint sets of modes and the ground energy of $H$ is the sum of their ground energies. Since $H_{B}$ is quadratic its ground energy is $-\|h''\|_1/4$, which is easily computed from the singular value decomposition of $h''$. Thus we can concentrate on $H_{imp}+H_{A}$
which acts on modes $\tilde{c}_1,\ldots,\tilde{c}_L$. Without loss of generality we shall assume $L$ is even in the following; if it is not even we may simply view $H_{imp}+H_A$ as acting on modes $\tilde{c}_1,\ldots,\tilde{c}_L$ in addition to one auxiliary Majorana mode.   We can map $H_{imp}+H_{A}$ to a Hamiltonian
describing a 1D chain of qubits using the standard Jordan-Wigner transformation: 
\[
\tilde{c}_1=X_1, \quad \tilde{c}_2=Y_1, 
\]
and
\[
\tilde{c}_{2a-1}=Z_1\cdots Z_{a-1} X_a \quad \mbox{and} \quad \tilde{c}_{2a}=Z_1\ldots Z_{a-1} Y_a
\]
for $a\ge 2$.  Here $X_a,Y_a,Z_a$ are the Pauli operators on the $a$-th qubit.
Since we have $L$ Majorana modes, the chain consists of $L/2$ qubits.  We may coarse-grain the chain such that the first $m$ qubits form the first site, the next $m$ qubits form the second site, and so on (the last qudit may consist of $<m$ qubits). Since $h'$ is block
 tridiagonal with block size upper bounded as $\dim{(\calK_j)}\leq m$, we have
\[
h'_{ij}=0 \quad \text{whenever} \quad |i-j|\geq 2m.
\]
This block tridiagonal structure implies that after the Jordan-Wigner transformation
the Hamiltonian $H_{imp}+H_A$ describes a 1D chain of qudits with nearest-neighbor
interactions.  Each qudit has dimension at most
\[
d=2^m =O(1).
\]
The chain has length $v=O(n)$. 
Since all of the above transformations are unitary, they preserve eigenvalues. 
Thus a gapped quantum impurity model can be efficiently mapped to a gapped 1D chain of qudits.
One can approximate the ground state energy of the latter within error $\delta$
in time $poly(n,\delta^{-1})$ using MPS-based algorithms, see Ref.~\cite{1Dgapped}.
\end{proof}

\subsection{Containment in QCMA}
\label{sec:QCMA}
We now consider the complexity of estimating the ground energy of a quantum impurity model to inverse polynomial precision. Formally, we consider the following decision problem (restated from Section \ref{sec:misc}).
\qprob*
In this Section we prove the following theorem.
\qcma*

We first review some facts concerning the representation of fermionic states on a quantum computer. Quantum states of $n$ fermionic modes are represented using $n$ qubits in the following way: for each $x\in \{0,1\}^n$, the Fock basis state $\prod_{i=1}^{n} (a_i^{\dagger})^{x_i}|0^n\rangle$ is identified with the $n$-qubit computational basis state $|x\rangle$. Majorana operators are represented via the Jordan-Wigner transformation:
\begin{align}
c_1&=X_1 \label{eq:JW1} \\
c_2&=Y_1  \label{eq:JW2} \\
c_{2a-1}&=Z_1\cdots Z_{a-1} X_a  \label{eq:JW3} \\
c_{2a}&=Z_1\ldots Z_{a-1} Y_a  \label{eq:JW4} 
\end{align}

With this representation, a Gaussian state can be prepared efficiently on a quantum computer. A simple strategy to prepare any Gaussian state was given in Ref.~\cite{Electronquantum}; we summarize it here for completeness. 

Let a Gaussian state $|\Phi\rangle$ be specified up to a global phase  by its covariance matrix $M$, as defined in Eq.~\eqref{M}. For each $i,j=1,2\ldots,2n$ and $\theta\in [0,\pi]$ define a unitary 
\[
U(\theta,i,j)=e^{\frac{\theta}{2} c_i c_j}.
\]
 Using the Majorana commutation relations one can easily check that the state $\Phi'$ defined by
\[
|\Phi'\rangle=U(\theta,i,j)|\Phi\rangle 
\]
has covariance matrix
\[
M'=R(\theta,i,j) M R(\theta,i,j)^{T}
\]
where $R(\theta,i,j)$ is a ``Givens rotation'' which acts nontrivially only on the subspace spanned by basis vectors $i,j$, and within this subspace its action is described by the $2\times 2$ matrix
\[
\left(\begin{array}{cc} \cos(\theta) & \sin (\theta) \\ -\sin(\theta) & \cos(\theta)\end{array}\right).
\]
Since $M$ is a $2n\times 2n$ real and antisymmetric matrix satisfying $M^2=-I$, there is an $\mathrm{SO}(2n)$ matrix $R$ such that $M=RM_{y}R^T$, where $M_{y}$ is the covariance matrix of a standard basis state, see Eq.~\eqref{V}. The matrix $R$ can be computed efficiently using linear algebra, e.g., from the real Schur decomposition of $M$ \cite{GolubVL}.  It can then be decomposed as a product
\[
R=R(\theta_q,i_q,j_q)\ldots R(\theta_2,i_2,j_2)R(\theta_1,i_1,j_1)
\]
for some angles $\theta_1,\ldots,\theta_q$ and qubit indices $i_1,j_1,i_2,j_2\ldots,i_q,j_q$, where $q=O(n^2)$. Such a decomposition can be computed, for example, using the strategy provided in Section 4.5.1 of the textbook Nielsen and Chuang \cite{nielsenchuang} for decomposing a given unitary into two-level unitaries. We then have
\[
|\Phi\rangle=e^{i\kappa} U(\theta_q,i_q,j_q)\ldots U(\theta_2,i_2,j_2)U(\theta_1,i_1,j_1)|y\rangle,
\]
for some global phase $e^{i\kappa}$ and computational basis state $y\in \{0,1\}^n$. From this expression it is evident that $|\Phi\rangle$ can be prepared efficiently, since the right hand side consists of $O(n^2)$ unitaries $U(\theta,i,j)$, each of which is an exponential of a Pauli operator. In the following we will only need to use the slightly weaker statement that there exists a polynomial sized quantum circuit which prepares $\Phi$ starting from $|0^n\rangle$.

We now give the proof of theorem \ref{thm:QCMA}.
\begin{proof}

Let $H,a,b$ be given. Define $\gamma=(b-a)/4$ and consider the truncated Hamiltonian $H(\gamma)$ from Eq.~\eqref{eq:Hkappa}. Here $H(\gamma)=H_0(\gamma)+H_{imp}$ where $H_0(\gamma)$ has spectral gap $\omega=\gamma/m$. The truncation lemma (lemma \ref{lem:deform}) states that 
\[
|e(\gamma)-e_g|\leq \gamma.
\]
Therefore 
\begin{align}
e_g <a &\implies e(\gamma)<a' \qquad \quad a'=a+(b-a)/4\nonumber\\
e_g>b &\implies e(\gamma)>b' \qquad \quad b'=b-(b-a)/4\nonumber\\
\end{align}
Note $b'-a'=(b-a)/2=1/\mathrm{poly}(n)$.  Also note that $H(\gamma)$ can be computed efficiently given $H,b,a$; we need only diagonalize the free fermion Hamiltonian $H_0$ which takes time $O(n^3)$.  We have shown that the solution to the quantum impurity problem for $H$ with precision parameters $a,b$ is equivalent to the quantum impurity problem for $H(\gamma)$ with precision parameters $b',a'$. Below we give a QCMA protocol for the latter.

We use corollary \ref{cor:gaussianrank} with $\delta=0.01$ (say) and Hamiltonian $H(\gamma)$. Since $m=O(1)$, it implies that there exists a normalized state $\phi$ of the form
\begin{align}
|\phi\rangle=\sum_{j=1}^{\chi} z_j |\theta_j\rangle \qquad \chi=e^{O(\log(\gamma^{-1}))}=\mathrm{poly}(1/\gamma)=\mathrm{poly}(n).
\label{eq:phistate}
\end{align}
where $\{\theta_j\}$ are orthonormal Gaussian states, and where $\|\phi-\psi\|\leq 0.01$ for some ground state $\psi$ of $H(\gamma)$.

As discussed above, each Gaussian state $\theta_j$ can be prepared by a quantum circuit of size $O(n^2)$. This implies that there is also a polynomial sized quantum circuit which prepares $\phi$ with high probability ($0.99$, say) and a flag qubit indicating success.  Indeed, the unitary
\[
\sum_{j=1}^{\chi} |j\rangle\langle j|\otimes U_j
\]
can be implemented with a circuit of size $O(\chi n^2)$ (see for example Lemma 8 of Ref.~\cite{CKS15}). Here the first register consists of $\lceil \log_2(\chi)\rceil$ ancilla qubits.  Applying this unitary to the state $\chi^{(-1/2)}\sum_{i=1}^{\chi}|i\rangle|0^n\rangle$ one obtains $\chi^{(-1/2)}\sum_{i=1}^{\chi}|i\rangle |\theta_i\rangle$. We may then perform the projective measurement $\{ |\vec{z}\rangle \langle \vec{z}|\otimes I, (I-|\vec{z}\rangle \langle \vec{z}|)\otimes I\}$, where
\[
|\vec{z}\rangle=\sum_{j=1}^{n} z_j |j\rangle.
\]
With probability $\chi^{-1}$ we obtain the desired state $\phi$. A quantum circuit which repeats this procedure $\Theta(\chi)=\mathrm{poly}(n)$ times will produce $\phi$ with high probability.

Now we are ready to describe the QCMA protocol. The protocol begins with Merlin sending Arthur a classical description of a polynomial sized quantum circuit with an $n$-qubit output register and a flag qubit.  Ideally, Arthur would like this to be the circuit which prepares $\phi$ with probability $0.99$ (with the flag qubit indicating whether or not $\phi$ has been successfully prepared).

Arthur applies the given circuit to the all zeros input state and measures the flag qubit. If he obtains measurement outcome $1$, then he performs the following test on the $n$-qubit output register, otherwise he rejects.
Using phase estimation \cite{nielsenchuang} on $e^{iH(\gamma)}$, he measures the eigenvalue of $H(\gamma)$ within precision $(b'-a')/4$. The parameters of the phase estimation are chosen so that the probability of failure is at most $\epsilon=0.01$ (say). He accepts if the measured eigenvalue is at most $a'+ (b'-a')/2$. 

Let us now analyze the completeness and soundness of this protocol. In the yes case Merlin can send the circuit which prepares $\phi$ with probability $0.99$. In this case there is a ground state $\psi$ with energy  $e(\gamma)<a'$ such that $\|\phi-\psi\|\leq 0.01$, which implies $|\langle \phi|\psi\rangle|^2 \geq 1-0.01^2$. Thus, if the phase estimation succeeds, then with probability at least $1-0.01^2$ Arthur will measure an energy which is at most $a'+(b'-a')/4$ causing him to accept. Thus the total probability for Arthur to accept is the product of 0.99 (the probability that the flag qubit is measured to be $1$), 0.99 (the probability phase estimation succeeds), and $1-0.01^2$ (the probability that Arthur accepts given that phase estimation succeeds). We have shown that in the yes case Merlin can provide a classical witness which causes Arthur to accept with probability at least $2/3$.

Next consider the no case. In this case the ground energy of $H(\gamma)$ is at least $b'$. Thus any eigenvalue of $H(\gamma)$, when approximated within precision $(b'-a')/4$, is greater than Arthur's threshold to accept, i.e., greater than $a'+ (b'-a')/2$. This means that, regardless of the circuit provided by Merlin, the only way Arthur will accept is if the phase estimation fails (which occurs with probability at most $0.01$). Thus, in the no case, Arthur rejects with probability at least $1/3$.

Finally, note that Arthur's computation has polynomial running time. He first performs the (polynomial-sized) circuit given to him by Merlin, and then performs phase estimation to precision $(b'-a')/4=1/\mathrm{poly}(n)$ and constant error probability. This phase estimation requires Arthur to implement Schrodinger time evolution $e^{itH(\gamma)}$ for times $t=\mathrm{poly}(n)$. This can be done efficiently since $H(\gamma)$ is a sparse and efficiently row-computable Hamiltonian \cite{AdiabaticSZK}. Indeed, since $H(\gamma)$ is a quantum impurity Hamiltonian it can be written as a sum of $n^2+2^m$ Majorana monomials, each of which corresponds to an $n$-qubit Pauli.  

\end{proof}

%%%%%%%%%%%%%%%%%%%%%%%%%%%%%%%%
\section{Simplified practical algorithm}
\label{sec:variational}

In this section we describe a variational algorithm that minimizes the energy of 
a quantum impurity Hamiltonian over low-rank superpositions of Gaussian states. 
This algorithm provides an upper bound on the ground energy $e_g$.
We also describe an extended (less efficient) version of this algorithm that provides both upper and lower
bounds on $e_g$. We benchmark the method using the single impurity Anderson model~\cite{Anderson61}.

Consider a system of $n$ fermi modes and a Hamiltonian 
composed of arbitrary quadratic and quartic Majorana monomials:
\begin{equation}
\label{varH1}
H=H_0+H_{imp},
\end{equation}
\begin{equation}
\label{varH2}
H_0=i\sum_{1\le p<q\le 2n} A_{pq} c_pc_q,
\end{equation}
\begin{equation}
\label{varH3}
H_{imp}=\sum_{1\le p<q<r<s\le m} U_{pqrs} c_p c_q c_r c_s.
\end{equation}
Here $A_{pq}$ and $U_{pqrs}$ are some real coefficients.
We shall  extend the range of all sums over Majorana modes
to the interval $[1,2n]$ assuming that $U_{pqrs}=0$ unless
$1\le p<q<r<s\le m$.
Since the Hamiltonian $H$ commutes with the parity operator $P$, see Eq.~(\ref{Parity}),
we can minimize the 
energy  of $H$ in the even and the odd subspaces separately. 
Equivalently, one can minimize the energy of $H$ and $c_1Hc_1$ within the even subspace.
Since both minimizations are exactly the same, below we consider only the even subspace. 

Our variational algorithm depends on an integer $\chi\ge 1$ that we call a {\em rank}.
Let
\[
\Phi=(\phi_1,\phi_2,\ldots,\phi_\chi)
\]
be a tuple of $\chi$ even Gaussian states $\phi_a\in \calG_n$
with covariance matrices $M_a$. 
The states $\phi_a$ may or may not be pairwise orthogonal. 
The algorithm works by minimizing an objective function $E(M_1,\ldots,M_\chi)$ defined as 
the smallest eigenvalue of $H$ restricted to  the linear  subspace  spanned by 
states $\phi_1,\ldots,\phi_\chi$:
\begin{equation}
\label{objective}
E(M_1,\ldots,M_\chi)=\min_{\psi \in \mathrm{span}(\Phi)}\; \; \frac{\langle \psi |H|\psi\rangle}{\langle\psi|\psi\rangle}.
\end{equation}
Here $\mathrm{span}(\Phi)\equiv \mathrm{span}(\phi_1,\ldots,\phi_\chi)$.
Note that  $\mathrm{span}(\Phi)$ is uniquely determined by the covariance matrices $M_1,\ldots,M_\chi$.

We shall parameterize  the covariance matrix of $\phi_a$ by a rotation $R_a\in SO(2n)$ such that  
\begin{equation}
\label{paramR}
M_a=R_a M_{vac} R_a^T.
\end{equation}
Here 
\[
M_{vac}=\bigoplus_{j=1}^n \left[ \ba{cc} 0 & 1 \\ -1 & 0 \\ \ea \right]
\]
is the covariance matrix of the vacuum state $|0^n\rangle$,
see Eq.~(\ref{V}). 
Thus we have to minimize the objective function Eq.~(\ref{objective}) over the group 
\[
SO(2n)\times \ldots \times SO(2n),
\]
where the direct product contains $\chi$ factors.

\subsection{Rank-$1$ algorithm}

Let us start from the simplest case $\chi=1$, i.e. minimizing the 
energy of $H$ over all Gaussian states. 
This corresponds to a generalized Hartree-Fock method proposed recently by Kraus and Cirac~\cite{KC10} for simulation of interacting fermions on a lattice. 
Let $M\equiv M_1$.
Applying Wick's theorem Eq.~(\ref{Wick}) one gets
\begin{equation}
\label{chi1}
E(M)=\frac12 \trace{(AM)} - \sum_{p<q<r<s} U_{pqrs} \pff{M[p,q,r,s]}.
\end{equation}
We initialize $M=RM_{vac}R^T$, where $R\in SO(2n)$ is chosen at random.
The energy $E(M)$ is then minimized 
by a greedy random walk algorithm.
Each step of the walk generates a random rotation $R\in SO(2n)$
such that $\|R-I\|\le \theta$ for some specified angle $\theta$.
A step is deemed successful if it decreases the value of the objective function,
$E(RMR^T)<E(M)$. In this case $M$ is replaced by $RMR^T$.
Otherwise $M$ remains unchanged. 
The angle $\theta$ is adjusted  
by keeping track of the fraction of successful steps $f$.
We fix some threshold value $f_0$ and
adjust the angle as $\theta\gets \theta (1+ \epsilon)$
if $f\ge f_0$ and $\theta\gets \theta (1-\epsilon)$ if $f<f_0$.
Here  $\epsilon$ is a small constant. 
We empirically found that $\epsilon=0.2$ and $f_0=0.1$ works reasonably well.

\subsection{Rank-$2$ algorithm}

Define a $2\times 2$ Gram matrix $G$ such that
$G_{a,b}=\langle \phi_a|\phi_b\rangle$. We can always
choose the relative phase of $\phi_1$ and $\phi_2$ such that $G$
is a real matrix. Then 
\begin{equation}
\label{rank2Gram}
G=\left[\ba{cc} 1 & g \\ g & 1 \\ \ea \right],
\quad \mbox{where} \quad g=\langle \phi_1|\phi_2\rangle = 
2^{-n/2} \det{(M_1+M_2)}^{1/4}.
\end{equation}
Here we used Eq.~(\ref{ip2}).
Using the generalized Wick's theorem Eq.~(\ref{Wick3}) one gets 
\begin{equation}
\label{rank2energy1}
\langle \phi_a | H_0|\phi_a\rangle = \frac12 \trace{(AM_a)} 
\end{equation}
and
\begin{equation}
\label{rank2energy2}
\langle \phi_2 | H_0|\phi_1\rangle = \frac{g}2 \trace{(A \Delta)}, 
\end{equation}
where
\[
\Delta= (-2I + iM_1 - iM_2)(M_1+M_2)^{-1}.
\]
Using Eqs.~(\ref{Wick},\ref{Wick3}) again one gets
\begin{equation}
\label{rank2energy3}
\langle \phi_a|H_{imp} |\phi_a\rangle = -\sum_{p<q<r<s} U_{pqrs} \cdot  \pff{M_a[p,q,r,s]},
\end{equation}
\begin{equation}
\label{rank2energy4}
\langle \phi_2|H_{imp} |\phi_1\rangle=-g\sum_{p<q<r<s} U_{pqrs} \pff{\Delta[p,q,r,s]}.
\end{equation}

Let $F_0$ and $F_{imp}$ be $2\times 2$ hermitian operators with matrix elements
\begin{equation}
\label{Fs}
\langle a|F_0|b\rangle = \langle \phi_a|H_0|\phi_b\rangle
\quad \mbox{and} \quad
\langle a | F_{imp} |b\rangle = \langle \phi_a|H_{imp} |\phi_b\rangle.
\end{equation}
Here $a,b=1,2$.
Standard linear algebra implies that the objective function
$E(M_1,M_2)$ from Eq.~(\ref{objective}) coincides with the
minimum eigenvalue of a $2\times 2$ matrix 
\[
G^{-1/2}(F_0+F_{imp}) G^{-1/2}.
\]
Equivalently, $E(M_1,M_2)$ is the minimum eigenvalue of 
a generalized eigenvalue problem $(F_0+F_{imp})\psi = \lambda G \psi$.
Thus the above formulas allow one to compute $E(M_1,M_2)$ in time $O(n^3)$.
The most time consuming steps are clearly computing the inverse and the determinant
of $M_1+M_2$.
We then minimize $E(M_1,M_2)$ over $M_1,M_2$ using the same 
greedy random walk algorithm 
as in the rank-$1$ case. The only difference is that at each step
we decide whether to rotate $M_1$ or $M_2$ at random. 

\subsection{Arbitrary rank}

Suppose now that $\chi\ge 3$. 
Let $\phi_0\in \calG_n$ be some fixed reference state
used to compute inner products, see Section~\ref{sec:inner}.
We chose $\phi_0$ as a ground state of $H_0$, but it could be 
an arbitrary Gaussian state. 
For each $a=1,\ldots, \chi$ define 
\begin{equation}
\label{ga}
g_a=\langle \phi_0|\phi_a\rangle.
\end{equation}
We can always choose the overall phase of $\phi_a$ such that $g_a$
are real and $g_a\ge 0$. 
Then Eq.~(\ref{ip2}) implies 
\begin{equation}
\label{ga1}
g_a=2^{-n/2} \det{(M_0+M_a)}^{1/4},
\end{equation}
where $M_0$ is the covariance matrix of the reference state.
Define a Gram matrix $G$ with matrix elements
$G_{a,b}=\langle \phi_a|\phi_b\rangle$, where $1\le a,b\le \chi$.
By normalization, one has $G_{a,a}=1$.
From Eq.~(\ref{ip3}) one gets
\begin{equation}
\label{rank3Gram}
G_{b,a}=\frac{2^n g_a g_b }{\pff{\Delta^{a,b}+M_0}}
\end{equation}
where
\[
\Delta^{a,b}=(-2I + iM_a - iM_b)(M_a+M_b)^{-1}.
\]
Note that $\Delta^{a,a}=M_a$ and $\Delta^{b,a}=(\Delta^{a,b})^*$.
Using the generalized Wick's theorem Eq.~(\ref{Wick3}) one gets 
\begin{equation}
\label{rank3energy2}
\langle \phi_b | H_0|\phi_a\rangle = \frac{G_{b,a}}2 \trace{(A \Delta^{a,b})}
\end{equation}
and
\begin{equation}
\label{rank3energy4}
\langle \phi_b|H_{imp} |\phi_a\rangle=-G_{b,a} \sum_{p<q<r<s} U_{pqrs} \pff{\Delta^{a,b}[p,q,r,s]}
\end{equation}
for all $1\le a,b\le \chi$.
Now we can compute $E(M_1,\ldots,M_\chi)$ by solving a
generalized eigenvalue problem $(F_0+F_{imp})\psi=\lambda G\psi$,
where now $F_0$ and $F_{imp}$ are $\chi\times \chi$
hermitian operators defined by Eq.~(\ref{Fs})
and choosing the minimum eigenvalue $\lambda$.

We then minimize $E(M_1,\ldots,M_\chi)$ using the same 
greedy random walk algorithm as above.
At each step of the walk
we decide which matrix $M_a$ to rotate  at random. 
After each rotation one has to compute the new coefficient $g_a$ and 
properly update matrices $G,\Delta^{a,b},F_0,F_{imp}$.
This takes time $O(\chi n^3)$.

\subsection{SDP lower bound on the ground energy}
\label{sec:numerics}

Let $e_g$ be the exact ground energy of the full Hamiltonian $H$. 
Clearly, the optimal value found by the variational rank-$\chi$ algorithm gives
an upper bound on $e_g$. Can we determine how 
good is this upper bound without computing $e_g$ exactly~?
In this section we describe a simple extension of the rank-$\chi$ algorithm
that gives both upper and lower bounds on $e_g$.
This enables us to benchmark the algorithm for large system sizes 
when  exact diagonalization of $H$ is not feasible. 
We shall obtain a lower bound on $e_g$ using a  semi-definite program (SDP) which is closely related to the 
2-RDM method commonly used in the quantum chemistry~\cite{mazziotti2005variational}.

Let  $\herm{d}$ be the set of all hermitian matrices of size $d\times d$.
The SDP will depend on an integer $N$ and a list of operators
$C_1,\ldots,C_N$ acting on the Fock space $\calH_n$.  
These operators must satisfy only two conditions. First, 
we require that
the full Hamiltonian $H$  can be written as
\begin{equation}
\label{hatH}
H=\sum_{p,q=1}^N H^{(1)}_{p,q} C_p^\dag C_q
\end{equation}
for some coefficients $H^{(1)}_{p,q}$.
Second, we require that  the identity operator $I$ on $\calH_n$
can be written as
\begin{equation}
\label{hatI}
I=\sum_{p,q=1}^N I^{(1)}_{p,q} C_p^\dag C_q
\end{equation}
for some coefficients $I^{(1)}_{p,q}$. 
Define a space of linear dependencies 
\begin{equation}
\label{calL}
\calL=\{ K \in \herm{N} \, : \, \sum_{p,q=1}^N K_{p,q}\, C_p^\dag C_q =0\}.
\end{equation}
Clearly, $\calL$ is a linear subspace (over $\RR$).
Let $d=\dim{(\calL)}$ and fix a basis $K^1,\ldots,K^d\in \calL$.
Consider the following semi-definite program:
\begin{equation}
\label{SDPprimal}
\ba{rcl}
\mbox{\bf variable} & : & X\in \herm{N} \\
\mbox{\bf minimize} & : & \trace{(H^{(1)}X)} \\
\mbox{\bf subject to} & : & X\ge 0 \\
&& \trace{(I^{(1)} X)}=1 \\
&&   \trace{(K^\alpha X)}=0 \quad \mbox{for $1\le \alpha\le d$.}\\
\ea
\end{equation}
Here we view the coefficients $H^{(1)}_{p,q}$ and $I^{(1)}_{p,q}$ as matrix elements
of hermitian $N\times N$ matrices $H^{(1)}$ and $I^{(1)}$.
\begin{lemma}
\label{lemma:SDP}
For any normalized state $\psi \in \calH_n$ there exists $X\in \herm{N}$
such that $X$ is a  feasible solution of the SDP defined above and
\begin{equation}
\label{eq1}
\langle \psi|H|\psi\rangle = \trace{(H^{(1)}X)}.
\end{equation}
\end{lemma}
\begin{proof}
Choose $X\in \herm{N}$ such that $X_{p,q}=\langle \psi|C_q^\dag C_p|\psi\rangle$
for all $1\le p,q\le N$. A direct inspection shows that
$X$ obeys all constraints of the SDP and Eq.~(\ref{eq1}) holds.
\end{proof}
The lemma implies that the optimal value of the SDP provides a lower bound on the
ground energy $e_g$. 
One can easily check that this lower bound is tight for
 free fermion Hamiltonians $H$ defined in Eq.~(\ref{bath1})
if one chooses $N=2n$ and $C_p=c_p$ for $p=1,\ldots,2n$.
Here $c_p$ are the Majorana operators defined in Eq.~(\ref{Majorana}).
Furthermore, one can easily check that the lower bound is tight for an arbitrary Hamiltonian if
one chooses $N=4^n$ and the list of operators $\{C_p\}$ 
includes all  Majorana monomials $c(x)$ with $x\in\{0,1\}^{2n}$.
Our choice of the operators $C_p$ for quantum impurity models
will be a mixture of these two extreme cases.

Let $X_{opt}$ be the optimal solution of the program Eq.~(\ref{SDPprimal}).
Standard duality arguments show that $\trace{(H^{(1)}X_{opt})}\ge y_0$
whenever
\begin{equation}
 y_0 I^{(1)} + \sum_{\alpha=1}^d y_\alpha K^\alpha \le H^{(1)} \label{SDPdual}
\end{equation}
for some  $y_0,y_1,\ldots,y_d\in \RR$. 
Here $y_0$ and $y_\alpha$ are Lagrangian multipliers for the constraints
$\trace{(I^{(1)} X)}=1$ and $\trace{(K^\alpha X)}=0$.
Thus any feasible solution of Eq.~(\ref{SDPdual}) gives a lower bound
on the ground energy, namely, $e_g\ge y_0$.

Consider  an impurity model $H=H_0+H_{imp}$ such that $H_{imp}$ acts
on the modes $c_1,\ldots,c_m$.
How do we choose $N$ and the operators $C_p$ to define the above SDP~?
Let $\psi$ be the exact ground state of $H$.
By Corollary~\ref{cor:gapped}, one can localize all bath excitations present 
in $\psi$ on a small subset of modes by some Gaussian unitary $U$, that is, 
\begin{equation}
\label{approx1}
\| \psi - U|\phi\otimes 0^{n-k}\rangle \| \le \delta.
\end{equation}
Here $\delta$ is a small precision parameter
and $k$ is determined by Eqs.~(\ref{chi}). 
Suppose first that we can guess the integer $k$ and the Gaussian unitary operator $U$ that
achieve the desired approximation precision $\delta$.
Empirically, we found that the following choice of the operators $C_p$
achieves a very good lower bound on $e_g$ at a reasonable computational cost.  
First, define an extended list of Majorana operators $d_1,\ldots,d_{m+2n}$ such that 
\begin{equation}
\label{dops1}
d_1=c_1, \qquad d_2=c_2, \qquad  \ldots \qquad d_m=c_m 
\end{equation}
and
\begin{equation}
\label{dops2}
d_{m+1}=Uc_1 U^\dag,  \qquad  d_{m+2}=Uc_2U^\dag, \qquad \ldots  \qquad d_{m+2n}=Uc_{2n} U^\dag.
\end{equation}
Note that the state $U|\phi\otimes 0^{n-k}\rangle$ in Eq.~(\ref{approx1}) can be obtained from the vacuum
$|0^n\rangle$ by some combination of the operators $d_{m+1},\ldots,d_{m+2k}$.
Let us choose   
\begin{equation}
\label{SDPN}
N= 2n+{ m+2k \choose 3}
\end{equation}
and choose the operators $C_1,\ldots,C_N$ as 
\begin{equation}
\label{Clist}
\{ C_p\} =\{ d_{m+1} ,\ldots,d_{m+2n}  \}  \cup  \{d(x)\, : \, x\in \{0,1\}^{m+2k}, \quad |x|=3 \}.  
\end{equation}
Here  $d(x)$ denotes the product of all operators $d_j$ with $x_j=1$.
One can easily check that the list Eq.~(\ref{Clist}) satisfies condition Eq.~(\ref{hatH})
for  any impurity Hamiltonian $H_{imp}$ with   quartic interactions acting on the first $m$ modes
$c_1,\ldots,c_m$. It is also clear that this list satisfies condition Eq.~(\ref{hatI})
since, for example, $I=d_{m+1}^\dag d_{m+1}$.

In practice, we can  obtain a good guess  of $k$ and $U$ in Eq.~(\ref{approx1})
by examining the optimal variational state $\psi_{opt}\in \calH_n$ found by the rank-$\chi$ algorithm.
Namely, consider the covariance matrix 
\[
M_{p,q}=(-i/2) \langle \psi_{opt} | c_p c_q - c_q c_p |\psi_{opt}\rangle.
\]
Using the generalized Wick's theorem Eq.~(\ref{Wick3}) one can compute $M$
in time $O(\chi^2 n^3)$ since $\psi_{opt}$ is specified as a superposition of $\chi$
Gaussian states. Let eigenvalues of $M$ be $\pm i s_j$, where 
\[
0\le s_1\le s_2 \le \ldots\le s_n\le 1
\]
are singular values of $M$. 
If $\psi_{opt}$ were the exact ground state, Theorem~\ref{thm:2} 
implies\footnote{In this paper we use two versions of a covariance matrix --
one for Majorana operators, see Eq.~(\ref{M}), and the other for creation-annihilation
operators, see Eq.~(\ref{Cmat}). If the latter has eigenvalues $\sigma_j\in [0,1]$
then the former has eigenvalues $\pm i(1-2\sigma_j)$, where $j=1,\ldots,n$.}
that the singular values $s_j$ rapidly approach one as $j$
increases. We found empirically that this is also true for the optimal variational
state $\psi_{opt}$. Let $k$ be the smallest integer  such that $s_j\ge 1-\epsilon$ for
all $j\ge k$, where $\epsilon$ is some fixed  precision parameter. In our simulations we choose
$\epsilon=10^{-4}$. One can always transform $M$ into a block-diagonal form
such that 
\[
R^T MR  = \bigoplus_{j=1}^n \left[ \ba{cc} 0 & s_j \\ -s_j & 0 \\ \ea \right]
\]
for some $R\in SO(2n)$. Let $U$ be  the Gaussian
unitary that implements the rotation $R$, see Eq.~(\ref{UvsR}). Then
\[
\langle \psi_{opt}U |a_j^\dag a_j|U^\dag \psi_{opt}\rangle = \frac12 (1-s_j) \le \epsilon/2,
\]
that is, all bath excitations present in the state $U^\dag |\psi_{opt}\rangle$ are "localized"
on the first $k$ modes and we can write $|\psi_{opt}\rangle\approx U|\phi\otimes 0^{n-k}\rangle$
for some $k$-mode state $\phi$.

 We define the SDP  according to
Eqs.~(\ref{dops1},\ref{dops2},\ref{SDPN},\ref{Clist}), where $U\in \calC_n$ is the Gaussian unitary constructed above.
This defines the SDP lower bound on the ground energy $e_g$ 
for a given variational state $\psi_{opt}$ and a precision parameter $\epsilon$.

To benchmark the variational algorithm and the SDP method we chose the single impurity Anderson
model~\cite{Anderson61}. This model has a Hamiltonian $H=H_0+H_{imp}$, where
\[
H_{imp}=U a_1^\dag a_1 a_2^\dag a_2 = \frac{U}4\left( - c_1 c_2 c_3 c_4  + ic_1 c_2 + ic_3 c_4 + I \right)
\]
for some $U\ge 0$. This corresponds to the impurity of size $m=4$
(we note  that $m=4$ is the smallest impurity size that gives rise to a non-trivial impurity model
since any even  Hamiltonian acting on $m\le 3$ Majorana modes must be quadratic).
For simplicity, we chose  $H_0$ as the critical Majorana chain~\cite{kitaev2001unpaired}
 with periodic boundary conditions:
\[
H_0=i\sum_{j=1}^{2n} c_j c_{j+1}, \qquad c_{2n+1}\equiv c_1.
\]
It is well-known~\cite{schultz1964two,Pfeuty1970} that  $H_0$ has a unique ground state with parity $P=1$
and the spectral gap proportional to $n^{-1}$. 

Let $e_g$ be the ground energy of $H$. For each choice of  $n$ and $U$ we numerically
computed two numbers $e_g^{+}$ and $e_g^{-}$ such that 
\[
e_g^{-}\le e_g\le e_g^{+}.
\]
The upper bound 
$e_g^{+}$ is the minimum energy found by the rank-$\chi$ variational algorithm.
In our simulations we only used  $\chi=1$ and $\chi=2$.
The lower bound  $e_g^{-}$ was computed 
using a combination of the rank-$2$ variational algorithm and  the SDP method with $k=4$
as described above. 
The corresponding SDP has size  $N=2n+220$,   see Eq.~(\ref{SDPN},\ref{Clist}).
Figure~\ref{fig:SIAM} shows the gap $e_g^{+}-e_g^{-}$ as a function of the system size $n$
for several values of $U$.    
In all cases the rank-$2$ algorithm approximates the ground energy
within an additive error less than $2\times 10^{-6}$
and we were able to estimate $e_g$ within the first eight significant digits,
see Eq.~(\ref{eq:table}).

\begin{figure}[htbp]
\centerline{\includegraphics[width=14cm]{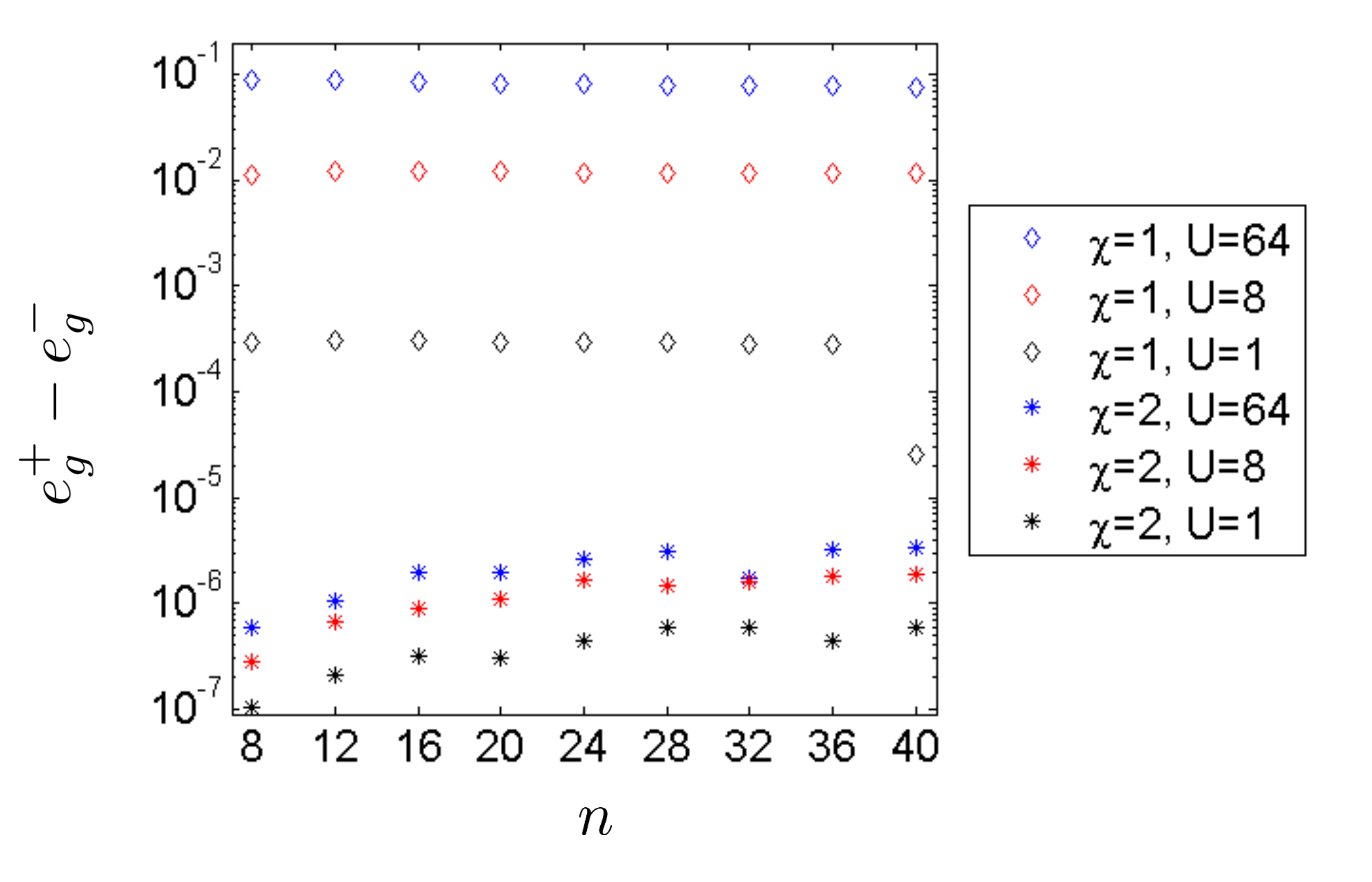}}
\caption{Separation between the upper and lower bounds $e_g^{\pm}$
on the ground energy of the single impurity Anderson model.
}
\label{fig:SIAM}
\end{figure}

\begin{equation}
\label{eq:table}
\begin{tabular}{r|c|c|c|}
  & $U=1$ & $U=8$ & $U=64$ \\
\hline
 $n=8$ & $-10.00932(5)$ & $-9.89010(8)$& $-9.81220(9)$ \\
$n=16$ & $-20.25487(5)$ & $-20.11633(4)$ & $-20.02426(8)$ \\
$n=24$ & $-30.46084(8)$ & $-30.31627(5)$ & $-30.21953(6)$ \\
$n=32$ & $-40.65683(1)$ & $-40.50931(8)$ & $-40.41024(6)$ \\
$n=40$ & $-50.84854(5)$ & $-50.69954(7)$  & $-50.59907(2)$ \\
\end{tabular}
\end{equation}

\bibliographystyle{unsrt}
\bibliography{Impurity}

\appendix

\section{Proof of the inner product formulas}
\label{app:inner}

In this appendix we prove the inner product formulas and the generalized Wick's theorem
of Section~\ref{sec:inner}.
Let $A$ be an arbitrary operator acting on the Fock space $\calH_n$.
Expanding $A$ in the basis of $4^n$ Majorana monomials $c(x)$,
$x\in \{0,1\}^{2n}$, 
one gets 
\[
A=\sum_x A_x c(x), \quad c(x)\equiv c_1^{x_1} \cdots c_{2n}^{x_{2n}}
\]
for some coefficients $A_x$.
Here and in the rest of this section all sums over binary
strings have range  $\{0,1\}^{2n}$. An operator $A$ as above can be described by a
generating function 
\[
A(\theta)=\sum_x A_x \theta(x), \quad \theta(x)\equiv \theta_1^{x_1} \cdots \theta_{2n}^{x_{2n}}.
\]
where $\theta_1,\ldots,\theta_{2n}$ are  formal variables obeying commutation
rules of the Grassmann algebra, that is, 
\[
\theta_p^2=0, \quad \theta_p \theta_q =- \theta_q \theta_p.
\]
An integral over Grassmann variables $\int D\theta$  is a linear functional defined by
\[
\int D\theta\cdot  \theta(x) =  \left\{ \ba{rcl}
1 &\mbox{if} & x=(11\ldots 1)\equiv \overline{1}, \\
0 & & \mbox{otherwise} \\
\ea \right.
\]
The action of $\int D\theta$ extends by linearity to an arbitrary function $A(\theta)$
by expanding $A(\theta)$ in the basis of Grassmann monomials $\theta(x)$. 
More details on the formalism of Grassmann variables can be found in~\cite{francesco2012conformal,bravyi2004lagrangian}.
Let $M$ be an anti-symmetric matrix of size $2n$. 
We shall consider $\theta=(\theta_1,\ldots,\theta_{2n})$ as a column vector and 
write $\theta^T$ for the corresponding row vector. Then  
\[
\theta^T M \theta \equiv \sum_{p,q=1}^{2n} M_{p,q} \theta_p \theta_q.
\]
Suppose $\phi\in \calG_n$ is a Gaussian state with a covariance matrix $M$
and let $A=|\phi\rangle \langle \phi|$.
Wick's theorem Eq.~(\ref{Wick}) can be rephrased as 
\begin{equation}
\label{generating}
A(\theta)= 2^{-n} \exp{\left[  -\frac{i}2 \theta^T M \theta \right]}.
\end{equation}
We shall need 
the following well-known formulas for Gaussian integrals:
\begin{equation}
\label{Gauss1}
\int D\theta \cdot \exp{\left[  \frac12 \theta^T M \theta \right]}=\pff{M},
\end{equation}
\begin{equation}
\label{Gauss2}
\int D\eta \cdot  \exp{\left[ \theta^T \eta + \frac12 \eta^T M \eta \right]}=\pff{M}\exp{\left[  \frac12 \theta^T M^{-1} \theta \right]}.
\end{equation}
Here $\eta=(\eta_1,\ldots,\eta_{2n})$ is another vector of $2n$ Grassmann variables.

Now we can easily prove Eq.~(\ref{ip2}).
Define a bilinear form
\[
\Gamma(x,y)=\sum_{1\le p<q\le 2n} x_p y_q {\pmod 2}.
\]
Here $x,y\in\{0,1\}^{2n}$. Then 
$\trace{\left[ c(x) c(y)\right]} = 2^n (-1)^{\Gamma(x,x)}$
if $x=y$ 
and $\trace{\left[ c(x) c(y)\right]}=0$ if $x\ne y$.
Furthermore, 
\[
\int D(\theta \eta) \cdot \theta(x) \eta(y) \exp{(\theta^T \eta)} =(-1)^n (-1)^{\Gamma(x,x)}
\]
if $x=y$ and the integral is zero if $x\ne y$. Here both $\theta$ and $\eta$
are vectors of $2n$ Grassmann variables
and $\int D(\theta \eta) \equiv \int D\theta \int D\eta$.
Let $A$, $B$ be arbitrary operators on $\calH_n$
and $A(\theta)$, $B(\eta)$ be their generating functions. The above shows that 
\begin{equation}
\label{double1}
\trace{(AB)} = (-2)^n \int D(\theta \eta) \cdot A(\theta) B(\eta)  \exp{(\theta^T \eta)}.
\end{equation}
Suppose now that $A,B$ are projectors onto some pure Gaussian states
$\phi_0,\phi_1$ with covariance matrices $M_0,M_1$ and parity $\sigma$.
Substituting Eq.~(\ref{generating}) into Eq.~(\ref{double1}), taking the integral
over $\eta$ using Eq.~(\ref{Gauss2}) and taking the integral over
$\theta$ using Eq.~(\ref{Gauss1}) one gets
\[
\trace{(AB)} =2^{-n} \pff{M_1} \pff{M_0+M_1}=\sigma 2^{-n} \pff{M_0+M_1}.
\] 
Here we noted that $M_1^{-1}=-M_1$.

Next let us  prove Eq.~(\ref{ip3a}). Consider any even-weight strings $x,y,z\in \{0,1\}^{2n}$ such that 
\begin{equation}
\label{xyz}
x+y+z=\overline{1}.
\end{equation}
Here and below  addition of binary strings is performed modulo two.
Let $P\equiv (-i)^n c(\overline{1})$ be the total parity operator
defined in Eq.~(\ref{Parity}). A simple algebra shows that 
\begin{equation}
\label{trace3a}
\trace{\left[  P c(x) c(y) c(z) \right]}=  i^n 2^n(-1)^{\Gamma(x,y)+ \Gamma(y,z) + \Gamma(z,x)}.
\end{equation}
Furthermore, the trace is zero whenever  $x+y+z\ne \overline{1}$.

Suppose $\theta,\eta,\mu$ are vectors of $2n$ Grassmann variables. 
Let $D(\theta \eta \mu)\equiv D\theta D\eta D\mu$. 
A simple algebra shows that 
\begin{equation}
\label{trace3b}
\int D(\theta \eta \mu) \cdot \theta(x) \eta(y) \mu(z) \exp{\left[ \theta^T \eta + \eta^T \mu + \mu^T \theta \right]}
=(-1)^n (-1)^{\Gamma(x,y)+ \Gamma(y,z) + \Gamma(z,x)}.
\end{equation}
Furthermore, the integral is zero whenever  $x+y+z\ne \overline{1}$.

We shall say that an operator acting on $\calH_n$ is {\em even}
if its expansion in the basis of Majorana monomials
includes only even-weight monomials. 
Let $A,B,C$ be any even  operators on $\calH_n$
and $A(\theta)$, $B(\eta)$, $C(\mu)$ be their generating functions.
Comparing Eqs.~(\ref{trace3a},\ref{trace3b}) one concludes that 
\begin{equation}
\label{PABC}
\trace{(PABC)} =(-i)^n 2^n \int D(\theta \eta \mu) \cdot A(\theta) B(\eta) C(\mu)
 \exp{\left[ \theta^T \eta + \eta^T \mu + \mu^T \theta \right]}
\end{equation}
Suppose now that $A,B,C$ are projectors onto some pure Gaussian states
$\phi_0,\phi_1,\phi_2$ with covariance matrices $M_0,M_1,M_2$
respectively. Recall that any Gaussian state has a fixed parity, that is,
$P\phi_\alpha = \sigma_\alpha \phi_\alpha$ for some $\sigma_\alpha=\pm 1$.
Clearly $\trace{(PABC)}=0$ unless all states $\phi_0,\phi_1,\phi_2$ have the same 
parity: $\sigma_1=\sigma_2=\sigma_3\equiv \sigma$. 
 Using Eqs.~(\ref{generating},\ref{Gauss1}) one gets
\[
\trace{(PABC)} =\sigma \langle \phi_2 |\phi_0\rangle \langle \phi_0|\phi_1\rangle \langle \phi_1|\phi_2\rangle
= (-i)^n 4^{-n}
 \pff{\left[  \ba{ccc}
-iM_0 & I  & -I \\
-I & -iM_1 &  I \\
I & -I & -iM_2 \\
\ea\right].}
\]
This is equivalent to Eq.~(\ref{ip3a}).
To prove Eq.~(\ref{ip3}) let us come back to Eq.~(\ref{PABC}) and 
take a partial integral over $\mu$. Applying Eq.~(\ref{Gauss2}) 
and noting that $\pff{-iM_2}=(-i)^n \sigma$ one gets
\[
\trace{(PABC)} =\sigma (-1)^n \int D(\theta \eta) \cdot A(\theta) B(\eta) 
\exp{\left[\theta^T \eta  -\frac{i}2 (\eta-\theta)^T M_2 (\eta-\theta)  \right]}.
\]
Next let us take a partial integral over $\eta$.
Applying Eq.~(\ref{Gauss2}) one gets
\[
\trace{(PABC)} =\sigma i^n 2^{-n} \pff{M_1+M_2} \int D\theta \cdot A(\theta) 
\exp{\left[ -\frac{i}2 \theta^T \Delta \theta\right]},
\]
where
\[
\Delta=M_2 - (I-iM_2)^T (M_1+M_2)^{-1} (I-iM_2).
\]
Taking into account that $M_a^{-1}=-M_a$ one can rewrite  $\Delta$ as
\[
\Delta=(-2I + iM_1- iM_2)(M_1+M_2)^{-1}.
\]
Finally, taking the integral over $\theta$ using Eq.~(\ref{Gauss1}) yields
\[
\trace{(PABC)} =\sigma 4^{-n} \pff{M_1+M_2} \pff{M_0+\Delta}.
\]
This is equivalent to Eq.~(\ref{ip3}) since $PA=\sigma A$.

Next let us  prove the generalized Wick's theorem Eq.~(\ref{Wick2}).
Choose a Gaussian unitary $U$ such that $Uc(x) U^\dag =c_1c_2\cdots c_w$,
where $w$ is the Hamming weight of $x$. Let $R\in O(2n)$ be the corresponding
rotation defined by Eq.~(\ref{UvsR}). Replacing each state $\phi_a$ by $U\phi_a$
is equivalent to replacing the covariance matrix $M_a$ by $R^T M_a R$.
Thus suffices to prove Eq.~(\ref{Wick2}) for the special case $c(x)=c_1c_2\cdots c_w$.

Define $C=c(x) |\phi_2\rangle\langle \phi_2|$. We would like to compute
the generating function $C(\mu)$ describing $C$.
Using the standard Wick's theorem Eq.~(\ref{Wick}) one gets
\[
|\phi_2\rangle\langle \phi_2|=2^{-n}\sum_y \pff{-iM_2[y]} c(y).
\]
Here we noted that $\trace{(c(y)^\dag c(z))}=2^n \delta_{y,z}$
and $c(y)^\dag =(-1)^{\Gamma(y,y)} c(y)$. Therefore
\[
C=c(x) |\phi_2\rangle\langle \phi_2| =2^{-n}\sum_y \pff{-iM_2[y]} (-1)^{\Gamma(y,x)} c(x+y).
\]
It follows that $C$ has a generating function
\[
C(\mu)=2^{-n}\sum_y \pff{-iM_2[y]} (-1)^{\Gamma(y,x)} \mu(x+y).
\]
Here $\mu=(\mu_1,\ldots,\mu_{2n})$ is a vector of $2n$ Grassmann variables.
Let us construct a linear map $\Phi_x$ acting on Grassmann monomials such that 
\[
\Phi_x\cdot \mu(y)=(-1)^{\Gamma(y,x)}\mu(x+ y)
\]
for all $y\in \{0,1\}^{2n}$.
This would imply that 
\[
C(\mu)=\Phi_x \cdot 2^{-n} \exp{\left[ -\frac{i}2 \mu^T M_2 \mu \right]}.
\]
It will be convenient to represent $\mu=(\mu',\mu'')$, where
$\mu'$ has length $w$ and $\mu''$ has length $2n-w$. Then 
simple algebra shows that 
\begin{equation}
\label{PhiMap}
(\Phi_x\cdot f)(\mu',\mu'') = \int D\tau\cdot
\prod_{j=1}^w (1+\tau_j \mu'_j) f(\tau,\mu'')=
 \int D\tau\cdot \exp{\left[ \tau^T \mu' \right]} f(\tau,\mu'').
\end{equation}
Here $\tau=(\tau_1,\ldots,\tau_w)$ is a vector of $w$ Grassmann variables.
Let $A(\theta)$ and $B(\eta)$ be the generating functions of projectors onto $\phi_0$
and $\phi_1$. Applying Eq.~(\ref{PABC}) and using the new definition of $C$ one arrives at
\begin{equation}
\label{WickGen1}
\trace{(PABC)}=\sigma \langle \phi_0|\phi_1\rangle \langle \phi_1|c(x)|\phi_2\rangle \langle \phi_2|\phi_0\rangle
\end{equation}
and
\begin{equation}
\label{WickGen2}
\trace{(PABC)}=(-i)^n 2^n \int D(\theta \eta \mu) \cdot A(\theta) B(\eta) C(\mu)
 \exp{\left[ \theta^T \eta + \eta^T \mu + \mu^T \theta \right]}.
\end{equation}
Combining Eqs.~(\ref{PhiMap},\ref{WickGen2}) one arrives at
\begin{eqnarray}
\label{WickGen3}
\trace{(PABC)}=(-i)^n 4^{-n} \int D(\theta \eta \mu \tau )
\cdot \exp{\left[
-\frac{i}2 \theta^T M_0 \theta  -\frac{i}2 \eta^T M_1 \eta  -\frac{i}2 (\tau,\mu'')^T M_2 (\tau,\mu'') \right.} \nonumber \\
\left.+ \tau^T \mu' + \theta^T \eta + \eta^T \mu + \mu^T \theta  \right]. \nonumber
\end{eqnarray}
Evaluating the integral using Eq.~(\ref{Gauss1}) gives
\begin{equation}
\label{WicjGen4}
\trace{(PABC)}=(-i)^n 4^{-n} \pff{ \ba{c|c|c|c}
-iM_0  & I & -I &  \\
\hline
-I & -iM_1 & I  &  \\
\hline
I & -I & -i D_x M_2  D_x & -J_x^T -iD_x M_2 J_x^T  \\
\hline
 &  &  J_x- i J_x M_2 D_x  & -i J_x M_2 J_x^T  \\
\ea
},
\end{equation}
where $J_x$ is a matrix of size $w\times 2n$ such that 
$(J_x)_{i,j}=1$ if $j$ is the position of the $i$-th nonzero of $x$
and $(J_x)_{i,j}=0$  otherwise. Furthermore, $D_x$ is a diagonal
matrix of size $2n$ such that $(D_x)_{j,j}=1-x_j$.
Combining this and Eq.~(\ref{WickGen1}) proves the generalized
Wick's theorem Eq.~(\ref{Wick2}).

To derive the simplified expression Eq.~(\ref{Wick3})
it suffices to substitute $C=c(x)$ into Eq.~(\ref{PABC}) and evaluate
the integrals using Eqs.~(\ref{Gauss1},\ref{Gauss2}).

%%%%%%%%%%%%%%%%%%%%%%%%%%%%%%%%%%%%%%%%
\section{Fast norm estimation}
\label{app:normest}
Here we describe an algorithm for estimating the norm 
of a state represented as a superposition of $\chi\ll 2^n$ Gaussian states. It is a Gaussian version of an algorithm presented in Ref.~\cite{BG16} for superpositions of stabilizer states. 
Suppose our goal is to compute the norm of a state
\[
|\psi\rangle=\sum_{a=1}^{\chi} x_a |\phi_a\rangle
\]
where $x_a\in \mathbb{C}$ and $\phi_a$ are normalized Gaussian states. The states $\{\phi_a\}$ are not assumed to be orthogonal. We assume that we are explicitly given the coefficients $\{x_a\}$ and that each Gaussian state $\phi_a$ is specified by its covariance matrix and its inner product with a reference Gaussian state $|r\rangle$.

A naive classical algorithm for computing $\langle \psi|\psi\rangle$ is to expand
\[
\langle \psi|\psi\rangle=\sum_{a,b=1}^{\chi} \bar{x}_b x_a\langle \phi_b |\phi_a\rangle
\]
and then to compute each term in the sum using the inner product algorithm from Section \ref{sec:inner}. The runtime of this naive algorithm scales quadratically with $\chi$.  Below we present a Monte Carlo algorithm which estimates $\langle \psi|\psi\rangle$  with runtime scaling only linearly with $\chi$. In particular, we prove the following lemma. In the lemma it is assumed that the state $\psi$ is classically represented as described above.

\begin{lemma}
\label{lem:normestimation}
Let $|\psi\rangle=\sum_{a=1}^{\chi} x_a |\phi_a\rangle$ where $\phi_a$ are Gaussian states and $x_a\in \mathbb{C}$. There is a classical algorithm which takes as input a precision parameter $\epsilon>0$, a failure probability $p_f>0$, and the state $\psi$ and outputs an estimate $\xi$ such that, with probabilty at least $1-p_f$ we have
\[
(1-\epsilon)\|\psi\|^2 \leq \xi\leq (1+\epsilon)\|\psi\|^2
\]
The runtime of the algorithm is $O(n^{7/2} \epsilon^{-2} p_f^{-1} \chi)$. 
\end{lemma}

In the remainder of this Appendix we prove the lemma.  As noted above, this result is a Gaussian analog of a similar norm estimation algorithm given in Ref. \cite{BG16} for superpositions of stabilizer states. However, while that algorithm uses the fact that stabilizer states form a $2$-design, one can show that, strictly speaking, no ensemble of Gaussian states has this property. To prove lemma \ref{lem:normestimation} we follow a technique from Ref. \cite{CLLW15} which established a criterion called \textit{Pauli mixing} for an ensemble of Clifford unitaries to form a $2$-design. We define an analogous criterion which we call \textit{Majorana mixing} for ensembles of Gaussian unitaries and we show how to use this criterion (rather than the $2$-design property) in the norm estimation algorithm.

\subsection{Symmetric tensor products of Gaussian states}

Recall that the Majorana monomials
\[
c(x)=c_1^{x_1}c_2^{x_2}\cdots c_{2n}^{x_{2n}} \qquad \quad x\in \{0,1\}^{2n}
\]
form a basis (over $\mathbb{R}$) for the space of Hermitian operators acting on the Hilbert space $\mathcal{H}_n$ of $n$ fermi modes. For any state $\rho\in \mathcal{H}_n\otimes \mathcal{H}_n$ we denote $\rho_{ab}$ for the coefficients of its expansion
\begin{equation}
\rho=\sum_{a,b\in\{0,1\}^{2n}} \rho_{ab}\; c(a)\otimes c(b)\qquad \quad
\rho_{ab} =\frac{1}{2^{2n}} \trace{\left(\rho \;c(a)\otimes c(b)\right)}.
\label{eq:expandrho}
\end{equation}

Define an operator
\[
\Lambda=\sum_{p=1}^{2n}c_p\otimes c_p
\]
which acts on $\mathcal{H}_n\otimes \mathcal{H}_{n}$, and write $\mathcal{K}$ for its null space. To construct a basis for $\mathcal{K}$ first note that the commuting operators $c_p\otimes c_p$ can be simultaneously diagonalized. We define a complete orthonormal basis $\{\Psi_x: x\in \{0,1\}^{2n}\}$ of $\mathcal{H}_n\otimes \mathcal{H}_n$, where
\[
c_p\otimes c_p|\Psi_x\rangle=(-1)^{x_p}|\Psi_x\rangle \qquad \quad p=1,2,\ldots,2n.
\]
A subset of these basis vectors span $\mathcal{K}$
\begin{equation}
\mathcal{K}=\mathrm{span}\{\Psi_x: |x|=n\},
\label{eq:spanL}
\end{equation}
and the projector onto $\mathcal{K}$ is 
\begin{equation}
\label{eq:pik}
\Pi_\mathcal{K}=\sum_{|x|=n}|\Psi_x\rangle\langle \Psi_x|=\frac{1}{2^{2n}}\sum_{|x|=n}\prod_{p=1}^{2n} \big(1+(-1)^{x_p}c_p\otimes c_p\big).
\end{equation}
We now derive a set of constraints on the coefficients $\rho_{ab}$ of a state $\rho\in \mathcal{K}$. 
\begin{prop}
There exist real numbers $\{F_0,F_1,\ldots,F_{2n}\}$ such that
\[
\sum_{a\in\{0,1\}^{2n}: |a|=k}\rho_{aa}=F_k \qquad \quad k=0,1,\ldots,2n
\]
for all $\rho \in \mathcal{K}$ with $\trace{\rho}=1$. 
\label{prop:F}
\end{prop}
\begin{proof}
Define $B_{-1}=B_{2n+1}=0$ and
\[
B_k=\sum_{a\in\{0,1\}^{2n}:|a|=k} c(a)\otimes c(a) \qquad \quad k=0,\ldots,2n.\\
\]
The following recursion is easily verified:
\begin{equation}
\Lambda B_k=(k+1)B_{k+1}+(2n-k+1)B_{k-1} \qquad \quad k=0,\ldots,2n.
\label{eq:Brec}
\end{equation} 
Since $\rho\in \mathcal{K}$ we have $\trace{\left( \rho \Lambda B_k\right)}=0$ for $0\leq k\leq 2n$. Using Eq.~\eqref{eq:Brec} and evaluating the trace gives
\begin{equation}
(k+1)F_{k+1}+(2n-k+1)F_{k-1}=0 \qquad \quad k=0,\ldots,2n,
\label{eq:lin}
\end{equation}
where $F_{-1}=F_{2n+1}=0$ and $F_j=\sum_{|a|=j}\rho_{aa}$ for $j=0,1,\ldots,2n$. The normalization constraint fixes $F_0=\frac{1}{2^{2n}}$. From Eq.~\eqref{eq:lin} we see that $F_1=F_3=\ldots=F_{2n-1}=0$ and that 
\[
F_{2j}=\frac{1}{2^{2n}}\prod_{m=1}^j \left(\frac{m-n-1}{m}\right)=\frac{(-1)^j}{4^n} {n\choose j}
\]
for $j=1,\ldots,n$.
\end{proof}
We are interested in the subspace $\mathcal{K}$ because it contains any symmetric tensor product of fermionic Gaussian states.

\begin{prop}
\label{prop:lambda}
Let $|\Phi\rangle=|\phi\otimes \phi\rangle$ where $\phi\in \mathcal{H}_n$ is a fermionic Gaussian state. Then $|\Phi\rangle\in \mathcal{K}$.
\end{prop}
\begin{proof}
First consider the special case in which $|\phi\rangle=|y\rangle$ is a standard basis state. Here $y\in \{0,1\}^n$. In this case $\Lambda|\Phi\rangle=0$ follows directly from
\[
c_{2j-1}\otimes c_{2j-1}|y\otimes y \rangle=-c_{2j}\otimes c_{2j}|y \otimes y\rangle \qquad \quad j=1,2,\ldots,n.
\]
For general Gaussian $\phi$, let $U$ be a Gaussian unitary and $y\in \{0,1\}^n$ such that $|\phi\rangle=U|y\rangle$. Then $Uc_pU^{\dagger}=\sum_{j} R_{pj}c_j$ for some $R\in \mathrm{O}(2n)$, and
\[
U\otimes U\Lambda U^\dagger \otimes U^\dagger=\sum_{j,k=1}^{2n}R_{pj}R_{pk}c_j\otimes c_k=\Lambda.
\] 
Therefore
\[
\Lambda|\phi\otimes\phi\rangle=\Lambda  U \otimes U|y\otimes y\rangle=U \otimes U \Lambda  |y\otimes y\rangle=0.
\]
\end{proof}

\subsection{Ensembles of Gaussian unitaries and Majorana mixing}
In this section we consider probability distributions over Gaussian unitaries and we define a notion of \textit{Majorana mixing} for such ensembles. We then show that a certain channel $\mathcal{T}_\mathcal{E}$ derived from any Majorana mixing ensemble $\mathcal{E}$ satisfies 
\[
\mathcal{T}_\mathcal{E}(|0^n\otimes 0^n\rangle\langle 0^n\otimes 0^n|)=\frac{\Pi_\mathcal{K}}{\dim{\mathcal{K}}}.
\]
The definitions and proof strategy are based on reference \cite{CLLW15} where a similar notion of Pauli mixing is shown to be sufficient for an ensemble of Clifford unitaries to  form a unitary $2$-design.

For each $k=0,1,\ldots,2n$ define  
\[
T_k={2n \choose {k}}^{-1} \sum_{y\in \{0,1\}^{2n}:|y|=k}c(y)\otimes c(y).
\]
\begin{dfn}
A probability distribution $\mathcal{E}=\{p_i,U_i\}$ over Gaussian unitaries is said to be Majorana mixing if
\[
\sum_{i}p_i U_i\otimes U_i \left[c(x)\otimes c(x)\right]U^\dagger_i\otimes U_i^\dagger =T_{|x|}
\]
for all $x\in \{0,1\}^{2n}$.
\end{dfn}

Let us now present a simple example of an ensemble $\mathcal{E}^{\star}$ which is Majorana mixing. To begin we define an ensemble $\{p_i,R_i\}$ of $\mathrm{O}(2n)$ matrices. To choose a matrix from this ensemble we simply select a uniformly random permutation $\pi\in \mathcal{S}_{2n}$ and let $R$ be the corresponding $2n\times 2n$ permutation matrix. Next define $\mathcal{E}^{\star}=\{p_i,U_i\}$ where $U_i$ is a Gaussian unitary such that $U_ic_pU_i^{\dagger}=\sum_{j} (R_i)_{pj}c_j$. To see that $\mathcal{E}^\star$ is Majorana mixing, note that
\[
U_i\otimes U_i \left[c(x)\otimes c(x)\right]U^\dagger_i\otimes U_i =c(\pi_i(x))\otimes c(\pi_i(x))
\]
Here the permutation $\pi_i$ associated to $U_i$ permutes the bits of $x$. In our ensemble $\mathcal{E}^\star$ the permutation $\pi_i$ is uniformly random and therefore $c(x)\otimes c(x)$ is mapped to $c(y)\otimes c(y)$ where $y$ is uniformly random among bit strings with the same Hamming weight as $x$. This shows that $\mathcal{E}^\star$ is Majorana mixing. 

For any distribution $\mathcal{E}=\{p_i,U_i\}$ we define a channel $\mathcal{T}_\mathcal{E}$ which consists of first applying a uniformly random symmetric Majorana monomial $c(x)\otimes c(x)$ and then applying $U\otimes U$ where $U$ is a random Gaussian unitary drawn from $\mathcal{E}$. In other words
\begin{equation}
\mathcal{T}_\mathcal{E}(\sigma)=\sum_{i}p_i U_i\otimes U_i \left[ \frac{1}{2^{2n}}\sum_{x\in \{0,1\}^{2n}} c(x)\otimes c(x) \sigma c(x)^\dagger\otimes c(x)^\dagger \right]U^\dagger_i\otimes U^{\dagger}_i.
\end{equation}
\begin{lemma}
Let $\rho\in \mathcal{K}$ with $\trace{\rho}=1$ and let $\mathcal{E}$ be an ensemble of Gaussian unitaries which is Majorana mixing. Then
\begin{equation}
\mathcal{T}_\mathcal{E}(\rho)=\frac{\Pi_{\mathcal{K}}}{\dim{\mathcal{K}}}.
\label{eq:T1}
\end{equation}
\label{lem:main}
\end{lemma}
\begin{proof}
The proof is based on Appendix D of \cite{CLLW15}. We use the expansion \eqref{eq:expandrho} and linearity of $\mathcal{T}_\mathcal{E}$. We first show that $\mathcal{T}_\mathcal{E}(c(a)\otimes c(b))=0$ whenever $a\neq b$. Note that since $a\neq b$ we have $a\neq 0$ or $b\neq 0$. Without loss of generality suppose $b\neq 0$. We shall use the following fact
\begin{prop}
For all $a,b\in \{0,1\}^{2n}$ with $a\neq b$ and $b\neq 0$ there exists $z\in\{0,1\}^{2n}$ such that $[c(a),c(z)]=0$ and $\{c(b),c(z)\}=0$.
\label{prop:commute}
\end{prop}
\begin{proof}
The Majorana monomials satisfy 
\[
c(x)c(y)=(-1)^{\sum_{j\neq k} x_j y_k}c(y)c(x) \qquad \quad x,y\in\{0,1\}^{2n}.
\]
The desired bit-string $z$ is therefore a solution to $2$ linear equations in $2n$ variables over $\mathbb{F}_2$, namely $\oplus_{j\neq k}a_j z_k=0$ and $\oplus_{j\neq k}b_j z_k=1$. Equivalently,
\[
\oplus_{k=1}^{2n} (a_k\oplus\alpha)z_k=0 \quad \text{and} \quad \oplus_{k=1}^{2n} (b_k\oplus\beta)z_k=1
\]
where $\alpha=\oplus_k a_k$ and $\beta=\oplus_k b_k$. To show there exists a solution for $z$ we must show that the vectors $v,w\in \mathbb{F}_2^{2n}$ defined by $v_k= a_k\oplus\alpha$  and $w_k= b_k\oplus\beta$ satisfy $v\neq w$ and $w\neq 0$. Note $w\neq 0$ follows from our assumption that $b\neq 0$. Since $a\neq b$ we have $v\neq w$ whenever $\alpha=\beta$. On the other hand if $\alpha\neq \beta$ then by definition we have 
\[
\oplus_{k=1}^{2n} (a_k\oplus \alpha)=\oplus_{k=1}^{2n} a_k\neq \oplus_{k=1}^{2n} b_k=\oplus_{k=1}^{2n} (b_k\oplus \beta).
\]
which shows $v\neq w$.
\end{proof}
 So let $z$ satisfy $[c(z),c(a)]=0$ and $\{c(z),c(b)\}=0$. Then
\begin{align*}
&2\sum_{x\in \{0,1\}^{2n}} c(x)\otimes c(x) \big(c(a)\otimes c(b)\big) c(x)^\dagger\otimes c(x)^\dagger  \\
&=\sum_{x\in \{0,1\}^{2n}} c(x)\otimes c(x) \big(c(a)\otimes c(b)\big) c(x)^\dagger\otimes c(x)^\dagger\\&+\sum_{x\in \{0,1\}^{2n}} c(x)c(z)\otimes c(x)c(z) \big(c(a)\otimes c(b)\big) c(z)^\dagger c(x)^\dagger\otimes c(z)^\dagger c(x)^\dagger\\
&=0.
\end{align*}
Next consider $\mathcal{T}_\mathcal{E}(c(a)\otimes c(a))$. Since
\[
c(x)\otimes c(x) \big(c(a)\otimes c(a)\big) c^\dagger(x)\otimes c^\dagger (x)=c(a)\otimes c(a)
\]
for all $x\in\{0,1\}^{2n}$ we get
\[
\mathcal{T}_{\mathcal{E}}\left(c(a)\otimes c(a)\right)=\sum_{i}p_i U_i\otimes U_i \left[c(a)\otimes c(a)\right]U^\dagger_i\otimes U_i
=T_{|a|}\]
where in the last equality we used the fact that $\mathcal{E}$ is Majorana mixing. Using linearity we arrive at
\begin{equation}
\mathcal{T}_\mathcal{E}(\rho)=\sum_{k=0}^{2n} \sum_{|a|=k} \rho_{aa} T_k=\sum_{k=0}^{2n} F_k T_k \label{eq:T2}
\end{equation}
where $F_k$ are the coefficients defined in Proposition \ref{prop:F}. It remains to show that the right-hand sides of Eqs.~(\ref{eq:T2},\ref{eq:T1}) are equal up to normalization. Inspecting Eq.~(\ref{eq:pik}) we see that monomials $c(x)\otimes c(y)$ with $x\neq y$ do not appear and that monomials $\{c(x)\otimes c(x): |x|=k\}$ corresponding to a given Hamming weight $k$ appear with equal weight, i.e., 
\begin{equation}
\frac{\Pi_\mathcal{K}}{\dim{\mathcal{K}}}=\sum_{k=0}^{2n} g_{k} T_k
\label{eq:pi}
\end{equation}
where $g_k$ are some coefficients. Finally, applying Proposition \ref{prop:F} to the density matrix $\frac{\Pi_\mathcal{K}}{\dim{\mathcal{K}}}$ forces $g_k=F_k$, which completes the proof.
\end{proof}

\subsection{Estimating the norm of a superposition of Gaussian states}
Let
\[
|\psi\rangle=\sum_{a=1}^{\chi}x_a|\phi_a\rangle
\]
where $|\phi_a\rangle\in \mathcal{H}_n$ are Gaussian states. Define a random variable 
\begin{equation}
X=2^{n}|\langle \theta|\psi\rangle|^2.
\label{eq:Xdef}
\end{equation}
Here $\theta$ is a random Gaussian state of the form
\begin{equation}
|\theta\rangle=U|y\rangle
\label{eq:theta}
\end{equation}
where $y\in\{0,1\}^{n}$ is chosen uniformly at random and $U$ is a random Gaussian unitary chosen from a Majorana mixing ensemble $\mathcal{E}=\{p_i,U_i\}$. For concreteness in the following we will assume it is taken from the ensemble $\mathcal{E}^\star$ defined in the previous section. Since Eq.~\eqref{eq:Xdef} is insensitive to the global phase of $|\theta\rangle$, we obtain the same random variable $X$ by choosing
\begin{equation}
|\theta\rangle=Uc(x)|0^n\rangle
\label{eq:theta2}
\end{equation}
where $x\in \{0,1\}^{2n}$ is chosen uniformly at random and $U$ is a random Gaussian unitary chosen from $\mathcal{E}^{\star}$.

The expected value of $X$ is 
\[
\expect{X}=\langle \psi|\sum_{i}p_i U_i \bigg(\sum_{y\in\{0,1\}^{n}}|y\rangle\langle y|\bigg) U_i^\dagger|\psi\rangle=\langle \psi|\psi\rangle.
\]
We now use the results of the previous section to upper bound the variance of $X$.
\begin{lemma}
$\mathrm{Var}(X)\leq 2\sqrt{n}\|\psi\|^4$.
\end{lemma}
\begin{proof}
Using the definition of $X$ from Eqs.~(\ref{eq:Xdef},\ref{eq:theta2}) and the definition of $\mathcal{T}_\mathcal{E}$ from the previous section we obtain
\[
\mathrm{Var}(X)\leq \expect{X^2}=4^n \langle \psi\otimes \psi|\mathcal{T}_\mathcal{E}\left(|0^n\otimes 0^n\rangle\langle 0^n\otimes 0^n|\right)|\psi\otimes \psi\rangle.
\]
Applying Lemma \ref{lem:main} gives
\[
\mathrm{Var}(X)\leq \frac{4^n}{\dim{\mathcal{K}}} \langle \psi\otimes \psi|\Pi_\mathcal{K}|\psi\otimes \psi\rangle \leq \frac{4^n}{\dim{\mathcal{K}}}\|\psi\|^4=4^n{2n \choose n}^{-1}\|\psi\|^4,
\]
where in the last line we used Eq.~\eqref{eq:spanL}. Plugging in the lower bound
\[
{2n \choose n}>\frac{4^n}{\sqrt{4n}}
\]
(equation (2.15) in \cite{K08}) completes the proof.
\end{proof}
Now consider an estimator
\[
\xi=\frac{1}{L}\sum_{i=1}^L 2^n|\langle \theta_i|\psi\rangle|^2
\]
where each $\theta_i$ is an independently chosen random state of the form Eq.~\eqref{eq:theta}. Then $\xi$ has expected value $\|\psi\|^2$ and variance
\[
\sigma^2 \leq  2n^{1/2}L^{-1}\|\psi\|^4.
\]
Applying Chebyshev's inequality we obtain
\[
\mathrm{Pr}\left[ |\xi -\|\psi\|^2|\geq \epsilon\|\psi\|^2 \right]\leq \sigma^2 \epsilon^{-2} \|\psi\|^{-4} \leq 2\sqrt{n} L^{-1} \epsilon^{-2}.
\]
Choosing 
\begin{equation}
L=2\sqrt{n} \epsilon^{-2} p_f^{-1}
\label{eq:Lchoice}
\end{equation}
ensures that with probability at least $1-p_f$ we have
\[
(1-\epsilon)\|\psi\|^2 \leq \xi \leq (1+\epsilon)\|\psi\|^2.
\]
Finally, let us describe the algorithm which computes $\xi$. 

The first step is to choose the $L$ random Gaussian states $\theta_i$. To choose each Gaussian state $\theta$ we select a random bit string $y\in \{0,1\}^n$ and a random Gaussian unitary $U$ from the ensemble $\mathcal{E}^{\star}$. Recall from the previous section that $U$ is chosen by selecting a uniformly random $2n\times 2n$ permutation matrix $R$. Letting $M_y$ be the covariance matrix of the basis state $|y\rangle$, the covariance matrix of $\theta$ is $R M_{y} R^{T}$ which can be computed in time $O(n^3)$.

The second step is to compute $\langle \theta_i|\psi\rangle$ for each $i=1,\ldots,L$.  We have
\[
\langle \theta_i|\psi\rangle=\sum_{a=1}^{\chi} \langle \theta_i|\phi_a\rangle
\]
where each term of the sum can be computed in time $O(n^3)$ using the inner product formula from Section \ref{sec:inner}. This step therefore requires time $O(\chi n^3)$ for each $i=1,\ldots,L$. Computing $\xi$ from this data is then straightforward.

The total runtime to compute $\xi$ using this algorithm is 
\[
O(\chi n^3 L)=O(\chi n^{7/2} \epsilon^{-2} p_f^{-1}).
\]

\end{document}